\documentclass[11pt]{article}
\usepackage[english]{babel}
\usepackage{cite}
\usepackage{enumerate}
\usepackage{amssymb,amsmath,amsthm,mathrsfs}

\newtheorem{lemma}{Lemma}
\newtheorem{proposition}{Proposition}
\newtheorem{theorem}{Theorem}
\newtheorem{corollary}{Corollary}

\theoremstyle{definition}
\newtheorem{definition}{Definition}
\newtheorem{example}{Example}

\theoremstyle{remark}
\newtheorem{remark}{Remark}

\newcommand{\due}{\hspace{0.2mm}}
\newcommand{\tre}{\hspace{0.3mm}}
\newcommand{\quattro}{\hspace{0.4mm}}
\newcommand{\cinque}{\hspace{0.5mm}}
\newcommand{\sei}{\hspace{0.6mm}}

\newcommand{\otto}{\hspace{0.8mm}}
\newcommand{\nove}{\hspace{0.9mm}}
\newcommand{\dieci}{\hspace{1mm}}

\newcommand{\sedici}{\hspace{1.6mm}}
\newcommand{\mtre}{\hspace{-0.3mm}}

\newcommand{\msei}{\hspace{-0.6mm}}
\newcommand{\msette}{\hspace{-0.7mm}}
\newcommand{\motto}{\hspace{-0.8mm}}

\newcommand{\fin}{\hspace{0.3mm}}
\newcommand{\spa}{\hspace{-2mm}}

\newcommand{\tr}{\mathrm{tr}}

\newcommand{\ran}{\mathrm{Ran}\hspace{0.3mm}}
\newcommand{\rank}{\mathrm{rank}}

\newcommand{\clco}{\overline{\mathrm{co}}\hspace{0.3mm}}

\newcommand{\tra}{\tr(A)}
\newcommand{\trra}{\tr(rA)}
\newcommand{\trb}{\tr(B)}
\newcommand{\trab}{\tr(A+B)}
\newcommand{\tramo}{\tr(A)^{-1}}
\newcommand{\trramo}{\tr(rA)^{-1}}
\newcommand{\trabmo}{\tr(A+B)^{-1}}

\newcommand{\bile}{\big(}
\newcommand{\biri}{\big)}

\newcommand{\Bile}{\Big(}

\newcommand{\Birich}{\Big)^{^{\mbox{\tiny $\hspace{-1.1mm}\smile$}}}\hspace{-1.3mm}}
\newcommand{\Birigr}{\Big)_{\hspace{-0.6mm}g}}
\newcommand{\Birigrmo}{\Big)_{\hspace{-0.6mm}\gmo}}
\newcommand{\bbile}{\bigg(}
\newcommand{\bbiri}{\bigg)}

\newcommand{\equa}{\spa & = & \spa}
\newcommand{\equb}{\nonumber \\ & = & \spa}
\newcommand{\equx}{\nonumber \\ & \times & \spa}
\newcommand{\lequa}{\spa & \le & \spa}
\newcommand{\lequb}{\nonumber \\ & \le & \spa}
\newcommand{\equdef}{\spa & \defi & \spa}

\newcommand{\defi}{\mathrel{\mathop:}=}
\newcommand{\ifed}{=\mathrel{\mathop:}}

\newcommand{\ima}{\mathrm{i}}
\newcommand{\ccc}{\mathbb{C}}
\newcommand{\de}{\mathrm{d}\hspace{0.2mm}}
\newcommand{\eee}{\mathrm{e}}

\newcommand{\scap}{\langle\hspace{0.5mm}\cdot\hspace{0.6mm},\cdot\hspace{0.5mm}\rangle}

\newcommand{\erre}{\mathbb{R}}
\newcommand{\erredn}{\mathbb{R}^{2n}}

\newcommand{\erreps}{\mathbb{R}_{\hspace{0.3mm}\ast}^{\mbox{\tiny $+$}}}
\newcommand{\erres}{\mathbb{R}_\ast^{\phantom{\mbox{\tiny $i$}}}}
\newcommand{\toro}{\mathbb{T}}

\newcommand{\hh}{\mathcal{H}}
\newcommand{\hau}{\mathcal{K}}

\newcommand{\trau}{\tr_{\hau}}

\newcommand{\unih}{\mathcal{U}{(\hh)}}
\newcommand{\unanh}{\overline{\mathcal{U}}{(\hh)}}
\newcommand{\punanh}{\mathrm{P}\overline{\mathcal{U}}{(\hh)}}

\newcommand{\trc}{\mathcal{B}_1(\hh)}
\newcommand{\trcsa}{\mathcal{B}_1(\hh)_{\mbox{\tiny $\mathbb{R}$}}}
\newcommand{\trcp}{\mathcal{B}_1(\hh)_{\mbox{\tiny $+$}}}
\newcommand{\sta}{\mathcal{S}(\hh)}
\newcommand{\pusta}{\mathcal{P}(\hh)}

\newcommand{\bopsa}{\mathcal{B}(\mathcal{H})_{\mbox{\tiny $\mathbb{R}$}}}
\newcommand{\bop}{\mathcal{B}(\hh)}
\newcommand{\bopp}{\mathcal{B}(\hh)_{\mbox{\tiny $+$}}}
\newcommand{\hs}{\mathcal{B}_2(\hh)}

\newcommand{\ltrc}{\mathscr{L}_1(\hh)}
\newcommand{\ltrcsa}{\mathscr{L}_1(\hh)_{\mbox{\tiny $\mathbb{R}$}}}
\newcommand{\spitrc}{\mathscr{S\hspace{-1.1mm}I}_{\hspace{-0.7mm}1}(\hh)_{\mbox{\tiny $+$}}}
\newcommand{\lbop}{\mathscr{L}(\hh)}
\newcommand{\bbop}{\mathcal{B}(\mathcal{B}(\hh))}
\newcommand{\btrc}{\mathcal{B}(\trc)}
\newcommand{\btrcsa}{\mathcal{B}(\trcsa)}
\newcommand{\effe}{\mathcal{E}(\hh)}

\newcommand{\trcha}{\mathcal{B}_1(\hh\otimes\hau)}

\newcommand{\eu}{E_1}
\newcommand{\ek}{E_k}
\newcommand{\en}{E_n}
\newcommand{\ptmu}{\ptm_1}
\newcommand{\ptmk}{\ptm_k}
\newcommand{\ptmn}{\ptm_n}

\newcommand{\babil}{\mathscr{B}\hspace{-0.4mm}\mathscr{L}_1(\hh)}
\newcommand{\babilsa}{\mathscr{B}\hspace{-0.4mm}\mathscr{L}_1(\hh)_{\mbox{\tiny $\mathbb{R}$}}}
\newcommand{\pot}{\mathscr{P}\hspace{-0.7mm}\mathscr{T}(\hh)}

\newcommand{\bltrc}{\mathcal{B}(\trc,\ltrc)}
\newcommand{\bbtrc}{\mathcal{B}(\trc,\btrc)}
\newcommand{\bltrcsa}{\mathcal{B}(\trcsa,\ltrcsa)}
\newcommand{\bbtrcsa}{\mathcal{B}(\trcsa,\btrcsa)}

\newcommand{\trn}{\|_1}
\newcommand{\trnor}{\|\cdot\trn}

\newcommand{\nori}{\|_{\infty}}
\newcommand{\normi}{\|\cdot\|_{\infty}}
\newcommand{\noro}{\|_{[1]}}
\newcommand{\normo}{\|\cdot\|_{[1]}}
\newcommand{\norin}{\|_{[\infty]}}
\newcommand{\normin}{\|\cdot\|_{[\infty]}}
\newcommand{\btrn}{\|_{1,1}}
\newcommand{\norb}{\|_{(1)}}
\newcommand{\normb}{\|\cdot\|_{(1)}}
\newcommand{\pp}{\mathsf{p}}

\newcommand{\normp}{\|\cdot\|_{\pp}}

\newcommand{\ptm}{\Phi}
\newcommand{\ptmb}{\Psi}
\newcommand{\ptma}{\Phi^{\ast}}
\newcommand{\ptmre}{\Phi_{\mbox{\tiny $\mathbb{R}$}}}
\newcommand{\mare}{\Psi}
\newcommand{\maco}{\Psi_{\mbox{\tiny $\mathbb{C}$}}}

\newcommand{\ptmt}{\Theta}

\newcommand{\delg}{\delta_g}
\newcommand{\delgmo}{\delta^{g^{-1}}}

\newcommand{\pro}{\boxdot}
\newcommand{\prore}{\boxminus}
\newcommand{\prone}{\boxtimes}
\newcommand{\argo}{(\hspace{0.3mm}\cdot\hspace{0.3mm})}
\newcommand{\bifo}{\argo\pro\argo}
\newcommand{\bifore}{\argo\prore\argo}
\newcommand{\bifone}{\argo\prone\argo}

\newcommand{\tran}{{\mbox{\tiny $\mathsf{T}$}}}
\newcommand{\pronet}{\boxtimes^{\tran}\hspace{-1.3mm}}
\newcommand{\bifonet}{\argo\pronet\argo}

\newcommand{\suse}{\mathscr{S}}

\newcommand{\spr}{\odot}
\newcommand{\sprod}{\argo\spr\argo}

\newcommand{\ba}{\check{A}}

\newcommand{\au}{A_{\mbox{\tiny $\Re$}}}
\newcommand{\ad}{A_{\mbox{\tiny $\Im$}}}
\newcommand{\ap}{A_{\mbox{\tiny $+$}}}
\newcommand{\am}{A_{\mbox{\tiny $-$}}}

\newcommand{\bb}{\check{B}}

\newcommand{\bu}{B_{\mbox{\tiny $\Re$}}}
\newcommand{\bd}{B_{\mbox{\tiny $\Im$}}}
\newcommand{\bp}{B_{\hspace{-0.3mm}\mbox{\tiny $+$}}}
\newcommand{\bm}{B_{\hspace{-0.3mm}\mbox{\tiny $-$}}}

\newcommand{\maap}{\mare(A)_{\mbox{\tiny $+$}}}
\newcommand{\maam}{\mare(A)_{\mbox{\tiny $-$}}}

\newcommand{\aup}{A_{\mbox{\tiny $\Re$};\mbox{\tiny $+$}}}
\newcommand{\adp}{A_{\mbox{\tiny $\Im$};\mbox{\tiny $+$}}}
\newcommand{\aum}{A_{\mbox{\tiny $\Re$};\mbox{\tiny $-$}}}
\newcommand{\adm}{A_{\mbox{\tiny $\Im$};\mbox{\tiny $-$}}}

\newcommand{\baup}{\check{A}_{\mbox{\tiny $\Re$};\mbox{\tiny $+$}}}
\newcommand{\badp}{\check{A}_{\mbox{\tiny $\Im$};\mbox{\tiny $+$}}}
\newcommand{\baum}{\check{A}_{\mbox{\tiny $\Re$};\mbox{\tiny $-$}}}
\newcommand{\badm}{\check{A}_{\mbox{\tiny $\Im$};\mbox{\tiny $-$}}}

\newcommand{\aupc}{a_{\mbox{\tiny $\Re$};\mbox{\tiny $+$}}}
\newcommand{\adpc}{a_{\mbox{\tiny $\Im$};\mbox{\tiny $+$}}}
\newcommand{\aumc}{a_{\mbox{\tiny $\Re$};\mbox{\tiny $-$}}}
\newcommand{\admc}{a_{\mbox{\tiny $\Im$};\mbox{\tiny $-$}}}

\newcommand{\bbup}{\check{B}_{\mbox{\tiny $\Re$};\mbox{\tiny $+$}}}

\newcommand{\bbum}{\check{B}_{\mbox{\tiny $\Re$};\mbox{\tiny $-$}}}
\newcommand{\bbdm}{\check{B}_{\mbox{\tiny $\Im$};\mbox{\tiny $-$}}}

\newcommand{\bupc}{b_{\mbox{\tiny $\Re$};\mbox{\tiny $+$}}}

\newcommand{\bumc}{b_{\mbox{\tiny $\Re$};\mbox{\tiny $-$}}}
\newcommand{\bdmc}{b_{\mbox{\tiny $\Im$};\mbox{\tiny $-$}}}

\newcommand{\nuu}{\nu_{\mbox{\tiny $\Re$}}}
\newcommand{\nud}{\nu_{\mbox{\tiny $\Im$}}}
\newcommand{\nup}{\nu_{\mbox{\tiny $\Re$};\mbox{\tiny $+$}}}
\newcommand{\ndp}{\nu_{\mbox{\tiny $\Im$};\mbox{\tiny $+$}}}
\newcommand{\num}{\nu_{\mbox{\tiny $\Re$};\mbox{\tiny $-$}}}
\newcommand{\ndm}{\nu_{\mbox{\tiny $\Im$};\mbox{\tiny $-$}}}

\newcommand{\lmap}{\mathsf{L}}
\newcommand{\rmap}{\mathsf{R}}
\newcommand{\lmapre}{\mathsf{L}_{\mbox{\tiny $\mathbb{R}$}}}
\newcommand{\rmapre}{\mathsf{R}_{\mbox{\tiny $\mathbb{R}$}}}
\newcommand{\lmapasa}{\lmap(A)_{\mbox{\tiny $\mathbb{R}$}}}

\newcommand{\convset}{\mathcal{S}_0}

\newcommand{\lspe}{\mathcal{L}_{\spr}}
\newcommand{\rspe}{\mathcal{R}_{\spr}}
\newcommand{\lpur}{\widehat{\mathcal{L}}_{\spr}}
\newcommand{\rpur}{\widehat{\mathcal{R}}_{\spr}}

\newcommand{\nume}{\mathtt{N}}
\newcommand{\nat}{\mathbb{N}}

\newcommand{\tk}{T_k^{\phantom{\ast}}}
\newcommand{\tl}{T_l^{\phantom{\ast}}}
\newcommand{\tka}{T_k^\ast}

\newcommand{\tu}{T_1^{\phantom{\ast}}}
\newcommand{\td}{T_2^{\phantom{\ast}}}
\newcommand{\prok}{p_k^{\phantom{\ast}}}

\newcommand{\trd}{\mathsf{d}_1}
\newcommand{\trdd}{\mathsf{d}_{1,1}}

\newcommand{\dom}{\mathrm{dom}\hspace{0.2mm}}
\newcommand{\fr}{\mathscr{F}(\hh)}

\newcommand{\lima}{\mathfrak{P}}

\newcommand{\limacl}{\bar{\lima}}

\newcommand{\astak}{A^{\star_{k}}}

\newcommand{\unef}{\mathscr{C}(E)}
\newcommand{\zeref}{\mathscr{C}(I-E)}

\newcommand{\sprx}{\overset{\xi}{\odot}}
\newcommand{\sprgx}{\hspace{-0.3mm}\overset{g[\xi]}{\odot}\hspace{-0.4mm}}
\newcommand{\sprodx}{\argo\overset{\xi}{\odot}\argo}

\newcommand{\sprh}{\overset{h}{\odot}}
\newcommand{\sprodh}{\argo\overset{h}{\odot}\argo}
\newcommand{\sprode}{\argo\overset{e}{\odot}\argo}

\newcommand{\sprgh}{\hspace{-0.1mm}\overset{gh}{\odot}\hspace{-0.4mm}}

\newcommand{\act}{[\hspace{0.5mm}\cdot\hspace{0.5mm}]}

\newcommand{\ure}{\hspace{0.3mm}\mathsf{U}\hspace{0.2mm}}

\newcommand{\mul}{\gamma}
\newcommand{\urep}{{\mathsf{U}\hspace{-0.5mm}\vee\hspace{-0.5mm}\mathsf{U}\hspace{0.2mm}}}

\newcommand{\vre}{\hspace{0.4mm}\mathsf{V}}

\newcommand{\vrep}{{\mathsf{V}\hspace{-0.5mm}\vee\hspace{-0.5mm}\mathsf{V}}}
\newcommand{\commv}{\mathcal{C}(\vre)}

\newcommand{\trace}{\vartheta}
\newcommand{\alge}{\mathcal{A}}
\newcommand{\stalg}{\mathcal{T}}
\newcommand{\stalgp}{\stalg_{\mbox{\tiny $+$}}}
\newcommand{\opro}{\star}
\newcommand{\oprod}{\argo\opro\argo}
\newcommand{\involu}{\mathsf{I}\hspace{0.3mm}}

\newcommand{\fiv}{F}
\newcommand{\fista}{\upsilon}
\newcommand{\fidu}{\hspace{0.3mm}F^\ast\hspace{-0.2mm}}

\newcommand{\fivg}{\fiv_{g}}
\newcommand{\fivgit}{\fiv_{\git}}
\newcommand{\fivggit}{\fiv_{g\git}}
\newcommand{\fivgitg}{\fiv_{\git g}}
\newcommand{\figihmo}{\fiv_{\git h^{-1}}}
\newcommand{\agit}{A_{\git}}
\newcommand{\agh}{A_{gh}}
\newcommand{\agith}{A_{\git h}}
\newcommand{\bh}{B_{h}}
\newcommand{\bgh}{B_{gh}}

\newcommand{\rhog}{\rho_{g}}
\newcommand{\sigg}{\sigma_{\hspace{-0.5mm}g}}
\newcommand{\rhogmo}{\rho_{g^{-1}}}
\newcommand{\siggmo}{\sigma_{\hspace{-0.5mm}g^{-1}}}
\newcommand{\fistag}{\fista_{\hspace{-0.2mm}g}}
\newcommand{\fistagmo}{\fista_{\hspace{-0.2mm}g^{-1}}}

\newcommand{\wrho}{\mathcal{W}_{\hspace{-0.3mm}\rho}}
\newcommand{\wfista}{\mathcal{W}_{\fista}}
\newcommand{\hwfista}{\widehat{\mathcal{W}}_{\fista}}
\newcommand{\wsigma}{\mathcal{W}_{\hspace{-0.2mm}\sigma}}
\newcommand{\wtau}{\mathcal{W}_{\tau}}
\newcommand{\wpsi}{\mathcal{W}_{\hspace{-0.2mm}\psi}}
\newcommand{\hkrho}{\mathcal{Q}_{\rho}}
\newcommand{\hksigma}{\mathcal{Q}_{\sigma}}

\newcommand{\intrdn}{\int_{\mathbb{R}^{2n}}\hspace{-1.5mm}}
\newcommand{\dx}{\de^{2n}\hspace{-0.3mm}x}
\newcommand{\dy}{\de^{2n}\hspace{-0.3mm}y}

\newcommand{\prom}{\varpi}
\newcommand{\hprom}{\widehat{\prom}}
\newcommand{\promg}{\prom^{\hspace{0.1mm}g}}
\newcommand{\promgr}{\prom_{\hspace{-0.3mm}g}}

\newcommand{\promgrmo}{\prom_{\hspace{-0.3mm}g^{-1}}}

\newcommand{\stofg}{\hspace{0.2mm}\overset{\hspace{0.4mm}\fistag}{\odot}\hspace{-0.2mm}}
\newcommand{\stofgm}{\hspace{0.2mm}\underset{\prom}{\overset{\hspace{0.4mm}\fistag}{\odot}\hspace{-0.2mm}}}
\newcommand{\stofmg}{\hspace{-0.3mm}\underset{\hspace{0.6mm}\promg}{\overset{\fista}{\odot}}\hspace{-0.5mm}}
\newcommand{\stofgmom}{\hspace{-2.7mm}\underset{\prom}{\overset{\hspace{3mm}\fistagmo}{\odot}}\hspace{-3mm}}
\newcommand{\stofmgr}{\hspace{-0.3mm}\underset{\hspace{0.6mm}\promgr}{\overset{\fista}{\odot}}\hspace{-0.5mm}}
\newcommand{\stofmgrmo}{\hspace{-3mm}\underset{\hspace{3.4mm}\promgrmo}{\overset{\fista}{\odot}}\hspace{-3.4mm}}
\newcommand{\stofgmg}{\hspace{-0.3mm}\underset{\hspace{0.6mm}\promg}{\overset{\hspace{0.4mm}\fistag}{\odot}}\hspace{-0.5mm}}
\newcommand{\stofgmgrmo}{\hspace{-3mm}\underset{\hspace{3.4mm}\promgrmo}{\overset{\hspace{0.4mm}\fistag}{\odot}}\hspace{-3.4mm}}

\newcommand{\cu}{\mathsf{c}_{\hspace{-0.1mm}\ure}^{\phantom{1}}}
\newcommand{\sqrcu}{\mathsf{c}_{\hspace{-0.1mm}\ure}^{1/2}}
\newcommand{\cumo}{\mathsf{c}_{\hspace{-0.1mm}\ure}^{-1}}

\newcommand{\ddu}{\mathsf{d}_{\hspace{-0.3mm}\ure}^{\phantom{1}}}
\newcommand{\ddumo}{\mathsf{d}_{\hspace{-0.3mm}\ure}^{-1}}
\newcommand{\duf}{D_{\hspace{-0.4mm}\ure}^{\phantom{1}}}

\newcommand{\mesaf}{\nu_{\hspace{-0.7mm}\mbox{\tiny $A\hspace{-0.1mm},\hspace{-0.5mm}\fiv$}\hspace{-0.1mm}}}
\newcommand{\mesru}{\nu_{\hspace{-0.6mm}\mbox{\tiny $\rho\hspace{-0.1mm},\hspace{-0.5mm}\fista$}\hspace{-0.1mm}}}
\newcommand{\mesafnu}{\nu_{\hspace{-0.7mm}\mbox{\tiny $A\hspace{-0.1mm},\hspace{-0.5mm}\fiv$}}\hspace{-0.8mm}\otimes\hspace{-0.6mm}\nu\hspace{0.2mm}}

\newcommand{\conv}{\hspace{-0.3mm}\circledcirc\hspace{-0.5mm}}

\newcommand{\hame}{\mu_G}
\newcommand{\modu}{\Delta_G}

\newcommand{\intG}{\int_G}
\newcommand{\intGG}{\int_{G\times G}}

\newcommand{\ltripr}{[\hspace{-0.6mm}[}
\newcommand{\rtripr}{]\hspace{-0.6mm}]}

\newcommand{\mtripr}{]\hspace{-0.6mm}]_\nu}
\newcommand{\mtriprcc}{]\hspace{-0.6mm}]_{\overline{\nu}}}
\newcommand{\mtriprng}{]\hspace{-0.6mm}]_{\nu^{g}}}

\newcommand{\mtriprngr}{]\hspace{-0.6mm}]_{\nu_{\hspace{-0.3mm}g}}}

\newcommand{\dtripr}{]\hspace{-0.6mm}]_\delta}

\newcommand{\argul}{\hspace{0.6mm}\cdot\hspace{0.6mm}}
\newcommand{\arguc}{\hspace{0.2mm}\cdot\hspace{0.6mm}}
\newcommand{\argur}{\hspace{0.2mm}\cdot\hspace{0.6mm}}
\newcommand{\mtriprod}{[\hspace{-0.6mm}[\argul,\arguc,\argur]\hspace{-0.6mm}]_{\argo}}

\newcommand{\ccme}{\overline{\nu}}

\newcommand{\bilfm}{\hspace{0.2mm}\underset{\nu}{\overset{\fiv}{\boxdot}}}
\newcommand{\bilfmo}{\argo\underset{\nu}{\overset{\fiv}{\boxdot}}\argo}
\newcommand{\bilfo}{\argo\overset{\fiv}{\boxdot}\argo}

\newcommand{\stof}{\hspace{0.2mm}\overset{\fista}{\odot}\hspace{-0.2mm}}
\newcommand{\stofm}{\hspace{0.2mm}\underset{\prom}{\overset{\fista}{\odot}}}
\newcommand{\stofmo}{\argo\underset{\prom}{\overset{\fista}{\odot}}\argo}
\newcommand{\stofo}{\argo\overset{\fista}{\odot}\argo}

\newcommand{\nug}{\nu^{\hspace{0.2mm}g}}
\newcommand{\nugr}{\nu_{\hspace{-0.3mm}g}}

\newcommand{\nuh}{\nu^{\hspace{0.3mm}h}}
\newcommand{\nuhr}{\nu_{h}}
\newcommand{\nugh}{\nu^{\hspace{0.2mm}gh}}
\newcommand{\nughr}{\nu_{\hspace{-0.3mm}gh}}

\newcommand{\come}{\mathcal{M}\hspace{0.3mm}(G)}
\newcommand{\prome}{\mathcal{P}\hspace{-0.5mm}\mathcal{M}\hspace{0.3mm}(G)}
\newcommand{\posme}{\mathcal{M}\hspace{0.3mm}(G)_{\mbox{\tiny $+$}}}

\newcommand{\gmo}{g^{-1}}
\newcommand{\hmo}{h^{-1}}
\newcommand{\git}{\tilde{g}}

\newcommand{\hit}{\tilde{h}}

\newcommand{\dequm}{\mathfrak{D}}
\newcommand{\qum}{\mathfrak{Q}}

\newcommand{\elleuno}{\mathrm{L\hspace{-0.2mm}}^{\mbox{\tiny $1$}}}
\newcommand{\elledue}{\mathrm{L\hspace{-0.2mm}}^{\mbox{\tiny $2$}}}
\newcommand{\ldg}{\elledue(G)}
\newcommand{\elleg}{\elledue(G,\hame;\ccc)}
\newcommand{\lug}{\elleuno(G)}
\newcommand{\phasp}{\mathbb{R}^n\hspace{-0.4mm}\times\mathbb{R}^n}
\newcommand{\lurrn}{\elleuno(\phasp)}

\newcommand{\lrrnco}{\elledue(\phasp,(2\pi)^{-n}\dqdpn\hspace{0.3mm};\mathbb{C})}
\newcommand{\ldrn}{\elledue(\mathbb{R}^n)}

\newcommand{\norluno}{\|_{\mathrm{L\hspace{-0.2mm}}^{1}}}
\newcommand{\bnorluno}{{\Big\|}_{\mathrm{L\hspace{-0.2mm}}^{1}}}
\newcommand{\norldue}{\|_{\mathrm{L\hspace{-0.2mm}}^{2}}}

\newcommand{\chara}{\breve{A}}

\newcommand{\charr}{\breve{\rho}}
\newcommand{\charu}{\breve{\fista}}
\newcommand{\chars}{\breve{\sigma}}

\newcommand{\charpro}{\breve{\prom}}

\newcommand{\dug}{\widehat{G}}

\newcommand{\supp}{\mathrm{supp}\hspace{0.2mm}}

\newcommand{\intrrn}{\int_{\phasp}\hspace{-1.5mm}}

\newcommand{\intrd}{\int_{\mathbb{R}^{2}}\hspace{-1mm}}

\newcommand{\dqdpn}{\de^n\hspace{-0.3mm}q\hspace{0.8mm}\de^n\hspace{-0.1mm}p}
\newcommand{\dqdpt}{\de^n\hspace{-0.3mm}\tilde{q}\hspace{0.8mm}\de^n\hspace{-0.1mm}\tilde{p}}
\newcommand{\dqdptu}{\de\tilde{q}\hspace{0.8mm}\de\tilde{p}}

\newcommand{\lmapr}{\big(\lmap(\rho)\big)}

\newcommand{\mms}{\Omega}
\newcommand{\stomms}{\hspace{0.2mm}\underset{\prom}{\overset{\mms}{\odot}}}
\newcommand{\stoham}{\hspace{0.2mm}\underset{\hspace{0.4mm}\hame}{\overset{\fista}{\odot}}}

\newcommand{\tiq}{\tilde{q}}
\newcommand{\tip}{\tilde{p}}
\newcommand{\qp}{{(q,p)}}
\newcommand{\mqp}{{(-q,-p)}}
\newcommand{\qpmt}{{(q-\tilde{q},p-\tilde{p})}}
\newcommand{\qpt}{{(\tilde{q},\tilde{p})}}
\newcommand{\qpqpt}{(q,p\hspace{0.5mm};\tiq,\tip)}

\newcommand{\eip}{\hspace{0.2mm}\eee^{\ima\hspace{0.3mm}p\cdot\hq}\hspace{0.2mm}}
\newcommand{\eiq}{\hspace{0.2mm}\eee^{-\ima\hspace{0.3mm}q\cdot\hp}\hspace{0.2mm}}
\newcommand{\eipa}{\hspace{0.2mm}\eee^{-\ima\hspace{0.3mm}p\cdot\hq}\hspace{0.2mm}}
\newcommand{\eiqa}{\hspace{0.2mm}\eee^{\ima\hspace{0.3mm}q\cdot\hp}\hspace{0.2mm}}
\newcommand{\teip}{\hspace{0.2mm}\eee^{\ima\hspace{0.3mm}\tip\cdot\hq}\hspace{0.2mm}}
\newcommand{\teiq}{\hspace{0.2mm}\eee^{-\ima\hspace{0.3mm}\tiq\cdot\hp}\hspace{0.2mm}}
\newcommand{\teipa}{\hspace{0.2mm}\eee^{-\ima\hspace{0.3mm}\tip\cdot\hq}\hspace{0.2mm}}
\newcommand{\teiqa}{\hspace{0.2mm}\eee^{\ima\hspace{0.3mm}\tiq\cdot\hp}\hspace{0.2mm}}
\newcommand{\eipt}{\hspace{0.2mm}\eee^{\ima\hspace{0.3mm}(p+\tip)\cdot\hq}\hspace{0.2mm}}
\newcommand{\eiqt}{\hspace{0.2mm}\eee^{-\ima\hspace{0.3mm}(q+\tiq)\cdot\hp}\hspace{0.2mm}}
\newcommand{\eipta}{\hspace{0.2mm}\eee^{-\ima\hspace{0.3mm}(p+\tip)\cdot\hq}\hspace{0.2mm}}
\newcommand{\eiqta}{\hspace{0.2mm}\eee^{\ima\hspace{0.3mm}(q+\tiq)\cdot\hp}\hspace{0.2mm}}

\newcommand{\chk}{\chi_k}
\newcommand{\etk}{\eta_k}
\newcommand{\phl}{\phi_l}
\newcommand{\psl}{\psi_l}
\newcommand{\cek}{|\chi_k\rangle\hspace{-0.4mm}\langle\eta_k|}
\newcommand{\fpl}{|\phi_l\rangle\hspace{-0.4mm}\langle\psi_l|}
\newcommand{\erk}{r_k}
\newcommand{\esl}{s_l}
\newcommand{\akl}{a_{kl}}
\newcommand{\bkl}{b_{kl}}

\newcommand{\povm}{\mathsf{E}_{\fista}}
\newcommand{\pov}{\mathsf{E}}

\newcommand{\twirl}{{\mu[\ure]}}
\newcommand{\twirlnv}{{\nu[\hspace{-0.3mm}\vre]}}
\newcommand{\twirlru}{{\mesru[\ure]}}
\newcommand{\twirlrum}{{(\mesru\hspace{-0.5mm}\conv\prom)[\ure]}}

\newcommand{\twirmenv}{{(\mesaf\conv\nu)[\hspace{-0.3mm}\vre]}}
\newcommand{\twirmesa}{{\mesaf[\hspace{-0.3mm}\vre]}}

\newcommand{\instru}{\mathcal{I}_\bor^{\hspace{0.3mm}\fista\hspace{-0.3mm},\hspace{0.1mm}\sigma}\hspace{-0.5mm}}
\newcommand{\instrar}{\mathcal{I}_{\mbox{\tiny $(\hspace{0.3mm}\cdot\hspace{0.3mm})$}}^{\hspace{0.3mm}\fista\hspace{-0.3mm},\hspace{0.1mm}\sigma}\hspace{-0.5mm}}
\newcommand{\ginstru}{\mathcal{I}_{g\hspace{0.2mm}\bor}^{\hspace{0.3mm}\fista\hspace{-0.3mm},\hspace{0.1mm}\sigma}\hspace{-0.5mm}}
\newcommand{\instrug}{\mathcal{I}_G^{\hspace{0.3mm}\fista\hspace{-0.3mm},\hspace{0.2mm}\sigma}\hspace{-0.5mm}}
\newcommand{\instrur}{\mathcal{I}_G^{\hspace{0.3mm}\fista\hspace{-0.3mm},\hspace{0.3mm}\rho}\hspace{-0.5mm}}
\newcommand{\instrut}{\mathcal{I}_G^{\hspace{0.3mm}\fista\hspace{-0.3mm},\hspace{0.1mm}\tau}\hspace{-0.5mm}}

\newcommand{\parop}{\hspace{0.2mm}\Pi\hspace{0.2mm}}
\newcommand{\paropa}{\hspace{0.2mm}\Pi^{\hspace{0.1mm}\ast}}
\newcommand{\paropmo}{\hspace{0.2mm}\Pi^{-1}}

\newcommand{\Hstar}{\mathrm{H}^\ast\hspace{-0.5mm}}

\newcommand{\bor}{\mathscr{E}}
\newcommand{\cast}{\mathrm{C}^\ast}

\newcommand{\hq}{{\hat{q}}}
\newcommand{\hp}{{\hat{p}}}

\newcommand{\parp}{)_{\mbox{\tiny $+$}}}
\newcommand{\parm}{)_{\mbox{\tiny $-$}}}

\newcommand{\parre}{)_{\mbox{\tiny $\Re$}}}
\newcommand{\parim}{)_{\mbox{\tiny $\Im$}}}

\begin{document}

\title{A class of stochastic products on the convex set of \\ quantum states}

\author{Paolo Aniello$^{1,2}$   \vspace{2mm}
\\
\small \it   $^1$Dipartimento di Fisica ``Ettore Pancini'', Universit\`a di Napoli ``Federico II'',  \\
\small \it   Complesso Universitario di Monte S.\ Angelo, via Cintia, I-80126 Napoli, Italy   \vspace{1mm}
\\
\small \it   $^2$Istituto Nazionale di Fisica Nucleare, Sezione di Napoli,  \\
\small \it   Complesso Universitario di Monte S.\ Angelo, via Cintia, I-80126 Napoli, Italy}

\date{}

\maketitle


\begin{abstract}
\noindent  We introduce the notion of \emph{stochastic product} as a binary operation on the convex set of quantum states
(the density operators) that preserves the convex structure, and we investigate its main consequences. We consider, in particular,
stochastic products that are covariant wrt a symmetry action of a locally compact group. We then construct an interesting class
of group-covariant, associative stochastic products, the so-called \emph{twirled stochastic products}. Every binary operation in
this class is generated by a triple formed by a square integrable projective representation of a locally compact group, by
a probability measure on that group and by a fiducial density operator acting in the carrier Hilbert space of the representation.
The salient properties of such a product are studied. It is argued, in particular, that, extending this binary operation from the
density operators to the whole Banach space of trace class operators, this space becomes a Banach algebra, a so-called
\emph{twirled stochastic algebra}. This algebra is shown to be commutative in the case where the relevant group is abelian.
In particular, the commutative stochastic products generated by the Weyl system are treated in detail. Finally, the physical
interpretation of twirled stochastic products and various interesting connections with the literature are discussed.
\end{abstract}


\section{Introduction}
\label{intro}


Operator algebras and various related structures turn out to play a remarkable role in quantum
mechanics~\cite{Emch,Strocchi,Landsman,Moretti,Stormer}. Moreover, essentially the same algebraic structures
and tools are fundamental in the context of quantum field theory~\cite{Haag}, in quantum statistical
mechanics~\cite{Bratteli} and in non-commutative geometry~\cite{Connes}. In particular, quantum states can be defined
as normalized positive functionals on the $\cast$-algebra of (bounded) observables~\cite{Emch,Strocchi,Moretti,Haag}.

In several contexts, e.g., in quantum information science~\cite{Nielsen,Bengtsson} and in quantum measurement
theory~\cite{Wiseman,Busch,Heino}, one actually restricts to a distinct class of states --- the so-called normal
or completely additive states~\cite{Emch,Moretti} (in the finite-dimensional setting, all states are of this kind) ---
that can be realized as Hilbert space operators: von~Neumann's statistical operators, or \emph{density operators};
namely, as normalized, positive trace class operators. These operators --- that in the following will be simply referred to
as (quantum) \emph{states} --- form a closed convex subset of the complex Banach space of trace class operators,
which, endowed with the ordinary product of operators (i.e., composition) and with the involution $A\mapsto A^\ast$
(the adjoining map), becomes a Banach $\ast$-algebra. This algebra is embedded in the $\cast$-algebra of bounded operators
as a two-sided ideal, and the pairing between states and observables --- that form the selfadjoint part of this $\cast$-algebra ---
is realized by means of the trace functional.

One should not forget, however, that the ordinary product of operators defines the algebraic structure
of quantum \emph{observables}, and \emph{states} are only indirectly involved in this structure as positive
functionals. Actually, the product of two positive operators is, in general, \emph{not} positive itself;
in particular, the product of two density operators is not, in general, a density operator. Endowing the Banach space of
trace class operators with, e.g., the \emph{Jordan} $(A,B)\mapsto (AB+BA)/2$ or with the \emph{Lie}
product $(A,B)\mapsto (AB-BA)/2\ima$, one obtains algebraic structures that preserve selfadjointness;
i.e., the --- Jordan or Lie --- product of two selfadjoint operators is selfadjoint too. But the Jordan and the Lie products are
not associative, and the composition or the Jordan product of two density operators is a quantum state if and only if
the two factors are equal and \emph{pure} --- i.e., trivially, we have the idempotent product of a rank-one projection operator
by itself --- whereas the Lie product of two density operators is \emph{never} a quantum state.

On the one hand, as previously stressed, these facts are not surprising, because all the mentioned products
pertain to the canonical structure of the $\cast$-algebra of quantum observables, whose selfadjoint part
can in fact be regarded as a Jordan-Lie Banach algebra~\cite{Landsman,Alfsen}. On the other hand, it is
quite natural to wonder whether one can endow the Banach space of trace class operators with the structure
of an algebra enjoying the property of being \emph{state-preserving}; namely, such that the product of
two quantum states is again a quantum state. Of course, in order to avoid trivial examples, one should
require that this product be a \emph{genuinely binary} operation; i.e., the binary operation should
effectively involve \emph{both} its arguments. Moreover, it would be desirable to have an \emph{associative}
algebra, whose product be \emph{continuous} wrt some suitable topology. Also, due to the central importance of
symmetry transformations in quantum mechanics~\cite{WignerGT,Bargmann,Kadison,SimonWT,AnielloSW,AnielloSW-bis},
it would be interesting to achieve a state-preserving product enjoying some suitable \emph{covariance} property
wrt a symmetry action of a certain abstract group. Finally, a physical interpretation of such a product would
be in order.

The analogy with classical (statistical) physics is often a powerful guide when dealing with new
aspects of quantum theory. In this regard, it is worth recalling that \emph{classical} states --- namely,
Borel probability measures on phase space (say, $\phasp$) --- are embedded in the Banach algebra of complex (finite)
Borel measures on $\phasp$, where the algebra product is realized by convolution~\cite{AnielloPF,AnielloFPT}.
In particular, the convolution of two probability measures is still a probability measure. Otherwise stated,
convolution is state-preserving: The convolution of two (classical) states is a state too. Notice that this
structure extends immediately to the Banach algebra of complex Radon measures on a locally compact group~\cite{Folland-AA}.
It is then natural to wonder whether a similar structure may be envisaged in the quantum setting as well.

Convolution is a group-theoretical notion; hence, we will consider a general group-theoretical framework
where $\phasp$ --- regarded as the group of translations on phase space --- is replaced by some abstract topological group $G$.
In spite of this great generality, however, one expects that the case of the group of translations on phase space
be a meaningful example. We will actually show that it is possible to define a whole class of associative state-preserving
products on the space of trace class operators, for both finite-dimensional and infinite-dimensional quantum systems.
Each product is generated by a \emph{square integrable} (in general, projective)
representation~\cite{Duflo,AnielloSDP,AnielloPR,Aniello-new,AnielloDMF} of a locally compact group $G$, by a Borel
probability measure $\prom$ on $G$ and by a fiducial state $\fista$, i.e., a density operator acting in the carrier
Hilbert space $\hh$ of the representation. Our construction will yield a state-preserving product on the Banach space
$\trc$ of trace class operators in $\hh$, which, endowed with this binary operation, will be shown to be a Banach algebra.
As in the case of convolution, this algebra turns out to be commutative if the group $G$ is abelian.
We have first briefly sketched this kind of product in~\cite{AnielloDMF}, in the simplest case where $\prom=\delta$
(the Dirac measure at the identity of $G$). We will now consider the general framework, studying the main properties of this
new class of products and providing complete proofs. Moreover, we will work out two remarkable examples: the case where $G$
is compact and the case of the abelian group $\phasp$. In particular, in the latter case, for $\prom=\delta$, we get a
`quantum convolution'. It is worth observing that, unlike the `classical' convolution, the
quantum convolution possesses a degree of freedom; namely, it depends on the choice of a fiducial state.
An interpretation of this fact in terms of Wigner functions and of the associated quantum characteristic functions
will be provided. We will also argue that the quantum convolution is intimately related to Werner's remarkable approach
to quantum harmonic analysis on phase space~\cite{Werner-QHA}.

The first part of the paper will actually be devoted to set the basic rules of the game; i.e., to outline the general algebraic
background of our work. We will introduce the abstract notion of \emph{stochastic product} as a binary operation on the
convex set of quantum states $\sta\subset\trc$ that preserves the convex structure, and we will study its salient properties.
In particular, we will argue that a stochastic product can be extended to a state-preserving bilinear map on $\trc$.
Such an algebra structure on $\trc$ will be called, in the associative case, a \emph{stochastic algebra}. In this context,
our group-theoretical construction of a stochastic product will yield the class of \emph{twirled stochastic algebras}.

The paper is organized as follows. In sect.~{\ref{preliminaries}}, we establish the relevant mathematical background. Next,
in sect.~{\ref{products}}, we introduce the abstract notions of stochastic product and stochastic algebra, and derive
their main consequences. We will then define and study the class of \emph{twirled products} of trace class operators,
that, under suitable assumptions, give rise to stochastic products and algebras, the aforementioned twirled stochastic algebras;
see sect.~{\ref{main}}. In sect.~{\ref{examples}}, we will show by means of examples that twirled stochastic products can
be constructed for every Hilbert space dimension. Eventually, in sect.~{\ref{conclusions}}, a physical interpretation of these
stochastic products will be proposed by discussing some interesting connections with quantum measurement theory, with a
final glance at some possible developments of our work.


\section{Preliminaries: state-preserving linear and bilinear maps}
\label{preliminaries}


In this section, we fix our main terminology and notations, and collect some basic facts and results that will
be fundamental in the rest of the paper.

Let $\hh$ be a separable complex Hilbert space. We will suppose that $\dim(\hh)\ge 2$ (we neglect the trivial
one-dimensional case). We assume the scalar product $\scap$ in $\hh$ to be linear in its \emph{second}
argument, and the symbol $I$ will be adopted for the identity operator. We denote by $\trc$ the
\emph{complex} Banach space of trace class operators in $\hh$, by $\trcsa$ the \emph{real}
Banach space of \emph{selfadjoint} trace class operators, by $\trcp\subset\trcsa$ the convex cone
of all \emph{positive} trace class operators operators and by $\sta\subset\trcsa$ the convex set of all
\emph{density operators} (unit trace, positive trace class operators). $\sta$ will be regarded as the set of
(completely additive) \emph{states} of a quantum system. The extreme points of the convex set $\sta$ form the set $\pusta$
of \emph{pure} states (the rank-one projectors in $\hh$). We denote by $\trnor$ the (trace) norm in both $\trc$ and $\trcsa$.
The Banach space \emph{dual} of $\trc$ is identified  --- via the pairing $\trc\times\bop\ni(A,B)\mapsto\tr(AB)$ ---
with $\bop$, the complex Banach space of all bounded operators in $\hh$, endowed with the standard operator norm $\normi$.
Similarly, the dual space of $\trcsa$ is $\bopsa$, the real Banach space of (bounded) quantum observables.
The symbols $\ltrc\equiv\btrc$, $\lbop\equiv\bbop$ will denote the Banach spaces of bounded linear operators in $\trc$ and $\bop$,
respectively, with norms $\normo$ and $\normin$. Finally, we denote by $\unih$ the unitary group of $\hh$ and by $\unanh$
the unitary-antiunitary group.


\subsection{Some basic facts concerning state-preserving linear maps}

Consider a map $\ptm\colon\trc\rightarrow\trc$. $\ptm$ is called \emph{positive} if $\ptm(\trcp)\subset\trcp$;
it is called \emph{trace-preserving} (alternatively, \emph{adjoint-preserving}) if, for every $A\in\trc$,
$\tr(\ptm(A))=\tr(A)$ (respectively, $\ptm(A^\ast)= \ptm(A)^\ast$). It is clear that, if $\ptm$ is adjoint-preserving,
then it is also \emph{selfadjoint-preserving}; namely, $\ptm(\trcsa)\subset\trcsa$.

If $\|\ptm(A)\trn\le\|A\trn$, for all $A\in\trc$, $\ptm$ is said to be \emph{contractive}. In the case where $\ptm$ is linear,
this condition amounts to requiring that $\ptm\in\ltrc$ and $\|\ptm\noro\le 1$.

We say that $\ptm$ \emph{preserves the set of states} $\sta$ --- in short, that it is \emph{state-preserving}
--- if $\ptm(\sta)\subset\sta$; we say that $\ptm$ is \emph{stochastic} if it is linear and state-preserving.
Clearly, a stochastic map in $\trc$ is the quantum analogue of a stochastic map (or matrix) of classical
probability theory~\cite{Bengtsson}. A (quantum) stochastic map which is also \emph{completely positive}~\cite{Busch,Heino,Holevo}
is often called a \emph{quantum channel} in quantum information science.

\begin{example}
A stochastic map is, in general, neither injective nor surjective. Consider, e.g., the \emph{collapse channel}
$\trc\ni A\mapsto \tr(A)\sei\rho_0$, where $\rho_0\in\sta$ is a fixed density operator. We will call $\rank(\rho_0)$
the \emph{image rank} of the collapse channel. The collapse channel is said to be \emph{pure} if $\rho_0\in\pusta$.
Collapse channels in $\trc$ form a convex subset of $\ltrc$ whose extreme points are the pure collapse channels.
\end{example}

Obviously, if $\ptm$ is positive and trace-preserving, then, in particular, it is state-preserving.
Conversely, in the case where $\ptm$ is linear, $\ptm$ is state-preserving
(if and) only if it is both positive and trace-preserving; i.e., if $\ptm$ is stochastic, then it is a trace-preserving
positive linear map. Indeed, for every $A\in\trc$, we can write
\begin{equation} \label{nota}
A = \au + \ima \sei\ad \fin , \ \ \ \mbox{with $\au,\ad\in\trcsa$} \fin ,
\end{equation}
where $\au$, $\ad$ are uniquely determined, and
\begin{equation} \label{notb}
\au = \aup \msei - \aum = \aupc\tre\baup\msei - \aumc\tre\baum \fin , \ \ \
\ad = \adp \msei - \adm = \adpc\tre\badp \msei - \admc\tre\badm \fin ,
\end{equation}
where:
\begin{equation} \label{notc}
\aup,\ldots,\adm\in\trcp \fin , \ \
\aupc,\ldots,\admc\ge 0 \fin , \ \
\baup,\ldots,\badm\in\sta\cup\{0\} \fin    .
\end{equation}
The four positive operators $\aup,\ldots,\adm$ are uniquely determined if we impose the condition
\begin{equation} \label{orto}
\aup\tre\aum=0=\adp\tre\adm \fin ,
\end{equation}
whose solution is: $2\sei\aup=|\au|+\au,\ldots ,2\sei\adm=|\ad|-\ad$. We then set
\begin{equation} \label{notd}
\aupc\tre\baup = \aup \fin , \ \baup\in\sta \fin , \otto
\mbox{for $\aup\neq 0$, and $\aupc\equiv 0$, $\baup\equiv 0$, otherwise;} \ \ldots
\end{equation}
Notice that, by this definition and by relation~{(\ref{orto})}, $\baup\tre\baum=0=\badp\tre\badm$.

Therefore, if $\ptm\colon\trc\rightarrow\trc$ is linear and state-preserving, then it is positive and
\begin{eqnarray}
\tr(\ptm(A)) \equa
\aupc\tre\tr(\ptm(\baup)) - \aumc\tre\tr(\ptm(\baum)) +
\ima\sei\bile\adpc\tre\tr(\ptm(\badp)) - \admc\tre\tr(\ptm(\badm))\biri
\equb
\aupc\msei-\aumc\msei +\ima\sei\adpc\msei -\ima\sei\admc = \tr(A) \fin , \ \ \
\forall \cinque A\in\trc \fin .
\end{eqnarray}
By an analogous reasoning, a selfadjoint-preserving linear map in $\trc$ is adjoint-preserving.

With the obvious exception of the notion of adjoint-preserving map, analogous definitions and facts
can be established for a map from $\trcsa$ into itself. Clearly, every selfadjoint-preserving map in
$\trc$ can be restricted to $\trcsa$. In particular, if $\ptm\colon\trc\rightarrow\trc$ is a positive
linear map, then, setting
\begin{equation}
\ptmre(A)\defi \ptm(A), \ \ \ \forall \cinque A\in\trcsa \fin ,
\end{equation}
we obtain a positive (real-)linear map $\ptmre\colon\trcsa\rightarrow\trcsa$, because, for every $A\in\trcsa$,
\begin{equation}
A = \ap - \am \fin , \ \ap,\am\in\trcp \ \
\Longrightarrow \ \ \ptmre(A)=\ptm(\ap)-\ptm(\am)\in\trcsa \fin .
\end{equation}
Namely, $\ptm$ is selfadjoint-preserving, and, being linear, also adjoint-preserving: $\ptm(A^\ast)= \ptm(A)^\ast$.
We will call $\ptmre$ the \emph{restriction} of the (selfadjoint-preserving) map $\ptm$ to $\trcsa$. Obviously,
the restriction $\ptmre$ is trace-preserving or contractive if $\ptm$ is. By linearity and by decomposition~{(\ref{nota})},
it is easy to see that, conversely, if $\ptmre$ is trace-preserving, then $\ptm$ is trace-preserving too.

On the other hand, if $\mare\colon\trcsa\rightarrow\trcsa$ is a linear map, then setting
\begin{equation}
\maco(A)\defi \mare(\au) + \ima \sei\mare(\ad) \fin , \ \ \
A\in\trc \fin ,
\end{equation}
we obtain an adjoint-preserving (complex-)linear map $\maco\colon\trc\rightarrow\trc$. Note that
$\maco$ --- the \emph{canonical extension} of the linear map $\mare$ to $\trc$, or, simply,
the \emph{complexification} of $\mare$ --- is the unique (adjoint-preserving) linear map in $\trc$
whose restriction to $\trcsa$ is $\mare$. It is clear that $\maco$ is positive or trace-preserving
if (and only if) $\mare$ is.

\begin{proposition} \label{tecpro}
Let $\ptm\colon\trc\rightarrow\trc$ be a linear map. If $\ptm$ is positive, then it is bounded: $\ptm\in\ltrc$;
if $\ptm$ is also trace-preserving, then it is mildly contractive, i.e., $\|\ptm\noro=1$. If $\ptm$ is
trace-preserving and contractive, then it is positive {\rm (thus, actually, mildly contractive: $\|\ptm\noro=1$)}.
Analogous results hold for a linear map from $\trcsa$ into itself.
\end{proposition}

\begin{proof}
It is known that every positive linear map in $\trc$ is bounded; see, e.g., Proposition~{6.2} of~\cite{Busch}.
Let $\ptm$ be such a map and let $\ptma\colon\bop\rightarrow\bop$ be its Banach space adjoint
(wrt the pairing $\trc\times\bop\ni(A,B)\mapsto\tr(A\tre B)$), which is a positive map too. Moreover, if $\ptm$ is
also trace-preserving, then $\ptma$ is unital; i.e., $\ptma(I)=I$. Therefore, if $\ptm$ is positive and trace-preserving,
then, by the Russo-Dye theorem (Corollary~{2.9} of~\cite{Paulsen}), $\|\ptma\norin = \|\ptma(I)\nori=\|I\nori=1$; which is equivalent to
$\|\ptm\noro = 1$. Now, let $\ptm\colon\trc\rightarrow\trc$ be a trace-preserving bounded linear map. If $\|\ptm\noro\le 1$,
then the unital map $\ptma$ satisfies $\|\ptma\norin\le 1$. Therefore, by Proposition~{2.11} of~\cite{Paulsen}, $\ptma$
--- hence, $\ptm$ --- is positive.

Next, if $\mare\colon\trcsa\rightarrow\trcsa$ is a positive linear map, then its extension $\maco$ to $\trc$ is a
positive linear map too. Therefore, $\maco$ --- hence, $\mare$ --- is bounded. If, in addition, $\mare$ is
trace-preserving, then $\maco$ is trace-preserving as well, so that $\|\maco\noro=1$; hence: $\|\mare\noro\le\|\maco\noro=1$.
On the other hand, for every $A\in\sta$,
\begin{equation}
1=\tr(A)=\tr(\mare(A))=\|\mare(A)\trn\le\|\mare\noro \fin ,
\end{equation}
and thus $\|\mare\noro=1$. Finally, if $\mare\colon\trcsa\rightarrow\trcsa$ is a contractive, trace-preserving linear map,
for every $A\in\trc$, $A\ge 0$, we have:
\begin{eqnarray}
\tr (\maap) + \tr (\maam) \equa
\|\mare(A)\trn
\lequb
\|A\trn
\equb
\tr(A)
\equb
\tr(\mare(A)) = \tr (\maap) - \tr(\maam) \fin .
\end{eqnarray}
Hence, $\maam\msette = 0$ and $\mare$ is a positive map.
\end{proof}

\begin{remark}
Notice that a trace-preserving and adjoint-preserving linear map $\ptm\colon\trc\rightarrow\trc$, which is also
contractive on $\trcsa$, is positive. Indeed, its restriction $\ptmre$ to $\trcsa$ is trace-preserving and contractive.
Hence, by Proposition~{\ref{tecpro}}, $\ptmre$ is positive, so that $\ptm$ is positive too.
\end{remark}

By the previous discussion and by Proposition~{\ref{tecpro}}, one immediately concludes that

\begin{proposition} \label{mainpro}
For every linear map $\ptm$ in $\trc$ {\rm (alternatively, in $\trcsa$)} the following facts are equivalent:
\begin{itemize}

\item $\ptm$ is stochastic.

\item $\ptm$ is the complexification {\rm (respectively, the restriction to $\trcsa$)} of a
stochastic map in $\trcsa$ {\rm (respectively, in $\trc$)}.

\item  $\ptm$ is positive and trace-preserving.

\item  $\ptm$ is trace-preserving and contractive.

\item  $\ptm$ is trace-preserving, bounded and mildly contractive $(\|\ptm\noro=1)$.

\end{itemize}
\end{proposition}


\subsection{Adjoint-preserving isometries, symmetries and pureness-preserving maps}

Later on, we will also exploit some useful facts related to Wigner's theorem on symmetry transformations
that will be now discussed.

\begin{theorem} \label{pseudowig}
Let $\ptm$ be an isometric, adjoint-preserving, surjective linear map in $\trc$. Then, $\ptm$ is
either of the form $\ptm(A)=s\sei UA\sei U^\ast$, or of the form $\ptm(A)=s\sei WA^\ast \tre W^\ast$,
where the (constant) factor $s\in\{1,-1\}$, and $U$, $W$ are a unitary and an antiunitary operator in
$\hh$, respectively, that are uniquely defined up to multiplication by a phase factor.
\end{theorem}

\begin{proof}
By a classical result due to Russo~\cite{Russo,Russo-bis} --- also see Theorem~5 and Remark~2 of~\cite{AnielloSW-bis} ---
a surjective linear isometry $\ptm\colon\trc\rightarrow\trc$ is either of the form $\ptm(A)=UA\sei V$, or of the form
$\ptm(A)=WA^\ast \due Z$, where $(U,V)$ and $(W,Z)$ are pairs, respectively, of unitary and antiunitary operators in $\hh$.
In particular, if $\ptm$ is adjoint-preserving, then we must have that $V=\pm U^\ast$ and $Z=\pm W^\ast$. Indeed, if $\psi\in\hh$
were such that $\phi\equiv U\psi\neq \pm V^\ast\psi\equiv\chi$ --- or $\eta\equiv W\psi\neq \pm Z^\ast\psi\equiv\vartheta$ --- then
the rank-one operator $U|\psi\rangle\langle\psi|V= |\phi\rangle\langle\chi|$, or $W|\psi\rangle\langle\psi|Z= |\vartheta\rangle\langle\eta|$,
would not be selfadjoint; in fact, we would have: $\|\phi\|=\|\chi\|$ and $\phi\neq\pm\chi$, or $\|\eta\|=\|\vartheta\|$ and $\eta\neq\pm\vartheta$.
Hence, we actually have: $U\psi = s\sei V^\ast\psi$ --- or $W\psi = s\sei Z^\ast\psi$ --- for all $\psi\in\hh$.
We stress that here the factor $s=\pm 1$ cannot depend on the choice of $\psi$, because the unitary (linear) operators $U$, $V^\ast$
--- or the antiunitary (antilinear) operators $W$, $Z^\ast$ --- generate the same symmetry transformation~\cite{Bargmann,SimonWT,AnielloSW,AnielloSW-bis},
so that $s$ must be a constant factor; see, in particular, Theorem~2 of~\cite{Bargmann}.
Therefore, an isometric, adjoint-preserving, surjective linear map in $\trc$ is --- possibly up to a factor $s=-1$ --- the
bonded linear extension of a symmetry transformation $\pusta\ni P\mapsto UP\sei U^\ast$ ($\pusta\ni P\mapsto WP \sei W^\ast$)
where, by Wigner's theorem~\cite{Bargmann,SimonWT,AnielloSW,AnielloSW-bis}, the unitary operator $U$  (the antiunitary operator $W$)
is uniquely defined up to a phase factor.
\end{proof}

\begin{corollary} \label{varwig}
Let $\ptm$ be an isometric, positive linear map (hence, a stochastic map) in $\trc$. If $\ptm$ is surjective,
then it is a symmetry transformation; i.e., either of the form $\ptm(A)=UA\sei U^\ast$, or of the form
$\ptm(A)=WA^\ast \tre W^\ast$, where $U$, $W$ are a unitary and an antiunitary operator in $\hh$,
respectively, uniquely defined up to a phase factor.
\end{corollary}

\begin{remark}
According to Wigner~\cite{WignerGT}, a \emph{symmetry transformation} is a bijective map on the pure states
preserving the transition probability. In our setting of linear maps in $\trc$, a symmetry transformation
will be simply a map in $\trc$ of the form specified in Corollary~{\ref{varwig}}. This is of course coherent with
Wigner's classical theorem~\cite{WignerGT,Bargmann,Kadison,SimonWT,AnielloSW,AnielloSW-bis}, and with a linear version of this
theorem~\cite{AnielloSW,AnielloSW-bis}, where the assumption of preservation of the transition probability
between pure states becomes immaterial.
\end{remark}

\begin{remark} \label{psli}
By Corollary~{\ref{varwig}}, the group (wrt composition) of all positive, surjective linear isometries in $\trc$
is precisely the group of all symmetry transformations in $\trc$, which is of course isomorphic to the
\emph{projective unitary-antiunitary group} $\punanh=\unanh/\toro$, where $\toro$ is identified with the
subgroup $\{z\tre I\}_{z\in\toro}$ of the group $\unanh$ of all unitary or antiunitary operators in $\hh$.
\end{remark}

\begin{remark}
Let $\ptm$ be a linear isometry in $\trc$. By Proposition~{\ref{tecpro}}, $\ptm$ is positive if (and only if)
it is trace-preserving.
\end{remark}

\begin{remark} \label{remfin}
Clearly, if $\dim(\hh)<\infty$, then a linear isometry in $\trc$ is automatically surjective.
\end{remark}

Observe that, taking into account the previous remark and the contractivity of stochastic maps, from Corollary~{\ref{varwig}}
one immediately derives the following:

\begin{corollary} \label{fincase}
Let $\ptm\colon\trc\rightarrow\trc$ be a stochastic map, with $\dim(\hh)<\infty$. Then, exactly one of the
following two alternatives is realized:
\begin{enumerate}

\item $\ptm$ is a symmetry transformation.

\item There exists a non-positive $A\in\trc$, with $\|A\trn =1$, such that $\|\ptm(A)\trn <1$.

\end{enumerate}
\end{corollary}

\begin{example} \label{tert}
In the case where $\dim(\hh)=\infty$, it is easy to show that the previous result does not hold by
constructing an explicit example of an isometric stochastic map which is not a symmetry transformation.
Indeed, for any $\nume\in\nat$, or for $\nume=\infty$, let $\{\tk\colon\hh\rightarrow\hh\}_{k=1}^{\nume}$ be a
set of linear or antilinear isometries such that
\begin{equation} \label{orthcond}
\ran(\tk)\perp \ran(\tl) \fin , \ \ \mbox{for $k\neq l$} \fin ,
\end{equation}
and let $\{\prok\}_{k=1}^{\nume}$ be a strictly positive probability distribution.
Consider the linear map $\ptm$ in $\trc$ defined by
\begin{equation} \label{stochiso}
\ptm(A) \defi \sum_{k=1}^\nume \prok\sei\tk \astak \tre\tka , \ \ \
\mbox{with: $\astak=A$ for $\tk$ linear, $\astak=A^\ast$ for $\tk$ antilinear.}
\end{equation}
Here, in the case where $\nume=\infty$, the sum converges wrt the norm $\trnor$ of $\trc$. $\ptm$ is an isometric
stochastic map --- as the reader may easily check --- which is a symmetry transformation if and only if $\nume=1$
\emph{and} $\ran(T)=\hh$, with $T\equiv\tu$. Indeed, for $\nume=1$, $\ptm(P)\in\pusta$, for all $P\in\pusta$, but
$\ran(\ptm(P))\subset\ran(T)$, so that $\ptm$ is a symmetry transformation if and only if $\ran(T)=\hh$
($\Leftrightarrow T$ is a unitary or antiunitary). For $\nume\ge 2$ and $P\in\pusta$,
the state $\ptm(P)$ is \emph{not} pure, because, by condition~{(\ref{orthcond})}, $\rank(\ptm(P))=\nume$;
thus, in this case $\ptm$ is not a symmetry transformation. Accordingly, $\ptm$ is surjective
if and only if $\nume=1$ and $\ran(T)=\hh$. In fact --- taking $0\neq\phi\in\ran(T)$ and
$0\neq\psi\in\ran(T)^\perp$, for $\nume=1$ and $\ran(T)\subsetneq\hh$; $0\neq\phi\in\ran(\tu)$ and
$0\neq\psi\in\ran(\td)$, for $\nume\ge 2$ --- we have: $\langle\phi,\ptm(A)\psi\rangle=0$,
for all $A\in\trc$ ($\ker(\tka)=\ran(\tk)^\perp$).
\end{example}

Notice that, in the previous example, the case $\nume=1$ corresponds precisely to those stochastic maps
of the form~{(\ref{stochiso})} with the property of mapping pure states into pure states, and this class
of maps includes, in particular, the symmetry transformations. We will now describe the whole class of
stochastic maps having the mentioned property.

\begin{definition} \label{purepres}
We say that a map $\ptm\colon\trc\rightarrow\trc$ is \emph{pureness-preserving} if $\ptm(\pusta)\subset\pusta$.
\end{definition}

The structure of pureness-preserving linear maps is determined by the following refinement of
Theorem~1 of~\cite{AnielloSW} (in the proofs of these results essentially the same arguments are used):

\begin{theorem} \label{genewt}
Let $\lima$ be a densely defined linear operator in $\trc$ --- with $\dom(\lima)=\fr$, the linear
subspace of finite rank operators --- mapping the set $\pusta$ of pure states into itself.
Then, $\lima$ is closable, and its closure is a bounded operator in $\trc$ of one of the following types:
\begin{enumerate}

\item an isometric stochastic map of the form
\begin{equation} \label{pseucanfor}
\trc\ni A\mapsto T A \sei T^\ast , \ \mbox{or} \ \; \trc\ni A\mapsto S\tre A^\ast \mtre S^\ast ,
\end{equation}
where $T$, $S$ are, respectively, a linear isometry and an antilinear isometry in $\hh$,
uniquely determined up to a phase factor;

\item a collapse channel of the form
\begin{equation} \label{colcol}
\trc\ni A\mapsto\tr(A)\sei P \fin ,
\end{equation}
for some (fixed) pure state $P\in\pusta$.

\end{enumerate}
\end{theorem}

\begin{proof}
By the spectral decomposition of a selfadjoint finite rank operator --- with
each spectral projection regarded as a sum of mutually orthogonal rank-one projections ---
and by the fact that $\pusta$ is mapped by $\lima$ into itself, one concludes that $\lima$ is positive,
trace-preserving and bounded on its dense domain $\fr\subset\trc$, the linear subspace of all finite rank
operators in $\hh$. Thus, $\lima$ is closable and its closure $\limacl$ is a bounded operator defined
on $\dom(\limacl)=\trc$. By the spectral decomposition of a selfadjoint trace class operator (converging
in the trace norm), one argues as above that $\limacl$ is a trace-preserving and positive too.

Therefore, $\limacl$ is, in particular, a positive linear map, mapping the set of all \emph{pure elements}
of the positive cone $\trcp$ of $\trc$ --- the rank-one positive operators~\cite{Davies-book} ---
into itself, namely, a \emph{pure} positive map. Then, by a classical result of Davies (Theorem~{3.1} of~\cite{Davies};
also see Theorem~{3.1}, in chapt.~{2} of~\cite{Davies-book}), $\limacl$ is of one of the following forms:
\begin{enumerate}[\rm (i)]

\item \label{unita}
$\limacl(A) = T A \sei T^\ast$, for some bounded linear operator $T$ in $\hh$,
uniquely determined up to a phase factor;

\item \label{antiunita}
$\limacl(A) = S\tre A^\ast \mtre S^\ast$, for some bounded antilinear operator $S$ in $\hh$,
uniquely determined up to a phase factor;

\item \label{collassa}
$\limacl(A) = \tr(A\tre B)\sei P$, for some positive operator $B\in\bop$ and some pure state $P\in\pusta$.

\end{enumerate}

In the case~{(\ref{unita})}, $T$ must be actually a linear isometry, since $\limacl$ is also trace-preserving,
so that
\begin{equation}
\tr(A)=\tr(\limacl(A))=\tr(A\sei T^\ast\tre T) \fin, \ \forall\cinque A\in\trc \
\Longrightarrow \ T^\ast\tre T=I \fin .
\end{equation}
Here, we have used the fact that $\trc^\ast$ is (isomorphic to) $\bop$. Analogously, in the case~{(\ref{antiunita})},
$\tr(A)=\tr(S\tre A^\ast \mtre S^\ast)=\tr(A^\ast \mtre S^\ast\mtre S)^\ast=\tr(A\sei S^\ast\mtre S)$, and $S$ must
be an antilinear isometry. Finally, in the case~{(\ref{collassa})}, we have: $\tr(A)=\tr(\limacl(A))=\tr(A \tre B)$,
for all $A\in\trc$; hence, $B=I$.
\end{proof}

\begin{corollary} \label{corpur}
Every pureness-preserving stochastic map $\ptm\colon\trc\rightarrow\trc$ is of the form~{(\ref{pseucanfor})}
or~{(\ref{colcol})}.
\end{corollary}

\begin{proof}
A stochastic map $\ptm$ in $\trc$ is bounded, so that, if it is also pureness-preserving, it must coincide with
the closure of a linear operator, with domain $\fr$ (the finite-rank operators), that maps $\pusta$ into itself.
Then, by Theorem~{\ref{genewt}} the statement follows.
\end{proof}

\begin{remark}
Recalling Remark~{\ref{remfin}}, observe that, in the case where ($2\le$) $\dim(\hh)<\infty$, the two possible
forms~{(\ref{pseucanfor})} or~{(\ref{colcol})} of a pureness-preserving stochastic map in $\trc$ realize,
respectively, the two alternative forms of a stochastic map in Corollary~{\ref{fincase}}.
\end{remark}

\begin{remark}
Notice that a pureness-preserving stochastic map preserves the \emph{purity} $\tr(\rho^2)$ of a state $\rho$
if and only if it is of the form~{(\ref{pseucanfor})} --- $\tr(T A\sei T^\ast)=\tr(A\sei T^\ast\tre T)=\tr(A)$
($T$ linear isometry), $\tr(S\tre A^\ast \mtre S^\ast)=\tr(A\sei S^\ast\mtre S)=\tr(A)$ ($S$ antilinear isometry)
--- whereas a map of the form~{(\ref{colcol})} strictly increases the purity of every non-pure state.
Besides, a stochastic map that does not decrease the purity of states must map pure states into pure states.
\end{remark}

By the previous remark, we have:

\begin{corollary} \label{punode}
A stochastic map $\ptm\colon\trc\rightarrow\trc$ is purity-nondecreasing --- $\tr(\ptm(\rho)^2)\ge\tr(\rho^2)$,
$\forall\cinque\rho\in\sta$ --- if and only if it is pureness-preserving.
\end{corollary}


\subsection{State-preserving bilinear maps}

We will now extend some of the previous definitions and results to a binary operation on $\trc$ or on $\trcsa$
(it is often useful to switch back and forth between the two spaces). We start with those cases where this
extension is straightforward.

\begin{definition}
We say that a map $\bifo\colon\trc\times\trc\rightarrow\trc$ is \emph{positive} if
\begin{equation}
A,B\in\trc \fin , \ A,B\ge 0 \ \ \Longrightarrow \ \ A\pro B \ge 0 \fin .
\end{equation}
We say that $\bifo$ \emph{preserves the set of states} $\sta$ --- in short, that it is
\emph{state-preserving} --- if, for every pair of states $\rho,\sigma\in\sta$, $\rho\pro \sigma\in\sta$.
We say that $\bifo$ is \emph{stochastic} if it is bilinear and state-preserving.
Analogous definitions we set for a map from $\trcsa\times\trcsa$ into $\trcsa$.
\end{definition}

\begin{definition} \label{adjpres}
We say that a map $\bifo\colon\trc\times\trc\rightarrow\trc$ is \emph{selfadjoint-preserving}
if $\trcsa\pro\tre\trcsa\subset\trcsa$; we say that it is \emph{adjoint-preserving} if
\begin{equation} \label{noninv}
A^\ast\pro B^\ast = (A\pro B)^\ast \fin , \ \ \ \forall\cinque A,B\in\trc \fin .
\end{equation}
\end{definition}

\begin{remark}
The definition of an adjoint-preserving binary operation on $\trc$ is already less obvious
than the previous ones. Indeed, e.g., $\trc$, endowed with the standard product of operators (composition) and with the
adjoining map $A\mapsto A^\ast$, becomes a Banach $\ast$-algebra (in particular, we have:
$\|A\tre B\trn\le\|A\nori\sei \|B\trn\le \|A\trn \sei \|B\trn$). On the other hand, the adjoining
map is an \emph{involution} --- in particular, $(A\tre B)^\ast=B^\ast A^\ast$ (compare with~{(\ref{noninv})}) ---
hence, the operator product is not adjoint-preserving.
\end{remark}

Clearly, an adjoint-preserving binary operation on $\trc$ is selfadjoint-preserving; moreover:

\begin{proposition}
Every selfadjoint-preserving bilinear map $\bifo\colon\trc\times\trc\rightarrow\trc$ is adjoint-preserving.
In particular, if a bilinear map on $\trc$ is positive, then it is (selfadjoint-preserving, hence) adjoint-preserving.
\end{proposition}

\begin{proof}
If $\bifo$ is bilinear and selfadjoint-preserving, for $A,B\in\trc$, writing $A=\au+\ima\sei\ad$ and $B=\bu+\ima\sei\bd$,
with $\au,\ldots,\bd\in\trcsa$, and noting that $\au\pro\bu,\ldots,\ad\pro\bu\in\trcsa$, we have:
\begin{equation}
A^\ast\pro B^\ast = \au\pro\bu - \ad\pro\bd - \ima\left(\au\pro\bd+\ad\pro\bu\right) =
(A\pro B)^\ast \fin ;
\end{equation}
i.e., $\bifo$ is adjoint-preserving. If $\bifo$ is a positive bilinear map, for every $A=\ap-\am$ and
$B=\bp-\bm$ in $\trcsa$, with $\ap,\am,\bp,\bm\in\trcp$, we see that $A\pro B$ is selfadjoint;
i.e., the map $\bifo$ is selfadjoint-preserving, hence, adjoint-preserving.
\end{proof}

An adjoint-preserving bilinear map --- in particular, a positive bilinear map --- on $\trc$ can be restricted to a
bilinear map on the real Banach space $\trcsa$.
Conversely, given a bilinear map $\bifore\colon\trcsa\times\trcsa\rightarrow\trcsa$,
for $A,B\in\trc$ --- with $A=\au+\ima\ad$, $B=\bu+\ima\bd$, $\au,\ad,\bu,\bd\in\trcsa$ --- one can set
\begin{equation}
A\pro B \defi \au\prore\bu - \ad\prore\bd + \ima\left(\au\prore\bd+\ad\prore\bu\right),
\end{equation}
so obtaining an adjoint-preserving (complex-)bilinear map on $\trc$, which will be called the \emph{complexification}
of the bilinear map $\bifore$ on $\trcsa$.

At this point, it is not immediately clear what a trace-preserving binary operation should be. For the moment,
just notice that the following definition is compatible with the requirement that the binary operation
be \emph{bilinear}.
\begin{definition}
We say that a map $\bifo\colon\trc\times\trc\rightarrow\trc$ is \emph{trace-preserving} if
\begin{equation}
\tr(A\pro B)= \tr(A) \sei \tr(B) \fin , \ \ \ \forall\cinque A,B\in\trc \fin .
\end{equation}
An analogous definition we set for a map from $\trcsa\times\trcsa$ into $\trcsa$.
\end{definition}

For a \emph{bilinear} map $\bifo\colon\trc\times\trc\rightarrow\trc$ and trace class operators $A,B\in\trc$ ---
using notations~{(\ref{nota})--(\ref{notd})} for $A$, and analogous notations
$\bu,\bd$, $\bupc,\ldots,\bdmc$, $\bbup,\ldots,\bbdm$ relative to $B$ --- we have:
\begin{eqnarray}
A\pro B \equa
\au\pro\bu - \ad\pro\bd + \ima\left(\au\pro\bd+\ad\pro\bu\right)
\equb \aupc\tre\bupc\tre(\baup\pro\bbup) + \cdots + \ima \sei \admc\tre\bumc\tre(\badm\pro\bbum)
\end{eqnarray}
and
\begin{equation}
\tr(A\pro B) =  \aupc\tre\bupc\tre\tr(\baup\pro\bbup) + \cdots
+ \ima \sei \admc\tre\bumc\tre\tr(\badm\pro\bbum) \fin .
\end{equation}
Hence:

\begin{proposition} \label{stopotra}
A stochastic map $\bifo\colon\trc\times\trc\rightarrow\trc$ is positive and trace-preserving.
An analogous result holds for a stochastic map from $\trcsa\times\trcsa$ into $\trcsa$.
\end{proposition}

Let us denote by $\babil$, $\babilsa$ the Banach space of bonded bilinear maps on $\trc$ and on $\trcsa$,
respectively, endowed with the norm
\begin{equation}
\|\bifo\norb\defi\sup\{\|A\pro B\trn\colon \|A\trn,\|B\trn\le 1\} \fin .
\end{equation}
Recall that there are two natural Banach space isomorphisms between $\babil$ (or $\babilsa$) and
the Banach space
\begin{equation}
\bltrc\equiv\bbtrc \ \ \ (\bltrcsa\equiv\bbtrcsa)
\end{equation}
that are given by
\begin{equation} \label{stabaiso}
\bifo\mapsto (A\mapsto A\pro\argo) \ \ \mbox{and}
\ \ \bifo\mapsto (A\mapsto \argo\pro A) \fin .
\end{equation}
These isomorphisms justify our use, in the following, of the symbol $\normb$ to denote, as well,
the standard operator norm in $\bltrc$ and in $\bltrcsa$.

\begin{definition}
We say that a map $\bifo\colon\trc\times\trc\rightarrow\trc$ is \emph{contractive}
on a set $\suse\subset\trc\times\trc$ if
\begin{equation}
\left\|A \pro B\right \trn \le \| A \trn \tre \| B \trn , \ \ \ \forall\cinque (A,B)\in\suse \fin .
\end{equation}
We say that $\bifo$ is \emph{mildly contractive} on $\suse$ if it is contractive on $\suse$ and
\begin{equation}
\sup\{\|A\pro B\trn /(\|A\trn \|B\trn)\colon (A,B)\in\suse,\ A,B\neq 0\} =1 \fin .
\end{equation}
We say \emph{tout court} that $\bifo$ is \emph{contractive} (alternatively, \emph{mildly contractive}) if it is contractive
(respectively, mildly contractive) on $\suse=\trc\times\trc$. Analogous definitions we set for a map from $\trcsa\times\trcsa$
to $\trcsa$.
\end{definition}

We have the following analogue of Proposition~{\ref{mainpro}}:

\begin{proposition} \label{mainprobis}
Let $\bifone$ be a bilinear map on $\trc$ {\rm (alternatively, on $\trcsa$)}. Then, the following facts are equivalent:
\begin{description}

\item[{\tt (P1)}]
$\bifone$ is stochastic.

\item[{\tt (P1$^\prime$)}]
The transpose $\bifonet$ of $\bifone$ --- $A\pronet B\defi B\prone A$ --- is stochastic.

\item[{\tt (P2)}]
$\bifone$ is the complexification {\rm (respectively, the restriction to a binary operation on $\trcsa$)} of a
stochastic map $\bifore\colon\trcsa\times\trcsa\rightarrow\trcsa$ {\rm (respectively, of a stochastic map
$\bifo\colon\trc\times\trc\rightarrow\trc$)}.

\item[{\tt (P3)}]
$\bifone$ is positive and trace-preserving.

\item[{\tt (P4)}]
For every $A\in\sta$, the linear map $A\prone\argo$ is trace-preserving,
bounded and satisfies the condition $\|A\prone\argo\noro=1$.

\item[{\tt (P5)}]
$\bifone$ is contractive on $\sta\times\trcsa$ and trace-preserving.

\item[{\tt (P6)}]
$\bifone$ is mildly contractive on $\sta\times\trcsa$ and trace-preserving.

\item[{\tt (P7)}]
$\bifone$ is contractive on $\trcsa\times\trcsa$ and trace-preserving.

\item[{\tt (P8)}]
$\bifone$ is mildly contractive on $\trcsa\times\trcsa$ and trace-preserving.

\end{description}
\end{proposition}

\begin{proof}
The equivalence between~{\tt (P1)} and~{\tt (P1$^\prime$)}, and between~{\tt (P1)} and~{\tt (P2)}, is clear.
By Proposition~{\ref{stopotra}}, {\tt (P1)} and~{\tt (P3)} are equivalent too.

Moreover, {\tt (P3)} implies~{\tt (P4)}. Indeed, if~{\tt (P3)} holds, then, for every $A\in\sta$,
the linear map $A\prone\argo$ is positive and trace-preserving. Hence, by the first assertion of
Proposition~{\ref{tecpro}}, $A\prone\argo$ is bounded and we have that $\|A\prone\argo\noro=1$. Conversely,
let~{\tt (P4)} hold. By linearity in its first argument, the map $\bifone$ is trace-preserving. Moreover, by the second
assertion of Proposition~{\ref{tecpro}}, the linear map $A\prone\argo$ is positive, for every $A\in\sta$;
hence, by linearity in its first argument, $\bifone$ is positive and {\tt (P3)} holds true.

Notice that {\tt (P4)} implies~{\tt (P6)}. Indeed, if {\tt (P4)} holds,  then --- since {\tt (P4)} $\Rightarrow$ {\tt (P3)}
--- $\bifone$ is trace-preserving and positive. Moreover,
\begin{equation}
\sup\{\|A\prone B\trn /\|B\trn\colon A\in\sta \fin , \sei B\in\trcsa\fin ,\sei B\neq 0\} =
\sup\{\|A\prone\argo\noro\colon A\in\sta\} = 1 \fin ;
\end{equation}
i.e., $\bifone$ is mildly contractive on $\sta\times\trcsa$. Here, in the case where $\bifone$ is regarded as a bilinear map on $\trc$,
we are using the fact that, for every $A\in\sta$, $A\prone\argo$ is positive and trace-preserving, so that
\begin{eqnarray}
\|A\prone\argo\noro = 1 = \sup\{\|A\prone B\trn\colon B\in\sta\}  =
\sup\{\|A\prone B\trn /\|B\trn\colon B\in\trcsa\fin ,\sei B\neq 0\} \fin .
\end{eqnarray}

Obviously, {\tt (P6)} implies~{\tt (P5)}. Also, {\tt (P5)} implies~{\tt (P3)}. In fact, if the map $\bifone$ is
trace-preserving and $\|A\prone B\trn\le \|B\trn$, for all $A\in\sta$ and $B\in\trcsa$,
then, for every $A\in\sta$, the linear map $\trcsa\ni B\mapsto A\prone B\in\trcsa$ is trace-preserving and contractive; hence positive, by
the second assertion of Proposition~{\ref{tecpro}}. Thus, $\bifone$ is trace-preserving and, by linearity in its first argument, positive.

{\tt (P3)} implies~{\tt (P7)}. Indeed, if the bilinear map $\bifone$ is positive and trace-preserving --- denoting by
$\bifore$ its restriction to a bilinear map on $\trcsa$, in the case where $\bifone$ is regarded as a binary operation
on $\trc$, or, otherwise, the map $\bifone$ itself --- we have:
\begin{eqnarray}
\left\|A\prore B \right\trn
\equa
\left\|(\ap - \am) \prore (\bp - \bm)\right\trn
\nonumber \\
& \le & \spa
\left\|\ap \prore \bp \right\trn + \left\|\ap \prore \bm \right\trn +
\left\| \am \prore \bp \right\trn + \left\| \am \prore \bm\right\trn
\equb
\tr((\ap + \am) \prore (\bp + \bm))
\equb
\tr(\ap + \am) \sei \tr(\bp + \bm)
\equb
\| A \trn \tre \| B \trn \fin , \ \ \ \forall\cinque A,B\in\trcsa \fin .
\end{eqnarray}
Thus, $\bifone$ is (trace-preserving and) contractive on $\trcsa\times\trcsa$.

Clearly, {\tt (P8)} $\Rightarrow$~{\tt (P7)}. Let us prove that {\tt (P7)} $\Rightarrow$~{\tt (P8)},
as well. If {\tt (P7)} holds, then
\begin{equation} \label{inesa}
\sup\{\|A\prone B\trn /(\|A\trn \|B\trn)\colon (A,B)\in\trcsa\times\trcsa,\ A,B\neq 0\} \le 1 \fin .
\end{equation}
Moreover, {\tt (P7)} $\Rightarrow$~{\tt (P5)} $\Rightarrow$~{\tt (P3)}; hence, $\|A\prone B\trn /(\|A\trn \|B\trn)=1$,
for all $A,B\in\sta$, i.e., inequality~{(\ref{inesa})} is actually saturated, and $\bifone$ is mildly contractive
on $\trcsa\times\trcsa$.

In conclusion, {\tt (P1$^\prime$)} $\Leftrightarrow$ {\tt (P1)} $\Leftrightarrow$ {\tt (P2)},
{\tt (P1)} $\Leftrightarrow$ {\tt (P3)} $\Rightarrow$ {\tt (P4)} $\Rightarrow$ {\tt (P6)}
$\Rightarrow$ {\tt (P5)} $\Rightarrow$ {\tt (P3)} $\Rightarrow$ {\tt (P7)} $\Leftrightarrow$ {\tt (P8)}
and {\tt (P7)} $\Rightarrow$ {\tt (P5)}; thus, the proof is complete.
\end{proof}

By the Principle of Uniform Boundedness~\cite{Conway}, it turns out that the \emph{separate} continuity of a
bilinear map $\bifo$ on $\trc$ or $\trcsa$ is equivalent to its \emph{joint} continuity --- wrt the norm
$\|(A, B)\btrn\defi\max\{\|A\trn, \|B\trn\}$ in $\trc\times\trc$ --- or, equivalently, to its boundedness.
Besides, the boundedness  of $\bifo$ is implied by the property of preserving the set of all states; precisely:

\begin{proposition} \label{boustoc}
Every stochastic map $\bifo\colon\trc\times\trc\rightarrow\trc$ is bounded, contractive on the sets
$\trcsa\times\trc$ and $\trc\times\trcsa$, and its norm satisfies $\|\bifo\norb\le 2$; whereas,
its restriction $\bifore$ to a bilinear map on $\trcsa$ is such that $\|\bifore\norb = 1$.
\end{proposition}

\begin{proof}
Exploiting decomposition~{(\ref{nota})--(\ref{notd})} for $A\in\trc$, and, next, the implication
{\tt (P1)} $\Rightarrow$ {\tt (P4)} in Proposition~{\ref{mainprobis}}, we have:
\begin{eqnarray}
\left\|A \pro B\right \trn \spa & \le & \spa
\left\|\au \pro B\right \trn + \left\|\ad \pro B\right \trn
\nonumber \\
& \le & \spa
\aupc\|\baup \pro B \trn + \cdots + \admc\|\badm \pro B \trn
\nonumber \\
& \le & \spa
(\aupc + \aumc + \adpc + \admc) \sei \| B \trn
\nonumber \\
& = & \spa
(\|\au \trn + \|\ad \trn) \sei \| B \trn \le 2 \sei \| A \trn \tre \| B \trn
\fin , \ \ \ \forall\cinque A,B\in\trc \fin .
\end{eqnarray}
In particular, setting $\ad=0$, we conclude that $\bifo$ is contractive on $\trcsa\times\trc$.
Replacing $\bifo$ with its transpose, we see that $\left\|A \pro B\right \trn\le 1$, for
$(A,B)\in\trc\times\trcsa$. Moreover, by the fact that~{\tt (P1)} implies~{\tt (P8)} in Proposition~{\ref{mainprobis}},
$\bifore$ is mildly contractive.
\end{proof}


\section{Stochastic products and algebras}
\label{products}


In this section, we will introduce the notion of stochastic product and study its main consequences. In particular, it
turns out that an associative stochastic product is always associated with a `stochastic algebra'. We will also consider
some special classes of points in the domain of a stochastic product and a natural property of group-covariance.
As in the previous section, we will suppose that $\dim(\hh)\ge 2$.

\subsection{Definition and basic facts}

We now define a stochastic product as a binary operation on $\sta$ preserving the natural convex structure of this set.
The sets $\sta$ and $\sta\times\sta$, endowed with the distances $\trd(\rho,\sigma) \defi \|\rho-\sigma\trn$ and
$\trdd((\rho,\sigma),(\tau,\upsilon)) \defi \max\{\|\rho-\sigma\trn,\|\tau-\upsilon\trn\}$, respectively,
become metric spaces. It would be natural to require that a stochastic product be continuous wrt the associated
topologies, but, as it will be clear soon, this property is automatically satisfied.

\begin{remark} \label{topologies}
The weak and the strong operator topologies on $\sta$, and the topology induced on $\sta$ by any
Schatten $\pp$-norm $\normp\fin$, $1\le\pp\le\infty$, all coincide; we will call this topology the
\emph{standard topology} on $\sta$. Indeed, applying Theorem~{2.20} of~{\cite{Simon}} with $\pp=1$,
we conclude that, if a sequence in $\sta$ converges wrt the weak operator topology, then it converges
wrt the trace norm topology as well; hence, \emph{a fortiori}, wrt any other of the previously mentioned
topologies. Thus, the topology induced on $\sta\times\sta$ by the distance $\trdd$ coincides with the
product topology associated with the standard topology on $\sta$.
\end{remark}

\begin{definition}
A \emph{stochastic product} on $\sta$ is a map $\sprod\colon\sta\times\sta\rightarrow\sta$ that is \emph{convex-linear}
in both its arguments, i.e.,
\begin{equation}
(\alpha\tre\rho + (1-\alpha)\sigma)\spr(\epsilon\tau + (1-\epsilon)\upsilon)
= \alpha\tre\epsilon\sei\rho\spr\tau + \alpha(1-\epsilon)\sei\rho\spr\upsilon +
(1-\alpha)\epsilon\sei\sigma\spr\tau + (1-\alpha)(1-\epsilon)\sei\sigma\spr\upsilon \fin ,
\end{equation}
for all $\rho,\sigma,\tau,\upsilon\in\sta$ and $\alpha,\epsilon\in [0,1]$.
\end{definition}

A fundamental connection of the previous definition with the state-preserving bilinear maps on $\trc$ is provided by
the following fact:

\begin{proposition} \label{estende}
Every stochastic product can be obtained as a suitable restriction of a unique state-preserving bilinear map; namely,
for every stochastic product $\sprod\colon\sta\times\sta\rightarrow\sta$, there is a uniquely determined
stochastic map $\bifo\colon\trc\times\trc\rightarrow\trc$ --- the so-called \emph{canonical extension} of $\sprod$
--- such that
\begin{equation}
\rho\spr\sigma = \rho\pro\sigma \fin , \ \ \
\forall\cinque\rho,\sigma\in\sta \fin .
\end{equation}
\end{proposition}

\begin{proof}
Let us define a bilinear map $\bifo\colon\trc\times\trc\rightarrow\trc$ by extending
the stochastic product $\sprod$, which will be now regarded as a map from $\sta\times\sta$
to $\trc$, convex-linear in both its arguments. We will extend this map from $\sta\times\sta$ to $\trcp\times\sta$,
then to $\trcsa\times\sta$ and $\trc\times\sta$, and finally to $\trc\times\trc$. We denote by $\erres$ the set
of nonzero real numbers and by $\erreps$ the set of strictly positive real numbers.

Let us first fix some arbitrary state $\rho\in\sta$ and set $0\pro\rho\equiv 0$. For every $A\in\trcp$,
$A\neq 0$, we define
\begin{equation} \label{indefi}
A\pro\rho\defi\tra\sei (\ba \spr \rho) \fin ,
\ \ \ \mbox{where} \ \ba\equiv\tramo A\in\sta \fin .
\end{equation}
Then, for every $r\in\erreps$ and $0\neq A\in\trcp$,
\begin{equation}
(rA)\pro\rho =\trra\sei \bile(\trramo\tre rA)\spr\rho\biri =
r\sei\tra\sei (\ba\spr\rho) = r\tre (A\pro\rho) \fin ,
\end{equation}
and, by the convex-linearity of $\sprod$ in its first argument, for $A,B\in\trcp$,
$A\neq 0\neq B$,
\begin{eqnarray}
(A+B)\pro\rho
\equa
\trab\sei \bile(\trabmo(A+B))\spr\rho\biri
\equb
\trab\sei \bile(\trabmo\tre\tra\sei \ba
+ \trabmo\tre\trb\sei \bb ) \spr\rho\biri
\equb
\tra\sei(\ba\spr\rho) + \trb\sei(\bb\spr\rho)
\equb
A\pro\rho + B\pro\rho \fin .
\end{eqnarray}

Let us further extend the map $\sprod$ to $\trcsa\times\sta$. To this aim, as usual it is convenient to write
\begin{equation}
\trcsa\ni A = \ap - \am , \ \ \ \ap\defi\frac{1}{2}\sei(|A|+A) \fin , \ \ \ \am\defi\frac{1}{2}\sei(|A|-A) \fin ,
\end{equation}
where $\ap,\am\in\trcp$. We then set
\begin{equation}
A\pro\rho\defi \ap\mtre\pro\rho - \am\mtre\pro\rho \fin , \ \ \ \forall\cinque A\in\trcsa \fin .
\end{equation}
By the previous definition and by~{(\ref{indefi})}, for every $r\in\erres$ and $A\in\trcsa$, we have:
\begin{eqnarray}
(r A)\pro\rho \equa
\bigg(\frac{r}{|r|}\sei(|r|\tre\ap-|r|\tre\am)\bigg) \mtre \pro\rho
\equb
\frac{r}{|r|}\tre\bile(|r|\tre\ap)\pro\rho - (|r|\tre\am)\pro\rho\biri
\equb
r\tre(\ap\mtre\pro\rho - \am\mtre\pro\rho) = r\tre(A\pro\rho) \fin ;
\end{eqnarray}
moreover, $(0\tre A)\pro\rho=0\pro\rho\equiv 0=0\sei(A\pro\rho)$. Next, notice that, given $A,B\in\trcsa$, we have:
\begin{equation}
(A+B\parp - (A+B\parm = A + B = (\ap-\am)+(\bp-\bm) \fin .
\end{equation}
Therefore, $(A+B\parp + \am + \bm = (A+B\parm + \ap + \bp$, and from this relation we get
\begin{equation}
(A+B\parp\mtre\pro\rho + \am\mtre\pro\rho + \bm\mtre\pro\rho
= (A+B\parm\mtre\pro\rho + \ap\mtre\pro\rho + \bp\mtre\pro\rho \fin .
\end{equation}
By this equation, we obtain:
\begin{eqnarray}
(A+B)\pro\rho \equdef
(A+B\parp\mtre\pro\rho - (A+B\parm\mtre\pro\rho
\equb
\ap\mtre\pro\rho - \am\mtre\pro\rho + \bp\mtre\pro\rho - \bm\mtre\pro\rho
\equb
A\pro\rho+B\pro\rho \fin .
\end{eqnarray}
We have therefore constructed a (real-)linear map $\argo\pro\rho$ in $\trcsa$.

We then extend the map $\sprod$ to $\trc\times\sta$, by setting, for every $A\in\trc$,
\begin{equation}
A\pro\rho\defi\au\pro\rho+\ima\sei\ad\pro\rho \fin , \ \ \
\au\defi \frac{1}{2}\sei(A+A^\ast)\in\trcsa \fin , \
\ad\defi -\frac{\ima}{2}\sei(A-A^\ast)\in\trcsa \fin .
\end{equation}
Then, for $z=x+\ima\tre y$, $x,y\in\erre$, and $A\in\trc$, we have:
\begin{equation}
(z A)\pro\rho = x (\au\pro\rho) - y (\ad\pro\rho) + \ima\tre( x (\ad\pro\rho) + \tre y (\au\pro\rho))
= z (\au\pro\rho+ \ima\tre\ad\pro\rho) = z (A\pro\rho) \fin .
\end{equation}
Moreover,
$(A+B\parre=\au+\bu$ and $(A+B\parim=\ad +\bd$; hence, $(A+B)\pro\rho=A\pro\rho + B\pro\rho$,
and we get a (complex-)linear map $\argo\pro\rho$ in $\trc$.

Observe that the mapping $\bifo\colon\trc\times\sta\rightarrow\trc$ constructed so far is such that, for every $A\in\trc$,
the map $A\pro\argo\colon\sta\rightarrow\trc$ is convex-linear.

Thus, arguing as above for extending the second argument of $\argo\pro\argo\colon\trc\times\sta\rightarrow\trc$,
we finally obtain a bilinear map from $\trc\times\trc$ to $\trc$, which is of course stochastic.
Notice that, if a stochastic product $\sprod$ is the restriction of a bilinear map $\bifo$, then the previous procedure
--- without the arguments \emph{proving} linearity --- can be regarded as a reconstruction of $\bifo$ \emph{by} linearity;
hence, this bilinear map is uniquely determined.
\end{proof}

\begin{remark} \label{samrea}
By the same kind of reasoning that one uses in the proof of the previous result, we may have called \emph{stochastic}
a convex-linear map in $\sta$, and then argued that such a map is the restriction to $\sta$ of a state-preserving linear
map in $\trc$.
\end{remark}

\begin{corollary} \label{autocont}
Every stochastic product on $\sta$ is jointly continuous.
\end{corollary}

\begin{proof}
By Proposition~{\ref{estende}}, every stochastic product is the restriction of a state-preserving bilinear map, which,
by Proposition~{\ref{boustoc}}, is bounded; thus, (jointly) continuous.
\end{proof}

Let us denote by $\pot\subset\ltrc$ the convex set of all trace-preserving, positive linear maps in $\trc$.
Proposition~{\ref{estende}} also implies the following:

\begin{corollary} \label{lrmaps}
For every stochastic product $\sprod\colon\sta\times\sta\rightarrow\sta$, there exist two bounded linear maps
\begin{equation}
\lmap\colon\trc\rightarrow\ltrc \ \ \mbox{and} \ \ \rmap\colon\trc\rightarrow\ltrc \fin ,
\end{equation}
--- respectively, the \emph{left partial map} and the \emph{right partial map} associated with $\sprod$ --- such that
$\lmap(\sta)\subset\pot\supset\rmap(\sta)$, uniquely determined by
\begin{equation}
\bile\lmap(\rho)\biri(\sigma) = \rho\spr\sigma = \bile\rmap(\sigma)\biri(\rho) \fin ,
\ \ \ \forall\cinque\rho,\sigma\in\sta \fin .
\end{equation}
Moreover, the sets $\lmap(\trcsa)$, $\rmap(\trcsa)$ consist of adjoint-preserving bounded maps, and
\begin{equation} \label{estima}
\|\lmap\norb =\|\rmap\norb =\|\bifo\norb \le 2 \fin ,\ \ \
\|\lmapre\norb = \|\rmapre\norb =\|\bifore\norb\le 1 \fin ,
\end{equation}
where $\bifo$ is the canonical extension of the stochastic product $\sprod$, $\bifore$ is the restriction
of $\bifo$  to a bilinear map on $\trcsa$ and the bounded linear map $\lmapre\in\bltrcsa$ is defined
by $\lmapre\colon\trcsa\ni A\mapsto\lmapasa=A\prore\argo\in\ltrcsa$.
\end{corollary}

\begin{proof}
We only need to observe that the map $\lmap$ is bounded --- $\|\lmap(A)\argo\noro\le\|A\trn\sei\|\bifo\norb$
--- and to justify relations~{(\ref{estima})}. To this aim, let us recall the norm
estimates in Proposition~{\ref{boustoc}}, and the Banach space isomorphisms~{(\ref{stabaiso})} between
$\babil$ (or $\babilsa$) and $\bltrc$ ($\bltrcsa$).
\end{proof}

We will say that a stochastic product $\sprod$ is \emph{left-constant} (alternatively, \emph{right-constant})
on $\convset\subset\sta$ if
\begin{equation}
\rho\spr\tau = \sigma\spr\tau \fin ,\ \ \mbox{(respectively, $\tau\spr\rho = \tau\spr\sigma$),}
\ \ \ \forall\cinque\rho,\sigma\in\convset \fin , \ \forall\quattro\tau\in\sta \fin ;
\end{equation}
i.e., if the left partial map $\lmap$ (the right partial map $\rmap$) is constant on $\convset$.
Clearly, $\sprod$ is left-constant (right-constant) \emph{tout court} when $\convset=\sta$.

\begin{remark}
Taking into account the fact that an element of $\trc$ can be expressed as a linear combination of
(at most) four density operators, it is clear that $\sprod$ is left-constant (right-constant) if and only if
the associated linear map $\lmap\in\bltrc$ ($\rmap$) is of the form $\lmap(A)=\tr(A)\sei\lmap(\rho_0)\equiv\tr(A)\sei\tau_0$
($\rmap(A)=\tr(A)\sei\rmap(\rho_0)$), where $\rho_0\in\sta$ is arbitrary; i.e., a collapse channel.
\end{remark}

A (nonempty) subset of $\sta$ of the form
\begin{equation}
\{\rho\in\sta\colon \ \lmap(\rho)=\ptm\} \ \ (\{\rho\in\sta\colon \ \rmap(\rho)=\ptm\}) \fin ,
\ \ \ \mbox{for some $\ptm\in\pot$} \fin ,
\end{equation}
is called a \emph{left level set} (a \emph{right level set}) for the stochastic product $\sprod$. These level sets are
\emph{convex} subsets of $\sta$. Clearly, two states $\rho$, $\sigma$ belong to the same left (right) level set
if and only if $\rho-\sigma\in\ker(\lmap)$ ($\rho-\sigma\in\ker(\rmap)$); thus: $\ker(\lmap)=\{0\}\Rightarrow\lmap$
is injective on $\sta$. Actually:

\begin{proposition} \label{lmapinj}
The  linear map $\lmap\colon\trc\rightarrow\ltrc$ $(\rmap)$ is injective on $\sta$ if and only if $\ker(\lmap)=\{0\}$
$(\ker(\rmap)=\{0\})$; i.e., if and only if it is injective.
\end{proposition}

\begin{proof}
We need to prove the `only if part'. For $A\in\trc$, using notations~{(\ref{nota})--(\ref{notd})},
we have that
\begin{equation}
0=\lmap(A) = \lmap(\au) + \ima \sei \lmap(\ad) \ \Longrightarrow \ \lmap(\au)=0=\lmap(\ad) \fin ,
\end{equation}
because $\lmap(\au),\lmap(\ad)\in\ltrc$ are adjoint-preserving maps. Moreover:
\begin{equation}
0=\lmap(\au) = \aupc\tre\lmap(\baup) - \aumc\tre\lmap(\baum) \ \Longrightarrow \
\mbox{($\aupc=\aumc$ and)} \ \lmap(\baup)=\lmap(\baum) \fin ,
\end{equation}
because $\lmap(\baup),\lmap(\baum)\in\pot\cup\{0\}$ (since $\baup,\baum\in\sta\cup\{0\}$). Similarly,
$\lmap(\ad)=0$ implies that $\lmap(\badp)=\lmap(\badm)$. If $A\neq 0$, then $\au\neq 0$ and/or $\ad\neq 0$.
Therefore, if $A\in\ker(\lmap)$ ($\Rightarrow\lmap(\au)=0=\lmap(\ad)$), with $A\neq 0$, supposing without
loss of generality that $\au\neq 0$, we have:
\begin{equation}
\lmap(\baup)=\lmap(\baum), \ \mbox{with $\baup\neq 0\neq\baum$ and $\baup\neq\baum$} \fin .
\end{equation}
Indeed, $\au=\aupc\tre\baup - \aumc\tre\baum\neq 0$ cannot be positive or negative (otherwise, $\lmap(\au)$
would be a strictly positive or negative multiple of a trace-preserving, positive linear map) --- thus,
$\baup,\baum\neq 0$ are both density operators --- and we recall that these density operators are mutually
orthogonal, i.e., $\baup\tre\baum=0$. In conclusion, we find: $\baup,\baum\in\sta$, $\baup\neq\baum$ and
$\lmap(\baup)=\lmap(\baum)$. Therefore, if $\ker(\lmap)\neq\{0\}$, then $\lmap$ is not injective on $\sta$.
\end{proof}

\begin{example}
It is clear that there is a one-to-one correspondence between the set $\pot$ and the left-constant
stochastic products on $\sta$, i.e.,
\begin{equation}
\ptm\in\pot \ \ \longleftrightarrow \ \
\sprod\colon\sta\times\sta\ni(\rho,\sigma)\mapsto\ptm(\sigma)\in\sta \fin ,
\end{equation}
and an analogous correspondence holds, of course, for the right-constant products.
However, these products are `trivial', in the sense that they actually `involve'
just one of their arguments. Moreover, they are, in general, non-associative.
The class of stochastic products that are \emph{simultaneously left-constant and
right-constant} consists of all \emph{collapse products}; i.e., the stochastic products
of the form
\begin{equation} \label{colla}
\rho\spr\sigma \defi \omega \fin ,
\ \ \ \forall\cinque\rho,\sigma\in\sta \fin ,
\end{equation}
where $\omega$ is a fixed state in $\sta$.
\end{example}

\begin{example}
There is a special class of left-constant stochastic products. Consider the stochastic
product $\sprod$ defined by
\begin{equation} \label{lefco}
\rho\spr\sigma \defi U \sigma \tre U^\ast \fin ,
\end{equation}
for all $\rho,\sigma\in\sta$, where $U$ is a unitary operator in $\hh$. We call such a product
a \emph{left-constant unitary product} in $\sta$. More generally, given a subset $\convset$ of $\sta$,
we will say that $\sprod$ is \emph{left-constant unitary on} $\convset$ if~{(\ref{lefco})} holds for all
$\rho\in\convset$ and $\sigma\in\sta$, and for some unitary operator $U$. An analogous definition can be given
for a \emph{left-constant antiunitary product}, associated with an antiunitary operator in $\hh$,
and, of course, for the right-constant counterparts of these products.
\end{example}

\begin{remark}
The (constant on $\sta$) left partial map $\lmap$ associated with a left-constant \emph{anti}unitary product
on $\sta$ is of the form
\begin{equation}
\bile\lmap(A)\biri(B) =   \tr(A)\sei U B^\ast\tre  U^\ast \fin ,
\ \ \ \forall\cinque A,B\in\trc \fin ,
\end{equation}
where $U$ is, of course, an antiunitary operator in $\hh$.
\end{remark}

The two previous examples motivate the following:

\begin{definition}
We say that a stochastic product $\sprod\colon\sta\times\sta\rightarrow\sta$ is
\emph{genuinely binary} if it is neither left-constant nor right-constant.
\end{definition}

\begin{example} \label{mainex}
Is is easy to construct a genuinely binary (but, in general, non-associative) stochastic product on $\sta$.
Indeed, for every pair of maps $\ptm,\ptmb\in\pot$, and for every $\alpha\in (0,1)$, let us set
\begin{equation}
\rho\spr\sigma\defi \alpha \sei\ptm(\rho) + (1-\alpha)\tre \ptmb(\sigma) \fin .
\end{equation}
This is a genuinely binary stochastic product if and only if neither $\ptm$ nor $\ptmb$ are constant on $\sta$; i.e.,
of the form $\trc\ni A\mapsto \tr(A)\sei\rho_0$, where $\rho_0$ is some \emph{fixed} state (collapse channel).
The left and right partial maps associated with this product are given by
\begin{equation}
\lmap(A) = \alpha\sei\tr\argo\sei\ptm(A) + (1-\alpha)\tre\tr(A)\sei\ptmb\argo \fin ,
\ \ \ \rmap(A) = \alpha\sei\tr(A)\sei\ptm\argo + (1-\alpha)\tre\tr\argo\sei\ptmb(A) \fin .
\end{equation}
Denoting by $\bifo$ the stochastic bilinear map extending this product, we have that
\begin{eqnarray}
\|A\pro B\trn
\lequa
\alpha\sei |\tr(B)|\sei \|\ptm(A)\trn + (1-\alpha)\sei |\tr(A)|\sei \|\ptmb(B)\trn
\lequb
\bile\alpha\sei \|\ptm\noro + (1-\alpha)\sei \|\ptmb\noro\biri\tre \|A\trn\sei\|B\trn
= \|A\trn\sei\|B\trn \fin ,
\end{eqnarray}
where we have used the inequality $|\tr(A)|\le \tr(|A|)\ifed\|A\trn$ --- see Theorem~{3.1} of~{\cite{Simon}}
--- and the fact that $\ptm,\ptmb\in\pot$ (hence: $\|\ptm\noro=\|\ptmb\noro=1$). Thus, in this case,
the inequality $\|\bifo\norb\le 2$ (Proposition~{\ref{boustoc}}) is not saturated.
\end{example}

\begin{example} \label{povmex}
Let $E\in\bopp$ be an \emph{effect}~\cite{Busch,Heino} --- i.e., $0\le E\le I$ --- and let $\ptm,\ptmb\in\pot$. We set:
\begin{equation} \label{duecase}
\rho\spr\sigma\defi \tr(\rho\tre E) \sei\ptm(\sigma) + \tr(\rho\tre (I-E)) \sei\ptmb(\sigma)
= \tr(\rho\tre E) \sei\ptm(\sigma) + (1-\tr(\rho\tre E)) \sei\ptmb(\sigma)\fin .
\end{equation}
We then obtain a stochastic product that admits a straightforward generalization. Denoting by $\effe$
the convex set of all effects in $\hh$, let $\{\eu, \ldots ,\en\}\subset\effe$ be a discrete POVM~\cite{Busch,Heino}; i.e.,
$\sum_k \ek =I$. Let, moreover, $\ptmu,\ldots ,\ptmn$ be linear maps in $\pot$. A stochastic product in $\sta$
is defined by setting
\begin{equation} \label{ncase}
\rho\spr\sigma\defi \sum_{k=1}^n \tr(\rho\tre \ek) \sei\ptmk(\sigma) \fin .
\end{equation}
\end{example}

\begin{example} \label{partra}
Let $\ptmt\colon\hh\otimes\hh\rightarrow\hh\otimes\hau$ be a trace-preserving positive linear map, where
$\hau$ is an `auxiliary' (separable complex) Hilbert space. Denoting by $\trau\colon\trcha\rightarrow\trc$
the partial trace, we set
\begin{equation} \label{protra}
\rho\spr\sigma\defi \trau(\ptmt(\rho\otimes\sigma)) \fin .
\end{equation}
We then obtain a stochastic product whose left and right partial maps are given by
\begin{equation}
\lmap(A) = \trau(\ptmt(A\otimes\argo)) \fin ,
\ \ \ \rmap(A) = \trau(\ptmt(\argo\otimes A)) \fin .
\end{equation}
\end{example}


\subsection{Stochastic algebras}
\label{stochal}

The notion of stochastic product is related, in a natural way, to a notion of `stochastic algebra'. We will require that
such an algebra be associative.

\begin{definition}
The Banach space $\trc$, endowed with a map $\bifo\colon\trc\times\trc\rightarrow\trc$ that is both stochastic and
an \emph{associative} binary operation, is called a \emph{stochastic algebra}.
\end{definition}

\begin{remark}
Since, by Proposition~{\ref{estende}}, every stochastic product $\sprod$ on $\sta$ is the restriction of a unique stochastic
bilinear map on $\trc$, which is easily seen to be an associative operation if $\sprod$ is associative, one may equivalently
define a stochastic algebra as a Banach space of trace class operators $\trc$ together with an \emph{associative} stochastic
product on $\sta$.
\end{remark}

\begin{remark} \label{banalge}
Restricting the binary operation of a stochastic algebra $(\trc,\bifo)$ to a bilinear map on $\trcsa$, by Proposition~{\ref{boustoc}}
one obtains a (real) \emph{Banach algebra} $(\trcsa,\bifore)$; in particular, $\left\|A\prore B \right\trn \le\| A \trn \tre \| B \trn$.
Taking into account the same proposition, it is clear that on may `renormalize' by a factor $1/2$ the binary operation $\bifo$ in such a way
to obtain a complex Banach algebra; but the renormalized product would not be state-preserving. Besides, as Example~{\ref{mainex}}
shows, a stochastic algebra may well be a (complex) Banach algebra.
\end{remark}


\subsection{Bijective and pureness-preserving points for a stochastic product}

It is natural to distinguish certain classes of points in $\sta$, relatively to a given stochastic product.

\begin{definition}
We say that $\rho\in\sta$ is a \emph{left injective point} (a \emph{right injective point}) for a stochastic
product $\sprod\colon\sta\times\sta\rightarrow\sta$ if $\lmap(\rho)\colon\trc\rightarrow\trc$ ($\rmap(\rho)$)
maps $\sta$ \emph{injectively} into itself. Left (and right) \emph{surjective} and \emph{bijective} points
for $\sprod$ are defined analogously. Denoting by $\lspe$ ($\rspe$) the set of left (right) bijective
points for $\sprod$, we say that this product is \emph{left bijective} (\emph{right bijective}) if
$\lspe=\sta$ ($\rspe=\sta$). A point in $\sta\setminus\lspe$ ($\sta\setminus\rspe$) is called a
\emph{left non-bijective point} (a \emph{right non-bijective point}) for $\sprod$.
\end{definition}

\begin{remark}
The term `left bijective product', which is coherent with the definition of a `left bijective point' for
the product, should not confuse the reader; a stochastic product which is \emph{left} bijective is actually
bijective wrt its \emph{right} entry (say `bijective on the right'), for each point in its \emph{left} entry
(kept fixed).
\end{remark}

\begin{proposition} \label{ulter}
$\rho\in\sta$ is a left (right) injective point for the stochastic product $\sprod$ if and only if
$\ker(\lmap(\rho))=\{0\}$ $(\ker(\rmap(\rho))=\{0\})$; i.e., if and only if the linear map $\lmap(\rho)$
$(\rmap(\rho))$ is injective. Moreover, if $\rho$ is a left (right) surjective point, then $\lmap(\rho)$
$(\rmap(\rho))$ is surjective.
\end{proposition}

\begin{proof}
To prove the first assertion, we use an argument similar to the proof of Proposition~{\ref{lmapinj}}.
Again, we prove the `only if part', the converse implication being obvious. For $A\in\trc$, since $\lmap(\rho)$
is adjoint-preserving, we have that
\begin{equation}
0=\lmapr(A) = \lmapr(\au) + \ima \sei \lmapr(\ad) \ \Longrightarrow \ \lmapr(\au)=0=\lmapr(\ad) \fin .
\end{equation}
Moreover, if $0=\lmapr(\au) = \aupc\tre\lmapr(\baup) - \aumc\tre\lmapr(\baum)$, then ($\aupc=\aumc$ and)
\begin{equation}
\lmapr(\baup)=\lmapr(\baum) \fin ,
\end{equation}
because $\baup,\baum\in\sta\cup\{0\}$ and $\lmap(\rho)\in\pot$. Analogously, $\lmapr(\ad)=0$ implies
that $\lmapr(\badp)=\lmapr(\badm)$. If $A\neq 0$, then $\au\neq 0$ and/or $\ad\neq 0$.
Therefore, if $A\in\ker(\lmap(\rho))$ --- hence, $\lmapr(\au)=0=\lmapr(\ad)$ --- with $A\neq 0$,
supposing that, say, $\au\neq 0$, we have:
\begin{equation}
\lmap(\baup)=\lmap(\baum), \ \mbox{with $\baup\neq 0\neq\baum$ and $\baup\neq\baum$} \fin .
\end{equation}
Indeed, $\au =\aupc\tre\baup - \aumc\tre\baum\neq 0$ cannot be positive or negative (otherwise, $\lmapr(\au)$
would be a strictly positive or negative multiple of a state) --- thus, $\baup,\baum\neq 0$ are both density
operators --- and we recall that these density operators are mutually orthogonal, i.e., $\baup\tre\baum=0$.
In conclusion: $\baup,\baum\in\sta$, $\baup\neq\baum$ and $\lmapr(\baup)=\lmapr(\baum)$.
Therefore, if $\ker(\lmap(\rho))\neq\{0\}$, then $\lmap(\rho)$ is not injective on $\sta$.

The second assertion is an immediate consequence of the fact that every trace class operator is a linear
combination of (at most) four density operators and of the linearity of $\lmap(\rho)$.
\end{proof}

Another natural notion is the following:

\begin{definition}
We say that $\rho$ is a \emph{left (right) pureness-preserving point} for a stochastic product
$\sprod\colon\sta\times\sta\rightarrow\sta$ if $\lmap(\rho)$ ($\rmap(\rho)$) is a pureness-preserving
map (Definition~{\ref{purepres}}). The subset of $\sta$ formed by all such points will be denoted by $\lpur$ ($\rpur$).
If $\lpur=\sta$ ($\rpur=\sta$), the product is called \emph{left (right) pureness-preserving}.
\end{definition}

A characterization of the bijective points is provided by the following result:

\begin{theorem} \label{chasp}
Given $\rho\in\sta$ and a stochastic product $\sprod$ on $\sta$, and denoting by $\lmap$ $(\rmap)$ the left (right)
partial map associated with this product, the following facts are equivalent:
\begin{enumerate}[\rm (i)]

\item \label{lefbi}
$\rho$ is a left (right) bijective point for $\sprod$.

\item \label{suriso}
$\lmap(\rho)$ $(\rmap(\rho))$ is isometric --- i.e., $\|\bile\lmap(\rho)\biri(A)\trn= 1$,
for all $A\in\trc$ --- and surjective.

\item \label{lefpps}
$\rho$ is a left (right) pureness-preserving point for $\sprod$ such that $\bile\lmap(\rho)\biri(\pusta)=\pusta$
$(\bile\rmap(\rho)\biri(\pusta) = \pusta)$.

\item \label{sytra}
$\rho$ is such that
\begin{equation} \label{symtra}
\rho\spr\sigma =  U \sigma \tre U^\ast \fin ,
\ \ (\sigma\spr\rho =  U \sigma \tre U^\ast) \fin ,
\ \ \ \forall\cinque\sigma\in\sta \fin ,
\end{equation}
where $U$ is a unitary or antiunitary operator in $\hh$, uniquely defined up to a phase factor;
equivalently, $\lmap(\rho)$ $(\rmap(\rho))$ is a symmetry transformation in $\trc$.

\end{enumerate}
\end{theorem}

\begin{proof}
We prove the equivalence between each of properties~{(\ref{lefbi})}--{(\ref{lefpps})} and~{(\ref{sytra})}.
Clearly, {(\ref{sytra})} implies ~{(\ref{lefbi})}--{(\ref{lefpps})}. If $\lmap(\rho)$ maps $\sta$ bijectively
onto itself, the map $\sta\ni\sigma\mapsto\bile\lmap(\rho)\biri(\sigma)=\rho\spr\sigma\in\sta$ is a
`Kadison automorphism'~\cite{SimonWT} (i.e., a convex-linear map from $\sta$ onto itself), so that,
by a variant of Wigner's theorem on symmetry transformations~\cite{Kadison,SimonWT,AnielloSW,AnielloSW-bis},
relation~{(\ref{symtra})} is verified (and, by linearity, $\lmap(\rho)$ is a symmetry transformation); therefore,
{(\ref{lefbi})} $\Rightarrow$ {(\ref{sytra})}. Moreover, if the positive linear map $\lmap(\rho)$ is a surjective
isometry, then, by Proposition~{\ref{varwig}}, $\lmap(\rho)$ is a symmetry transformation in $\trc$ and
relation~{(\ref{symtra})} is again satisfied; namely, {(\ref{suriso})} $\Rightarrow$ {(\ref{sytra})}.
Finally, {(\ref{lefpps})} implies~{(\ref{sytra})} because, if $\lmap(\rho)$ maps $\pusta$ onto itself, then,
by a linear version of Wigner's theorem~\cite{AnielloSW,AnielloSW-bis}, $\lmap(\rho)$ is once again
a symmetry transformation.
\end{proof}

By the previous result, $\lspe\subset\lpur$; we will see that, actually, $\lspe\subsetneq\lpur$. Moreover:

\begin{corollary}
If $\dim(\hh)<\infty$, $\rho$ is a left non-bijective (a right non-bijective) point for a stochastic
product $\sprod$ on $\sta$ if and only if there is a non-positive operator $A\in\trc$, with
$\|A\trn =1$, such that $\|\bile\lmap(\rho)\biri(A)\trn <1$ $(\|\bile\rmap(\rho)\biri(A)\trn <1)$.
This result does not hold if $\dim(\hh)=\infty$.
\end{corollary}

\begin{proof}
Taking into account the equivalence of properties~{(\ref{lefbi})} and~{(\ref{sytra})} in Theorem~{\ref{chasp}},
just recall the dichotomy in Corollary~{\ref{fincase}}, for the first assertion. The second one follows from
Example~{\ref{tert}}.
\end{proof}

Let us now consider the pureness-preserving points.

\begin{theorem} \label{chapur}
Let $\rho$ be a left (right) pureness-preserving point for $\sprod$. Then, $\lmap(\rho)$ $(\rmap(\rho))$
is of one of the following types:
\begin{enumerate}

\item an isometric stochastic map of the form $\trc\ni A\mapsto T A \sei T^\ast$, or
$\trc\ni A\mapsto S\tre A^\ast \mtre S^\ast$, where $T$, $S$ are, respectively, a linear isometry
and an antilinear isometry in $\hh$, uniquely determined up to a phase factor;

\item a collapse channel of the form $\trc\ni A\mapsto\tr(A)\sei P$, for some (fixed) pure state $P\in\pusta$.

\end{enumerate}
Therefore, $\lspe\subsetneq\lpur$ $(\rspe\subsetneq\rpur)$, and $\rho\in\lpur$ $(\rho\in\rpur)$ is an
injective point if and only if $\lmap(\rho)$ $(\rmap(\rho))$ is not a collapse channel. Moreover, a
left (right) pureness-preserving point is left (right) bijective if and only if it is left (right) surjective.
\end{theorem}

\begin{proof}
The first assertion follows immediately from the definition of a pureness-preserving point and from
Corollary~{\ref{corpur}}. By the assumption that $\dim(\hh)\ge 2$, a pureness-preserving point $\rho$,
such that $\lmap(\rho)$ is collapse channel, belongs to $\lpur\setminus\lspe\neq\varnothing$, and all
non-injective points in $\lpur$ are of this type. Moreover, if $\rho$ is a left pureness-preserving point,
the stochastic map $\lmap(\rho)$ is bijective if and only if it is of the first kind in the dichotomy of the
first assertion of the theorem \emph{and} surjective; i.e., a symmetry
transformation in $\trc$.
\end{proof}

By the second assertion of the previous theorem, the pureness preserving points fall in two classes.

\begin{definition}
$\rho\in\sta$ is called a \emph{left collapsing point} (a \emph{right collapsing point})
for $\sprod$ if $\lmap(\rho)$ ($\rmap(\rho)$) is a collapse channel. In particular, the
set $\lpur$ ($\rpur$) can be partitioned into the subset of all collapsing  and the subset
of all injective left (right) pureness-preserving points.
\end{definition}

Notice that, in the finite-dimensional case, every linear or antilinear isometry in $\hh$ is a
unitary or an antiunitary operator; thus, in this case, $\lspe$ coincides with the set of all
injective left pureness-preserving points.

Let us now derive a few further consequences of Theorem~{\ref{chapur}}.

\begin{corollary}
A stochastic product $\sprod$ on $\sta$ is left-constant (right-constant) on every convex set
consisting of left (right) pureness-preserving points. Therefore, $\lpur$ $(\rpur)$ is either empty or
partitioned into maximal convex subsets that are left (right) level sets for the stochastic product.
In particular, a left (right) pureness-preserving stochastic product cannot be genuinely binary,
because it must be left-constant (right-constant).
\end{corollary}

\begin{proof}
The statement is trivial if the convex set of pureness-preserving points is a singleton.
Then, let $\rho,\sigma\in\sta$ be two states belonging to a convex set $\convset\subset\sta$
of left pureness-preserving points for $\sprod$, with $\rho\neq\sigma$. Let us first suppose
that these states are both injective points for the stochastic product. Then, for every
$\alpha\in (0,1)$, $\upsilon=\alpha\tre\rho + (1-\alpha)\sigma\in\convset$ and, supposing that
$\upsilon$ is an injective point too, by Theorem~{\ref{chapur}} we have that
\begin{eqnarray}
U \tau \tre U^\ast \equa
\bile\lmap(\alpha\tre\rho + (1-\alpha)\sigma)\biri(\tau)
\equb
\alpha \bile\lmap(\rho)\biri(\tau) + (1-\alpha) \bile\lmap(\sigma)\biri(\tau) =
\alpha \sei V \tau \tre V^\ast + (1-\alpha) \tre W \tau \tre W^\ast \fin ,
\ \ \ \forall\cinque\tau\in\sta\fin ,
\end{eqnarray}
where $U,V,W$ are linear or antilinear isometries in $\hh$, uniquely defined up to a phase factor.
This relation implies that, actually, $U,V,W$ generate the same stochastic map in $\trc$;
i.e., that they coincide, up to an irrelevant phase factor. This is easily shown by taking
$\tau=P$, where $P$ ranges over the pure states $\pusta$. Besides, it is also clear that
$\upsilon$ cannot be a collapsing point, if both $\rho$ and $\sigma$ are not. By an analogous
reasoning, one concludes that the pair $(\rho,\sigma)$ cannot be formed by a collapsing and a
non-collapsing (i.e., an injective pureness-preserving) point. The only other possibility is that
they are both collapsing points, as well as their convex combinations, and they give rise to the same
collapse channel. Hence, every convex subset of $\sta$ formed by left pureness-preserving points is contained
in a left level set of $\sprod$, which is a convex set itself and, then, a maximal convex subset of $\lpur$.
\end{proof}

We can also establish some constraints concerning the injective pureness-preserving points and the collapsing
points.

\begin{corollary} \label{bolera}
If a stochastic product possesses both left and right collapsing points, then all the associated collapse
channels must coincide. If the product possesses both left and right injective pureness-preserving points,
then all these points are states of the same rank. If the product admits both left collapsing points
and right injective pureness-preserving points, or vice versa, then each collapsing point has rank equal
to the image rank of the associated collapse channel.
\end{corollary}

\begin{proof}
Let $(\rho,\sigma)$ be a pair formed by a left and a right collapsing point, respectively.
Then, $\bile\lmap(\rho)\biri(\tau)=\rho_0$ and $\bile\rmap(\sigma)\biri(\upsilon)=\sigma_0$
--- for all $\tau,\upsilon\in\sta$, and for some fixed states $\rho_0$ and $\sigma_0$ ---
so that, actually, $\rho_0=\bile\lmap(\rho)\biri(\sigma) = \rho\spr\sigma = \bile\rmap(\sigma)\biri(\rho)=\sigma_0$.
Next, let $(\rho,\sigma)$ be a pair formed by a left and a right injective pureness-preserving point,
respectively. Then, by Theorem~{\ref{chapur}} we have:
$U \sigma \tre U^\ast=\bile\lmap(\rho)\biri(\sigma)=\rho\spr\sigma=\bile\rmap(\sigma)\biri(\rho)=V \rho \tre V^\ast$,
for some linear or antilinear isometries $U$ and $V$; hence, $\rank(\rho)=\rank(\sigma)$.
Finally, if, say, $\rho$ is a left collapsing point whereas $\sigma$ is a non-collapsing right pureness-preserving point,
then $\rho_0=\bile\lmap(\rho)\biri(\sigma) = \rho\spr\sigma = \bile\rmap(\sigma)\biri(\rho)=V \rho \tre V^\ast$,
with $V$ a linear or antilinear isometry; hence, the image rank of the collapsing channel $\lmap(\rho)$ ($\defi\rank(\rho_0)$)
is equal to the rank of the collapsing point $\rho$. The reader will easily complete this reasoning.
\end{proof}

Finally, something can be also said about the points that are not pureness-preserving.

\begin{proposition}
$\rho$ belongs to $\sta\setminus\lpur$ $(\sta\setminus\rpur)$ if and only if, for some state $\sigma\in\sta$,
$\tr((\rho\spr\sigma)^2)<\tr(\sigma^2)$ $(\tr((\sigma\spr\rho)^2)<\tr(\sigma^2))$.
\end{proposition}

\begin{proof}
The statement follows from Corollary~{\ref{punode}}.
\end{proof}

\begin{example}
We now reconsider the stochastic product~{(\ref{duecase})} in Example~{\ref{povmex}}. Let us assume that $\ptm\neq\ptmb$
and $0\neq E\neq I$. Suppose, moreover, that $\ptm$ is a symmetry transformation, and let  $\unef$ be the (possibly empty)
convex subset of $\sta$ defined by $\unef\defi\{\rho\in\sta\colon \ \tr(\rho\tre E)=1\}$. Notice that
$0\neq E\neq I\Rightarrow\unef,\zeref\subsetneq\sta$. Moreover $\unef\cap\zeref=\varnothing$ and, by the equivalence
of the properties~{(\ref{lefbi})} and~{(\ref{sytra})} in Theorem~{\ref{chasp}}, we have:
\begin{equation*}
\mbox{$\rho$ left bijective point} \ \ \Longleftrightarrow \ \ \mbox{$\rho\in\unef$,
or $\ptmb$ is a symmetry transformation too and $\rho\in\zeref$} \fin .
\end{equation*}
Therefore, the product is left constant --- unitary or antiunitary --- on the convex set $\unef$, which is (either empty or)
a left level set for the stochastic product; moreover, $\unef=\lspe$ if and only if $\ptmb$ is \emph{not} a symmetry transformation
and/or $\zeref=\varnothing$. In the case where $\ptmb$ is a symmetry transformation too (with $\ptmb\neq\ptm$) and
$\unef\neq\varnothing\neq\zeref$, we have that $\lspe=\unef\sqcup\zeref$, with the partition formed by two maximal convex subsets
of $\lspe$. Consider next the more general stochastic product~{(\ref{ncase})}. If the product is right injective, then the discrete
POVM $\{\eu, \ldots ,\en\}$ must be \emph{informationally complete}~\cite{Busch,Heino}; i.e., for every $\rho,\sigma\in\sta$, with
$\rho\neq\sigma$, we must have that $\tr(\rho\tre\ek)\neq\tr(\sigma\tre\ek)$, for some $k\in\{1,\ldots ,n\}$. Hence, by a well known
elementary result --- see Proposition~{3.49} of~\cite{Heino} --- if $\dim(\hh)<\infty$, the constraint $n\ge\dim(\hh)^2$ must be satisfied.
\end{example}


\subsection{Group-covariant stochastic products and equivariant families of products}

Let $G$ be a locally compact, second countable, Hausdorff topological group; in short, a l.c.s.c.\ group,
and let $\ure\colon G \rightarrow\unih$, $\vre\colon G \rightarrow\unih$ pair of projective representations~\cite{Raja} of
this group in a (separable complex) Hilbert space $\hh$. As usual, these representations are supposed to be weakly Borel maps~\cite{Raja}.

Let $\spitrc\simeq\punanh$ be the group of positive, surjective linear isometries in $\trc$ --- see Remark~{\ref{psli}} ---
endowed with the strong operator topology of bounded linear operators in $\trc$. The map
\begin{equation} \label{isorep}
\urep\colon G\rightarrow\spitrc \fin , \ \ \
\urep(g) A\defi\ure(g)\tre A \tre \ure(g)^\ast = \ure(g)\tre A \tre \ure(g)^{-1} ,
\end{equation}
associated with the projective representation $\ure$, is a continuous group homomorphism (see~\cite{AnielloBM}, Proposition~{4.1}).
This map gives rise to a (symmetry) action of the group $G$ on the space of trace class operators where the quantum states live.
We stress that, although $\ure$ is (in general) projective,
\begin{equation}
\ure(gh)= \mul (g,h) \tre \ure(g)\ure(h), \ \ \ \mul (g,h)\in\toro,
\end{equation}
$\urep$ behaves like an ordinary group representation (i.e., as already mentioned, homomorphically):
\begin{equation}
\urep(gh)= \urep(g)\otto\urep(h).
\end{equation}

Let us first introduce a natural notion of covariance for a stochastic product.

\begin{definition}
We say that a stochastic product $\sprod\colon\sta\times\sta\rightarrow\sta$ is \emph{left (right) covariant} wrt the pair
of projective representations $(\ure,\vre)$ if
\begin{equation} \label{procov}
(\urep(g)\tre\rho)\spr\sigma = \vrep(g)\tre(\rho\spr\sigma) \ \ \ (\rho\spr(\urep(g)\tre \sigma) = \vrep(g)\tre(\rho\spr\sigma))
\fin ,
\end{equation}
for all $g\in G$ and all $\rho,\sigma\in\sta$. We say that the product is \emph{left (right) covariant} wrt $\ure$ if~{(\ref{procov})}
holds with $\vre\equiv\ure$. An analogous definition we set for any algebra product (i.e., any bilinear map) on $\trc$.
\end{definition}

Let us deduce a few immediate consequences of covariance of a stochastic product. Observe that,
in the case where $\dim(\hh)=n<\infty$, $\sta$ admits a \emph{maximally mixed state} $\mms\defi n^{-1}I$.

\begin{lemma} \label{stact}
Let the projective representation $\vre\colon G \rightarrow\unih$ be irreducible, and let $\rho\in\sta$.
Then:
\begin{equation}
\vrep(g)\tre\rho =\rho \fin , \ \forall\cinque g\in G \ \Longleftrightarrow \
\dim(\hh)<\infty \ \mbox{and}\ \rho=\mms \fin .
\end{equation}
\end{lemma}

\begin{proof}
Note that $\vrep(g)\tre\rho =\rho$, for all $g\in G$, if and only if $\rho$ belongs to the commutant $\commv$
of $\vre$. In the case where $\vre$ is irreducible, by Schur's lemma (that holds also for genuinely
projective representations), $\commv=\{z\tre I\}_{z\in\ccc}$; hence, in this case, $\rho\in\commv$ if and only if
$\dim(\hh)<\infty$ and $\rho$ is the maximally mixed state.
\end{proof}

\begin{proposition} \label{propamms}
Suppose that $\hh$ is finite-dimensional and $\vre\colon G \rightarrow\unih$ is irreducible, and
let $\sprod\colon\sta\times\sta\rightarrow\sta$ be a stochastic product, left (right) covariant
wrt the pair $(\ure,\vre)$. Then: $\mms\spr\sigma=\mms$ $(\rho\spr\mms=\mms)$.
\end{proposition}

\begin{proof}
Indeed, for every $g\in G$, we have that $\vrep(g)\tre(\mms\spr\sigma)=(\urep(g)\tre\mms)\spr\sigma=\mms\spr\sigma$.
Therefore, by Lemma~{\ref{stact}}, for every $\sigma\in\sta$, $\mms\spr\sigma=\mms$.
\end{proof}

If the stochastic product $\sprod$ is commutative, then all collapsing points are both left and right collapsing.
Then, by the first assertion of Corollary~{\ref{bolera}}, all the associated collapse channels must coincide.
This conclusion is coherent with the following result.

\begin{proposition} \label{propbmms}
Let $\vre\colon G \rightarrow\unih$ be irreducible, and let $\sprod\colon\sta\times\sta\rightarrow\sta$
be a commutative stochastic product, (left and right) covariant wrt the pair $(\ure,\vre)$. If $\rho\in\sta$
is a (left and right) collapsing point, then $\dim(\hh)<\infty$ and
\begin{equation}
\rho\spr\sigma =\mms=\sigma\spr \rho \fin , \ \ \ \forall\cinque\sigma\in\sta \fin .
\end{equation}
\end{proposition}

\begin{proof}
Let $\rho$ be a collapsing point for the commutative stochastic product $\sprod$; i.e.,
$\rho\spr\sigma=\rho_0=\sigma\spr\rho$, for all $\sigma\in\sta$. Since $\sprod$ is both
left and right covariant wrt the pair $(\ure,\vre)$, we have:
$\vrep(g)\tre\rho_0=\vrep(g)\tre(\rho\spr\sigma)=\rho\spr(\urep(g)\tre \sigma)=\rho_0$,
for all $g\in G$. Hence, by Lemma~{\ref{stact}}, for every $\sigma\in\sta$, $\rho_0=\rho\spr\sigma=\mms$.
\end{proof}

In addition to covariance, there are two further group-theoretical notions concerning stochastic products,
but this time involving a \emph{family} of products: invariance and equivariance. Precisely, let $X$ be a
$G$-space~\cite{Raja} wrt to a (left) group action $\act\colon G\times X\ni (g,\xi) \mapsto g[\xi]\in X$.
Suppose that the points of $X$ label a family of stochastic products, all defined on the same set of states $\sta$.

\begin{definition}
We say that the family of stochastic products $\{\sprodx\colon\sta\times\sta\rightarrow\sta\}_{\xi\in X}$ is \emph{invariant}
wrt the $G$-action $\act$ if
\begin{equation}
\rho\sprx\sigma = \rho\sprgx\sigma  \fin , \ \ \ \forall\cinque g \in G \fin ,
\ \forall\cinque \rho,\sigma\in\sta \fin .
\end{equation}
We say that $\{\sprodx\}_{\xi\in X}$ is \emph{left (right) inner equivariant} wrt the pair $(\act,\ure)$ if
\begin{equation}
(\urep(\gmo)\tre\rho)\sprx\sigma = \rho\sprgx\sigma \ \ \ (\rho\sprx(\urep(\gmo)\tre \sigma) = \rho\sprgx\sigma)
\fin , \ \ \ \forall\cinque g \in G \fin , \ \forall\cinque \rho,\sigma\in\sta \fin .
\end{equation}
Moreover, we say that $\{\sprodx\}_{\xi\in X}$ is \emph{outer equivariant} wrt the pair $(\act,\ure)$ if
\begin{equation}
\urep(\gmo)\sei(\rho\sprx\sigma) = \rho\sprgx\sigma \fin , \ \ \
\forall\cinque g \in G \fin , \ \forall\cinque \rho,\sigma\in\sta \fin .
\end{equation}
Analogous definitions we set for any algebra product in $\trc$.
\end{definition}

\begin{example}
One can easily construct examples of equivariant families of stochastic products from a given stochastic
product $\sprod\colon\sta\times\sta\rightarrow\sta$ and a projective representation $\ure\colon G \rightarrow\unih$.
Indeed, consider the family of stochastic products defined by
\begin{equation}
\{\sprodh\colon\sta\times\sta\rightarrow\sta\}_{h\in G} \fin , \ \ \
\rho\sprh\sigma\defi (\urep(\hmo)\tre\rho)\spr\sigma \fin .
\end{equation}
This family of products is left inner equivariant wrt the pair $(\act,\ure)$, where $\act$ is the left action
of $G$ on itself; i.e.,
\begin{equation}
(\urep(\gmo)\tre\rho)\sprh\sigma = (\urep(\hmo\gmo)\tre\rho)\spr\sigma = \rho\sprgh\sigma \fin .
\end{equation}
In an analogous way, one can construct an outer equivariant family of stochastic products. It is clear that,
conversely, every family of stochastic products labeled by the points of the group $G$, and left inner equivariant
wrt to the pair $(\act\colon G\times G\rightarrow G,\ure)$, where $\act$ is the left action of $G$ on itself, is
of the previously specified form. In fact, denoting by $e$ the identity of $G$, it is sufficient to set
\begin{equation}
\sprod\defi \sprode \fin .
\end{equation}
\end{example}


\section{Constructing a class of stochastic products: twirled products}
\label{main}


We will now explicitly construct a class of group-covariant, associative stochastic products. As clarified in the previous
section, every associative stochastic product is embedded in a suitable algebra on the Banach space $\trc$ of trace class
operators, a so-called \emph{stochastic algebra}. It will be convenient to develop some technical tools first, then to construct
such an algebra structure --- we will actually achieve a larger class of algebras, the `twirled algebras' --- next to prove the
covariance properties of these algebras wrt the relevant group actions and, finally, to obtain our stochastic products by imposing
suitable conditions and by restricting to density operators.


\subsection{First step: technical tools}

The basic ingredients of our construction are the following:

\begin{enumerate}

\item We consider a \emph{pair of projective representations} $\ure\colon G \rightarrow\unih$,
$\vre\colon G \rightarrow\unih$ of a l.c.s.c.\ group $G$ in a separable complex Hilbert space $\hh$.
We will further assume that the group $G$ is \emph{unimodular} --- i.e., $\modu\equiv 1$, where $\modu$ is the
modular function~\cite{Folland-AA,Raja} on $G$ --- and that the representation $\ure$ is (irreducible and)
\emph{square integrable}~\cite{Duflo,AnielloSDP,AnielloPR,Aniello-new,AnielloDMF}.

\item We fix a \emph{complex Borel measure}~\cite{Folland-RA} $\nu$ on $G$. We will denote by $\come$ the Banach space of all
complex Borel measures on $G$. Notice that, since $G$ is a l.c.s.c.\ topological space, every measure in $\come$ is
regular, hence, a Radon measure; see, e.g., Theorem~{7.8} of~\cite{Folland-RA}. We will denote by $\posme\subset\come$ the
convex cone of all finite, positive Borel measures on $G$ and by $\prome\subset\posme$ the convex set of Borel probability measures.
Recall that, for every $\nu\in\posme$, the union of all $\nu$-null open subsets of $G$ is a $\nu$-null open set too, whose
complement is the support $\supp(\nu)$ of the Radon measure $\nu$ (equivalently, $\supp(\nu)$ is the intersection of all
closed subsets of $G$ of full $\nu$-measure); see~\cite{Folland-RA}, chapt.~{7}, sect.~{1}, or~\cite{Bogachev}, chapt.~{7},
sect.~{2}. Let us also point out that, since the topological group $G$ is second countable, for every pair $\mu_1$, $\mu_2$
of Radon measures on $G$, the standard product measure $\mu_1\otimes\mu_2$ is a Radon measure too; see Theorem~{7.20}
of~\cite{Folland-RA}.

\item We also fix a \emph{fiducial operator} $\fiv$ in $\trc$. The twirled algebra that we are going to define will
depend on the choice of this trace class operator, as it will be clear soon, as well as on the choice of
the previously mentioned complex measure $\nu$.

\end{enumerate}

The fact that the representation $\ure$ is square integrable entails that, for every pair of vectors
$\phi,\psi\in\hh$, the Borel function $G\ni g\mapsto\langle \phi, \ure(g) \psi \rangle\in\ccc$ is square integrable
wrt (a normalization of) the Haar measure~\cite{Raja,Folland-AA} $\hame$ on $G$, which, in this case ($G$ unimodular),
is both left and right invariant. Moreover, the so-called \emph{orthogonality relations} hold; i.e.,
\begin{equation} \label{ortorel}
\intG \de\hame(g) \nove \langle \eta, \ure(g) \phi \rangle \langle \ure(g) \psi , \chi \rangle =
\cu\tre \langle \eta, \chi \rangle \langle \psi , \phi  \rangle, \ \ \
\forall\cinque\eta, \chi,\psi , \phi\in\hh \fin ,
\end{equation}
where $\cu$ is a (strictly) positive constant, depending only on $\ure$ and on the normalization of $\hame$.

\begin{remark} \label{duflom}
In the general (i.e., not necessarily unimodular) case, the orthogonality relations for a square integrable
representation involve a positive selfadjoint linear operator in $\hh$ --- the so-called \emph{Duflo-Moore
operator}~\cite{AnielloSDP,AnielloPR,Aniello-new,AnielloDMF} --- which is bounded if and only if $G$ is unimodular,
and, in such case, this operator is a multiple of the identity: $\duf = \ddu \sei I$, $\ddu\equiv\sqrcu>0$. In particular,
every irreducible unitary representation of a \emph{compact} group is square integrable, since the Haar measure
$\hame$ of such a group is \emph{finite}. Moreover, if $\hame$ is normalized as a probability measure,
then $\cu=\dim(\mathcal{H})^{-1}$ (according to the Peter-Weyl theorem~\cite{Folland-AA}).
\end{remark}

Recall that the projective representations $\ure$ and $\vre$ give rise to (strongly) continuous isometric representations
\begin{equation}
\urep,\vrep\colon G\rightarrow\spitrc
\end{equation}
of the l.c.s.c.\ group $G$, acting in the Banach space $\trc$; see~{(\ref{isorep})}. Notice that, if $\mul\colon G\times G\rightarrow\toro$
is the multiplier~\cite{Raja} of $\ure$, then $\ure(g)^\ast=\mul(g,\gmo)\quattro\ure(\gmo)$; hence, by the cyclic property of the trace,
\begin{equation} \label{cicpro}
\tr\big(A\tre (\urep(g)\fiv)\big) = \tr\big((\urep(\gmo) A)\tre\fiv\big) ,
\ \ \ \forall\cinque A,\fiv \in \trc \fin , \forall\cinque g\in G \fin .
\end{equation}

We will now establish two fundamental technical facts.

\begin{lemma} \label{tralem}
For every pair of trace class operators $A,\fiv\in\trc$, the bounded continuous function
$G\ni g\mapsto \tr\big(A\tre (\urep(g) \fiv)\big)=\tr(A\tre \ure(g) \tre\fiv\tre \ure(g)^\ast)\in\ccc$ is $\hame$-integrable and
\begin{equation} \label{fundrel}
\cumo\intG \de\hame(g) \nove \tr\big(A\tre (\urep(g) \fiv )\big) = \tr(A) \sei \tr(\fiv) \fin ,
\end{equation}
where $\cu > 0$ is the constant appearing in the orthogonality relations~{(\ref{ortorel})} for the square integrable
representation $\ure$. Therefore, we can define a complex measure $\mesaf\in\come$ by setting
\begin{equation} \label{defmesaf}
\de\mesaf(g)\defi\cumo\sei\tr\big(A\tre (\urep(g) \fiv)\big)\sei\de\hame(g) \fin .
\end{equation}
\end{lemma}

\begin{proof}
The bounded function $G\ni g\mapsto \tr\big(A\tre (\urep(g) \fiv)\big)\in\ccc$ is also continuous, because the representation $\urep$
is strongly continuous. Relation~{(\ref{fundrel})} is the `second trace formula for square integrable
representations'; see~\cite{Aniello-FT}, Proposition~7.
\end{proof}

\begin{remark}
By the previous lemma, it is clear that, if $A$ and $\fiv$ are \emph{density operators}, then $\mesaf$ is a Borel
\emph{probability measure} on $G$.
\end{remark}

\begin{lemma} \label{lempos}
For every trace class operator $B\in\trc$ and every complex measure $\nu\in\come$, the vector-valued function
$G\ni g\mapsto(\vrep(g) B )\in\trc$ is Bochner-integrable wrt to $\nu$. Moreover, the trace class operator
$\twirlnv\sei B$, defined as a Bochner integral
\begin{equation} \label{bochinte}
\twirlnv\sei B\defi\intG \de \nu(g) \nove  (\vrep(g) B )\in\trc \fin ,
\end{equation}
has trace equal to $\nu(G)\sei\tr(B)$, and, in the case where $\nu\in\posme$ and $B\in\trcp$, it is a positive
element of $\trc$. Finally, if $\nu\in\prome$, then $\twirlnv\sei B$ is contained in the closed
convex hull
\begin{equation}
\clco(\{\vrep(g) B\in\trc\colon g\in\supp(\nu)\})\subset\clco(\vrep(G) B)\subset\trc \fin .
\end{equation}
\end{lemma}

\begin{proof}
Taking into account that the Banach space $\trc$ is separable ($\hh$ being separable), by Pettis' measurability
theorem~\cite{Diestel}, the continuous function $G\ni g\mapsto(\vrep(g) B )\in\trc$ is $\nu$-measurable. Moreover,
this vector-valued function is bounded (wrt the norm $\trnor$). By these facts, it is Bochner-integrable wrt $\nu$.
Observe now that, exchanging the trace with the Bochner integral (so getting an ordinary integral of a $\ccc$-valued
function),
\begin{equation}
\tr(\twirlnv\sei B)= \intG \de \nu(g) \nove \tr(\vre(g)\tre B\tre \vre(g)^\ast)=
\nu(G)\sei\tr(B) \fin ,
\end{equation}
and, for every $\psi\in\hh$,
\begin{equation}
\langle \psi , (\twirlnv\sei B)\tre\psi\rangle =
\intG \de \nu(g) \nove \langle \psi , (\vre(g)\tre B\tre \vre(g)^\ast)\tre\psi\rangle \fin .
\end{equation}
Thus, if, in particular, $\nu$ is a positive measure and $B\in\trcp$, then $\twirlnv\sei B\in\trcp$.
The final assertion relies on the fact that we can restrict the integral in~{(\ref{bochinte})} to
$\supp(\nu)$ and on a well known property of Bochner's integral; see~\cite{Diestel},
chapt.~{II}, Corollary~{8}.
\end{proof}

By the two previous lemmas, for every tern of operators $A,\fiv,B\in\trc$, we can consistently define
a trace class operator $\ltripr A,\fiv,B\mtripr\in\trc$ by setting
\begin{equation} \label{expcon}
\ltripr A,\fiv,B\mtripr \defi\twirmenv=\intG \de(\mesaf\conv\nu)(g) \nove (\vrep(g) B ) \fin ,
\end{equation}
where the integral on the rhs is, as in~{(\ref{bochinte})}, a Bochner integral, and the measure $\mesaf\conv\nu\in\come$
is the \emph{convolution}~\cite{Folland-AA} of the complex Borel measures $\mesaf$ and $\nu$.

\begin{proposition} \label{promea}
For every $A,\fiv,B\in\trc$, we have that
\begin{eqnarray}
\ltripr A,\fiv,B\mtripr
\equa
\intGG \de\mesafnu (g,h) \nove (\vrep(gh) B)
\equb
\intG \de\mesaf(g) \intG \de \nu(h) \nove (\vrep(gh) B)
\equb \label{defmtripr}
\cumo \intG \de\hame(g)\intG \de \nu(h) \nove \tr\big(A\tre (\urep(g) \fiv)\big)\tre (\vrep(gh) B) \fin ,
\end{eqnarray}
where $\mesafnu$ is the standard product measure, all integrals are in the sense of Bochner and the iterated integrals
can be interchanged.
\end{proposition}

\begin{proof}
The function $G\times G\ni (g,h)\mapsto(\vrep(gh) B )\in\trc$ is continuous --- hence, measurable wrt the
product measure $\mesafnu$ ($\trc$ being separable) --- and norm bounded. Therefore, it is Bochner-integrable.
By Fubini's theorem for Bochner integrals (see, e.g., Theorem~{3.7.13} of~\cite{Hille}),
\begin{equation} \label{bochints}
\intGG \de\mesafnu (g,h) \nove (\vrep(gh) B) = \intG \de\mesaf(g) \intG \de \nu(h) \nove (\vrep(gh) B) \fin ,
\end{equation}
where the iterated Bochner integrals on the rhs can be interchanged. Moreover, for every $\phi,\psi\in\hh$, we have that
\begin{equation}
\intG \de(\mesaf\conv\nu)(g) \nove \langle\phi, (\vrep(g) B)\psi\rangle
= \intG \de\mesaf(g) \intG \de \nu(h) \nove \langle\phi, (\vrep(gh) B)\psi\rangle \fin ,
\end{equation}
by the definition of the convolution of measures (notice that this relation holds even if the integrand
functions are only assumed to be continuous and bounded; see Remark~2 of~\cite{AnielloSOP}). Thus,
by~{(\ref{bochints})},
\begin{equation}
\ltripr A,\fiv,B\mtripr \defi
\intG \de(\mesaf\conv\nu)(g) \nove (\vrep(g) B ) = \intGG \de\mesafnu (g,h) \nove (\vrep(gh) B) \fin ,
\end{equation}
where, again, Bochner integrals are understood.
\end{proof}

Note that --- for $\nu=\delta\equiv\delta_e$ (Dirac measure at the identity $e$ of $G$) ---
$\mesaf\conv\nu=\mesaf\conv\delta=\mesaf$ and we have:
\begin{equation}
\ltripr A,\fiv ,B\rtripr \equiv \ltripr A,\fiv ,B\dtripr = \twirmesa = \cumo
\intG \de\hame(g) \nove \tr\big(A\tre (\urep(g) \fiv)\big)\tre (\vrep(g) B) \fin .
\end{equation}

In the following, we will also need a further technical result. Denoting by $|\mesaf|\in\posme$
the \emph{total variation} measure~\cite{Folland-RA} associated with the complex measure $\mesaf$,
we have:
\begin{lemma} \label{lemnorms}
For every pair of trace class operators $A,\fiv\in\trc$, the $\come$-norm $\|\mesaf\|\defi|\mesaf|(G)$
of the associated complex measure $\mesaf$ satisfies the inequality
\begin{equation} \label{tecresu}
\|\mesaf\|\le \|A\trn \tre \|\fiv\trn \fin .
\end{equation}
\end{lemma}

\begin{proof}
Let us consider the singular value decomposition of $A,\fiv\in\trc$ (see, e.g.,~\cite{Moretti}, chapt.~{4}):
\begin{equation}
A=\sum_k\erk\sei\cek \fin , \ \ \ \fiv=\sum_l \esl\sei\fpl   \fin .
\end{equation}
Here, $\{\erk\}$, $\{\esl\}$ are the singular values of $A$, $\fiv$ --- thus: $\|A\trn=\sum_k\erk$
and $\|\fiv\trn=\sum_k\esl$ --- $\{\chk\}$, $\{\etk\}$, $\{\phl\}$, $\{\psl\}$ are orthonormal systems
in $\hh$ and convergence wrt the trace norm is understood. We then have:
\begin{equation} \label{serexp}
\tr\big(A\tre (\urep(g) \fiv )\big) = \sum_k\sum_l\erk\tre\esl\sei \tr\big(\cek\tre (\urep(g) \fpl )\big)
= \sum_{k,l}\erk\tre\esl\sei\akl(g)\tre\overline{\bkl(g)} \fin ,
\end{equation}
for all $g\in G$ (absolute convergence on the rhs), where $\akl(g)\defi\langle\etk,\ure(g)\tre\phl\rangle$,
$\bkl(g)\defi\langle\chk,\ure(g)\tre\psl\rangle$ ($|\akl(g)|,|\bkl(g)|\le 1$). Since $\ure$ is a square integrable representation,
$\akl,\bkl\in\ldg\equiv\elleg$ and, by the orthogonality relations~{(\ref{ortorel})},
\begin{equation}
\|\akl\norldue=\bbile\intG \de\hame(g)\nove |\akl(g)|^2 \bbiri^{\motto 1/2}\motto = \sei \sqrcu
= \bbile\intG \de\hame(g)\nove |\bkl(g)|^2\bbiri^{\motto 1/2} .
\end{equation}
Moreover, the function $\akl\tre\overline{\bkl}$ belongs to $\lug$. Hence, by~{(\ref{serexp})}, we have a series
of functions in $\lug$ converging point-wise to the function $\tr\big(A\tre (\urep\argo \fiv )\big)\in\lug$ (Lemma~{\ref{tralem}}).
By a well known result (an immediate consequence, e.g., of Corollary~{2.32} of~{\cite{Folland-RA}}), it follows that
the series $\sum_{k,l}\erk\tre\esl\sei\akl\tre\overline{\bkl}$ --- which converges absolutely in $\lug$, since
$\|\akl\tre\overline{\bkl}\norluno\le\|\akl\norldue\tre\|\bkl\norldue =\cu$ --- must converge to
$\tr\big(A\tre (\urep\argo \fiv )\big)$ in $\lug$, as well; therefore:
\begin{eqnarray}
\|\tr\big(A\tre (\urep\argo \fiv )\big)\norluno
\equa
\Big\|\sum_{k,l}\erk\tre\esl\sei\akl\tre\overline{\bkl}\tre\bnorluno
\nonumber \\
& \le & \spa
\sum_{k,l}\erk\tre\esl\sei \|\akl\tre\overline{\bkl}\norluno
\nonumber \\
& \le & \spa
\sum_{k,l}\erk\tre\esl\sei \|\akl\norldue\tre\|\bkl\norldue
=\cu\sei \|A\trn \tre \|\fiv\trn  \fin .
\end{eqnarray}
At this point, it is sufficient to notice that
\begin{equation}
|\mesaf|(G) = \cumo\intG \de\hame(g) \nove |\tr\big(A\tre (\urep(g) \fiv )\big)|
= \cumo\sei \|\tr\big(A\tre (\urep\argo \fiv )\big)\norluno \le
\|A\trn \tre \|\fiv\trn \fin ,
\end{equation}
and the proof is complete.
\end{proof}


\subsection{Second step: constructing twirled algebras}

Let us now study the main properties of the trace class operator $\ltripr A,\fiv,B\mtripr$. It will be convenient,
from this point onwards, to fix the normalization of the Haar measure $\hame$ in such a way that $\cu =1$.

It is clear that the map
\begin{equation} \label{fourma}
\mtriprod\colon\trc\times\trc\times\trc\times\come\rightarrow\trc
\end{equation}
is linear wrt each of its four arguments; besides:

\begin{proposition}
For every $A,\fiv,B\in\trc$ and every $\nu\in\come$, we have that
\begin{equation} \label{reladj}
\big(\ltripr A,\fiv,B\mtripr\big)^\ast = \ltripr A^\ast,\fidu,B^\ast\mtriprcc  \fin ,
\end{equation}
where $\ccme$ is the complex conjugate of the measure $\nu$. Hence, if the operators $A,\fiv,B$ are selfadjoint
and $\nu$ is real-valued (i.e., a finite signed Borel measure), then $\ltripr A,\fiv,B\mtripr$ is selfadjoint too.
Furthermore,
\begin{equation} \label{relpos}
A\ge 0,\ \fiv\ge 0,\ B\ge 0,\ \nu\ge 0 \ \Longrightarrow \ \ltripr A,\fiv,B\mtripr \ge 0 \fin .
\end{equation}
\end{proposition}

\begin{proof}
Observe that we have $(\vrep(g) B )^\ast= \vrep(g) B^\ast$ and, by
the basic properties of the trace,
\begin{equation}
\tr\big(A\tre (\urep(g) \fiv)\big)^\ast = \tr\big((\urep(g) \fiv)^\ast A^\ast\big)
= \tr\big(A^\ast (\urep(g) \fidu) \big) .
\end{equation}
Then, by~{(\ref{defmtripr})}, relation~{(\ref{reladj})} holds. Relation~{(\ref{relpos})}
is a consequence of the second assertion of Lemma~{\ref{lempos}} and of the fact that,
for $A,\fiv\ge 0$, $\mesaf$ is a positive measure.
\end{proof}

Another remarkable property of the map $\mtriprod$ is the following:
\begin{proposition}
For every $A,\fiv,B\in\trc$ and every $\nu\in\come$, we have that
\begin{equation}
\tr\big(\ltripr A,\fiv,B\mtripr\big) = \nu(G)\sei\tr(A)\sei\tr(\fiv)\sei\tr(B)
\end{equation}
and
\begin{equation}
\left\|\ltripr A,\fiv,B\mtripr\right\trn \le \|\nu\|\sei \|A\trn\tre \|\fiv\trn\tre \|B\trn \fin .
\end{equation}
\end{proposition}
\begin{proof}
Recalling the expression of the trace class operator $\ltripr A,\fiv,B\mtripr$ appearing in the last line
of~{(\ref{defmtripr})} (with $\cu=1$), observe that, exchanging the trace with the Bochner integrals, we have:
\begin{eqnarray}
\tr\big(\ltripr A,\fiv,B\mtripr\big)
\equa
\intG \de\hame(g)\nove \tr\big(A\tre (\urep(g) \fiv)\big) \sei
\intG \de \nu(h) \nove\tr(\vrep(gh) B)
\equb
\nu(G)\sei\tr(B)
\intG \de\hame(g)\nove \tr\big(A\tre (\urep(g) \fiv)\big) \fin .
\end{eqnarray}
Then, by the first assertion of Lemma~{\ref{tralem}}, we conclude that
$\tr\big(\ltripr A,\fiv,B\mtripr\big)=\nu(G)\sei\tr(A) \sei \tr(\fiv) \sei \tr(B)$.
Next, by a well known property of the Bochner integral (see, e.g., property~{(c)} in Proposition~{4.5}
of~{\cite{Busch}}), and taking into account the norm estimate $\|\mesaf\|\le \|A\trn \tre \|\fiv\trn$
in Lemma~{\ref{lemnorms}}, we find:
\begin{eqnarray}
\left\|\ltripr A,\fiv,B\mtripr\right\trn
\spa & \le & \spa
\intG \de\hame(g)\nove \big|\tr\big(A\tre (\urep(g) \fiv)\big)\big|
\left\|\intG \de \nu(h) \nove (\vrep(gh) B)\right\trn
\equb
\intG \de|\mesaf|(g)\nove
\left\|\intG\de\nu(h) \nove (\vrep(h) B)\right\trn
\nonumber \\
& \le & \spa
\intG \de|\mesaf|(g)\nove \intG\de|\nu|(h)\nove \|\vrep(h) B\trn
\equb
|\mesaf|(G)\sei |\nu|(G)\sei\|B\trn =
\|\mesaf\|\sei \|\nu\|\sei \|B\trn
\le \|\nu\|\sei \|A\trn\tre \|\fiv\trn\tre \|B\trn \fin .
\end{eqnarray}
for all $A,\fiv,B\in\trc$ and $\nu\in\come$.
\end{proof}

The associativity of the stochastic algebras that we are going to define is ensured by the following:

\begin{proposition} \label{proas}
If the representation $\vre$ coincides with the square integrable representation $\ure$,
then, for every $A,B,C,\fiv\in\trc$ and every $\nu\in\come$, we have that
\begin{equation} \label{associa}
\ltripr\ltripr A,\fiv,B\mtripr , \fiv,C \mtripr = \ltripr A,\fiv,\ltripr B,\fiv,C \mtripr\mtripr \fin .
\end{equation}
\end{proposition}

\begin{proof}
Recalling relation~{(\ref{defmtripr})}, setting $\vre\equiv\ure$ (with $\ure$ square integrable) and exchanging
the trace with the Bochner integrals, we have:
\begin{eqnarray}
\ltripr\ltripr A,\fiv,B\mtripr , \fiv,C \mtripr
\equa
\intG \de\hame(\git) \intG \de \nu(\hit) \intG \de\hame(g) \intG \de \nu(h)
\nove \tr(A (\urep(g) \fiv) )
\nonumber \\ \label{associaa}
& \times & \spa
\tr((\urep(g\tre h) B ) (\urep(\git) \fiv)) \tre (\urep(\git\tre\hit)\tre C) \fin .
\end{eqnarray}
Here, \emph{all} the Bochner integrals can be interchanged (not only the first one with the second, and the remaining two,
by the final claim of Proposition~{\ref{promea}}). Indeed --- setting $\fivg\equiv\urep(g) \fiv$ etc.\ ---
the function $(g,h,\git,\hit)\mapsto\tr(A\tre\fivg)\sei\tr(\bgh\tre\fivgit)$ is integrable wrt the product measure
$\mu\equiv\hame\otimes\nu\otimes\hame\otimes\nu$. This fact is easily checked by taking $A$, $B$ and $\fiv$ in $\trcp$, and
$\nu$ in $\posme$, at first, so that, by Tonelli's theorem and by the first assertion of Lemma~{\ref{tralem}},
\begin{eqnarray}
\int\de\mu(g,h,\git,\hit)\nove\tr(A\tre\fivg)\sei\tr(\bgh\tre\fivgit)
\equa
\intG\de\hame(g)\intG\de\nu(h)\intG\de\hame(\git)\intG\de\nu(\hit)
\nove\tr(A\tre\fivg)\sei\tr(\bgh\tre\fivgit)
\equb
\nu(G)^2\sei\tr(B)\sei\tr(\fiv) \intG\de\hame(g)\nove\tr(A\tre\fivg)
\equb
\nu(G)^2\sei\tr(A)\sei\tr(B)\sei\tr(\fiv)^2 \fin .
\end{eqnarray}
By linearity, one extends this conclusion to all $A,B,\fiv\in\trc$ and $\nu\in\come$ (recall~\cite{Folland-RA}
that $\nu$ admits a decomposition $\nu=\nuu + \ima \sei\nud=\nup\msei -\num \msei + \ima \sei(\ndp\msei-\ndm)$,
$\nup,\num,\ndp,\ndm\in\posme$, analogous to~{(\ref{nota})--(\ref{notc})}). Therefore, the integrals in~{(\ref{associaa})}
can be freely interchanged by Fubini's theorem for Bochner integrals, because the function
$(g,h,\git,\hit)\mapsto\urep(\git\tre\hit)\tre C$ is continuous and norm bounded, hence,
Bochner-integrable wrt $\tr(A\tre\fivg)\sei\tr(\bgh\tre\fivgit)\sei\de\mu(g,h,\git,\hit)$.

Besides, since the bounded linear operator $\urep(gh)$ in $\trc$ and the relevant Bochner integrals can be
interchanged, and taking into account the final claim of Proposition~{\ref{promea}}, we have:
\begin{eqnarray}
\ltripr A,\fiv,\ltripr B,\fiv,C \mtripr\mtripr
\equa
\intG \de \nu(h) \intG \de\hame(g) \intG \de \nu(\hit) \intG \de\hame(\git)
\nove \tr(A (\urep(g) \fiv) )
\equx
\tr(B (\urep(\git) \fiv))\tre (\urep(g\tre h\tre\git\tre\hit) C)
\equb
\intG \de \nu(h) \intG \de\hame(g) \intG \de \nu(\hit) \intG \de\hame(\git)
\nove \tr(A (\urep(g) \fiv) )
\equx
\tr(B  (\urep(\hmo\gmo\git) \fiv))\tre (\urep(\git\tre\hit) C)
\equb
\intG \de \nu(h) \intG \de\hame(g) \intG \de \nu(\hit) \intG \de\hame(\git)
\nove \tr(A (\urep(g) \fiv) )
\equx \label{associab}
\tr((\urep(g\tre h) B ) (\urep(\git) \fiv))\tre (\urep(\git\tre\hit) C) \fin ,
\end{eqnarray}
Here, for obtaining the second equality, we have exploited the change of variables
$\git\mapsto\hmo\gmo\git$ and the invariance of the Haar measure $\hame$, while the
third equality follows from relation~{(\ref{cicpro})}. Finally, comparing the expressions
obtained in~{(\ref{associaa})} and~{(\ref{associab})}, where the order of the iterated
integrals is irrelevant, we see that relation~{(\ref{associa})} is verified.
\end{proof}

Finally, we focus on the case where the l.c.s.c.\ group $G$ is abelian.

\begin{proposition}
If $G$ is abelian and $\vre$ coincides with the square integrable representation $\ure$,
then, for every $A,B,\fiv\in\trc$ and every $\nu\in\come$, we have that
\begin{equation} \label{commab}
\ltripr A,\fiv,B\mtripr = \ltripr B,\fiv,A\mtripr \fin .
\end{equation}
\end{proposition}

\begin{proof}
As in the proof of Proposition~{\ref{proas}}, we set $\fivg\equiv\urep(g) \fiv$ etc. Notice that
$\fivggit=(\fivgit)_g$, $(A\fiv)_{g}=A_{g}\tre\fivg$ and, since $G$ is abelian, $\fivggit=\fivgitg$.
Let us assume that $A$, $B$ and $\fiv$ are in $\trcp$, and $\nu$ in $\posme$, at first,
and take any positive bounded operator $C\in\bop$, so that we deal with positive functions
in the following argument and we can freely exchange integrals, by Tonelli's theorem.
By~{(\ref{defmtripr})} and, next, by the first assertion of Lemma~{\ref{tralem}}, we have:
\begin{eqnarray}
\tr(\ltripr A,\fiv,B\mtripr\sei C)
\equa
\intG\de\hame(g)\intG\de\nu(h)
\nove\tr(A\tre\fivg)\sei\tr(\bgh\tre C)
\equb
\intG\de\hame(g)\intG\de\nu(h)\intG\de\hame(\git)
\nove\tr(\bgh\tre C\tre(A\tre\fivg)_{\git}) \fin .
\end{eqnarray}
At this point, let us observe that
\begin{equation}
\tr(\bgh\tre C\tre(A\tre\fivg)_{\git}) =
\tr(\bgh\tre C\tre\agit\tre\fivgitg)
= \tr(C\agit\tre\fivggit\tre\bgh) = \tr(C\agit\tre(\fivgit\tre\bh)_g) \fin .
\end{equation}
By this relation, by suitably exchanging the integrals and, next, again by the first assertion of
Lemma~{\ref{tralem}}, we find that
\begin{eqnarray}
\tr(\ltripr A,\fiv,B\mtripr\sei C)
\equa
\intG\de\nu(h)\intG\de\hame(\git)\intG\de\hame(g)
\nove\tr(C\agit\tre(\fivgit\tre\bh)_g)
\equb
\intG\de\nu(h)\intG\de\hame(\git)
\nove\tr(C\agit)\sei\tr(\fivgit\tre\bh)
\equb
\intG\de\nu(h)\intG\de\hame(\git)
\nove\tr(C\agit)\sei\tr(\figihmo\tre B)
\equb
\intG\de\nu(h)\intG\de\hame(\git)
\nove\tr(C\agith)\sei\tr(\fivgit\tre B)
\equb
\intG\de\hame(g)\intG\de\nu(h)
\nove\tr(\fivg\tre B)\sei\tr(\agh\tre C)
= \tr(\ltripr B,\fiv,A\mtripr\sei C) \fin .
\end{eqnarray}
Note that the third and the fourth of the above equalities are obtained by relation~{(\ref{cicpro})}
(since $\hmo\git=\git\hmo$) and by the change of variables $\git\mapsto\git h$, respectively.
Hence, by linearity, we conclude that $\tr(\ltripr A,\fiv,B\mtripr\sei C)=\tr(\ltripr B,\fiv,A\mtripr\sei C)$,
without any assumption of positivity of the integrand functions; i.e., for all $A,B,\fiv\in\trc$, $\nu\in\come$
and $C\in\bop$. Finally, by the arbitrariness of the operator $C$ in $\bop$, it follows that relation~{(\ref{commab})}
is verified.
\end{proof}

Summarizing all the previous facts, we obtain the following result:

\begin{theorem}
Let $\ure\colon G \rightarrow\unih$, $\vre\colon G \rightarrow\unih$ be projective representations of
a unimodular l.c.s.c.\ group $G$ in a separable complex Hilbert space $\hh$, with $\ure$ square integrable,
and, for every $\fiv\in\trc$ and $\nu\in\come$, let us consider the bilinear map
\begin{equation} \label{defibilf}
\bilfmo\colon\trc\times\trc\ni (A,B)\mapsto\ltripr A,\fiv,B\mtripr\in\trc \fin .
\end{equation}
If $\nu(G)\sei\tr(\fiv)=1$, this map is trace-preserving. If $\fiv\in\trcsa$ and $\nu$ is real-valued, then the
bilinear map is adjoint-preserving; in particular, it is positive, if $\fiv,\nu\ge 0$. If the representation
$\vre$ coincides with the square integrable representation $\ure$, the product~{(\ref{defibilf})} is associative,
hence if, moreover, $\fiv\in\trc$ and $\nu\in\come$ are such that $\|\nu\|\sei \|\fiv\trn\le 1$, then the Banach space
$\trc$, endowed with the product~{(\ref{defibilf})}, is a Banach algebra; in particular,
\begin{equation} \label{ineqba}
\left\|A\bilfm B\right\trn \le \|A\trn \tre \|B\trn \fin ,
\ \ \ \forall\cinque A,B \in\trc \fin  .
\end{equation}
Finally, if $G$ is abelian and $\vre$ coincides with the square integrable representation $\ure$, then
the product~{(\ref{defibilf})} is commutative.
\end{theorem}

There is a particular case of the previous result which deserves a special attention:

\begin{corollary} \label{corsto}
If $\ure=\vre\colon G \rightarrow\unih$ is a square integrable projective representation, $\fiv$ a density operator in $\hh$
and $\nu$ a Borel probability measure on $G$, then the Banach space $\trc$, endowed with the binary operation~{(\ref{defibilf})},
is both a stochastic algebra and a Banach algebra. This algebra is also commutative in the case where $G$ is abelian.
\end{corollary}

\begin{remark}
Under the assumptions of the previous corollary, in the Banach algebra $\trc$ the inequality~{(\ref{ineqba})} is actually
saturated by every pair of positive operators $A$ and $B$.
\end{remark}

We will call a pair of the form
\begin{equation}
(\trc,\bilfmo\colon\trc\times\trc\rightarrow\trc)
\end{equation}
the \emph{twirled algebra} on $\trc$ generated by the tetrad $(\ure,\vre\colon G \rightarrow\unih;\fiv\in\trc,\nu\in\come)$.


\subsection{Third step: proving covariance, equivariance and invariance properties of twirled algebras}

Given a complex measure $\nu\in\come$, let us first define the associated left and right $g$-translate measures ---
denoted by $\nug$ and $\nugr$, respectively --- as
\begin{equation}
\nug (\bor) \defi \nu(\gmo\bor) \fin , \ \ \ \nugr (\bor) \defi \nu(\bor g) \fin ,
\end{equation}
where $\bor$ is a Borel subset of $G$; namely, for every Borel function $f\colon G\rightarrow\ccc$,
\begin{equation}
\intG \de\nug(h) \nove f(h) = \intG \de\nu(h) \nove f(g h) \fin , \ \ \
\intG \de\nugr(h) \nove f(h) = \intG \de\nu(h) \nove f(h\gmo) \fin .
\end{equation}
Observe that the maps
\begin{equation}
G\times \come\ni(g,\nu)\mapsto\nug\in\come \ \ \mbox{and} \ \
G\times \come\ni(g,\nu)\mapsto\nugr\in\come
\end{equation}
are (left) group actions, because
\begin{equation}
\nugh = (\nuh)^g  \ \ \mbox{and} \ \ \nughr = (\nuhr)_g \fin .
\end{equation}

\begin{lemma}
For every $A,\fiv,B\in\trc$ and every $\nu\in\come$, we have:
\begin{equation} \label{prire}
\vrep(g)\tre \ltripr A,\fiv,B\mtripr  = \ltripr \urep(g) A,\fiv, B\mtripr
\end{equation}
\begin{equation} \label{secre}
\ltripr A, \urep(\gmo) \fiv, B\mtripr = \ltripr A,\fiv, B\mtriprng,
\ \ \ \ltripr A,  \fiv, \vrep(\gmo) B\mtripr = \ltripr A,\fiv, B\mtriprngr,
\end{equation}
\begin{equation} \label{trire}
\ltripr A,  \urep(g) \fiv,  B\rtripr = \ltripr A,\fiv, \vrep(\gmo) B\rtripr \fin .
\end{equation}
\end{lemma}

\begin{proof}
Let us prove relation~{(\ref{prire})}. In fact, $\vrep(g)$, being a bounded operator in $\trc$,
commutes with the Bochner integrals and we have:
\begin{eqnarray}
\vrep(g)\tre \ltripr A,\fiv,B\mtripr \equa
\intG \de\hame(\git) \intG \de \nu(h)
\nove \tr\big(A\tre (\urep(\git) \fiv)\big)\tre (\vrep(g\git h) B)
\equb
\intG \de\hame(\git) \intG \de \nu(h)
\nove \tr\big(A\tre (\urep(\gmo\git) \fiv)\big)\tre (\vrep(\git h) B)
\equb
\intG \de\hame(\git) \intG \de \nu(h)
\nove \tr\big((\urep(g) A)\tre (\urep(\git) \fiv)\big)\tre (\vrep(\git h) B)
\equb
\ltripr \urep(g) A,\fiv, B\mtripr \fin .
\end{eqnarray}
Regarding the first of relations~{(\ref{secre})}, observe that, by the right-invariance
of the Haar measure $\hame$, we find:
\begin{eqnarray}
\ltripr A, \urep(\gmo) \fiv, B\mtripr \equa
\intG \de\hame(\git) \intG \de \nu(h)
\nove \tr\big(A\tre (\urep(\git \gmo) \fiv)\big)\tre (\vrep(\git h) B)
\equb
\intG \de\hame(\git) \intG \de \nu(h)
\nove \tr\big(A\tre (\urep(\git ) \fiv)\big)\tre (\vrep(\git g h) B)
\equb
\ltripr A,\fiv, B\mtriprng \fin .
\end{eqnarray}
The proofs of the second of relations~{(\ref{secre})} and of relation~{(\ref{trire})} are similar.
Relation~{(\ref{trire})} can also be regarded as a consequence of the two relations~{(\ref{secre})} and of
the fact that $\delg=\delgmo$.
\end{proof}

\begin{proposition} \label{procova}
Let
\begin{equation}
(\trc,\bilfmo\colon\trc\times\trc\rightarrow\trc)
\end{equation}
be the twirled algebra generated by the tetrad $(\ure,\vre\colon G \rightarrow\unih;\fiv\in\trc,\nu\in\come)$.
Then, the algebra product is left-covariant wrt the pair $(\ure,\vre)$. Moreover, the family of products
\begin{equation} \label{famproa}
\{\bilfmo\colon\fiv\in\trc,\ \nu\in\come\}
\end{equation}
is right inner equivariant wrt the pair $(\act,\vre)$, where $\act$ is the (left) group action
\begin{equation} \label{actia}
\act\colon G\times (\trc\times\come) \rightarrow \trc\times\come \fin , \ \ \
g[(\fiv,\nu)] \defi (\fiv , \nugr) \fin ;
\end{equation}
it is invariant wrt the group action
\begin{equation} \label{actib}
\act\colon G\times (\trc\times\come) \rightarrow \trc\times\come \fin , \ \ \
g[(\fiv,\nu)] \defi (\urep(g)\fiv , \nug) \fin .
\end{equation}
Finally, setting $\nu=\delta$, the family of products
\begin{equation} \label{famprob}
\{\bilfo\colon\fiv\in\trc\}
\end{equation}
is right inner equivariant wrt the pair $(\act,\vre)$, where $\act$ is the group action
\begin{equation} \label{actic}
\act\colon G\times\trc\rightarrow \trc\times\come \fin , \ \ \
g[\fiv] \defi \urep(g)\fiv \fin .
\end{equation}
\end{proposition}

\begin{proof}
The first assertion follows directly from~{(\ref{prire})}. The right inner equivariance of the families of
products~{(\ref{famproa})} and~{(\ref{famprob})} is a consequence of the second of relations~{(\ref{secre})}
and of relation~{(\ref{trire})}, respectively. Exploiting the first of relations~{(\ref{secre})}, one obtains
the invariance of the family of products~{(\ref{famproa})} wrt the action~{(\ref{actib})}.
\end{proof}


\subsection{Fourth step: defining twirled stochastic products}

We now complete our program. Given a \emph{square integrable projective representation} $\ure\colon G \rightarrow\unih$,
for every \emph{fiducial state} $\fista\in\sta$ and every \emph{probability measure} $\prom\in\prome$, we define
\begin{equation} \label{defsto}
\rho\stofm\sigma\defi
\intG \de\hame(g)\intG \de \prom(h) \nove \tr\big(\rho\sei (\urep(g)\tre\fista)\big)\tre (\urep(gh)\tre \sigma) \fin ,
\ \ \ \rho , \sigma\in\sta \fin ,
\end{equation}
where, as above, the integrals are in the Bochner sense and the Haar measure $\hame$ is normalized in such a way that $\cu=1$.
For the sake of conciseness, we set $\rhog\equiv\urep(g)\tre\rho$. Recall, moreover, that $\promg$ and $\promgr$ denote, respectively,
the left and the right $g$-translates of the Borel probability measure $\prom$.

Taking into account Proposition~{\ref{promea}}, and applying Corollary~{\ref{corsto}} and Proposition~{\ref{procova}},
we obtain the following result:
\begin{theorem} \label{theosto}
With the previous notations and assumptions, the pair
\begin{equation}
(\sta,\stofmo\colon\sta\times\sta\rightarrow\sta)
\end{equation}
is an associative stochastic product that is left-covariant wrt the representation $\ure$, namely,
\begin{equation}
\rhog\stofm\sigma=\Bile\rho\stofm\sigma\Birigr \fin .
\end{equation}
Moreover, the family of stochastic products
\begin{equation} \label{famproa-bis}
\{\stofmo\colon\fista\in\sta,\ \prom\in\prome\}
\end{equation}
is right inner equivariant wrt the pair
$(\act,\ure)$, where $\act$ is the group action
\begin{equation} \label{actia-bis}
\act\colon G\times (\sta\times\prome) \rightarrow \sta\times\prome \fin , \ \ \
g[(\fista,\prom)] \defi (\fista , \promgr) \fin ,
\end{equation}
namely,
\begin{equation}
\rho\stofm\siggmo=\rho\stofmgr\sigma \fin ;
\end{equation}
it is invariant wrt the group action
\begin{equation} \label{actib-bis}
\act\colon G\times (\sta\times\prome) \rightarrow \sta\times\prome \fin , \ \ \
g[(\fista,\prom)] \defi (\fistag\equiv\urep(g)\tre\fista , \promg) \fin ,
\end{equation}
namely,
\begin{equation}
\rho\stofm\sigma=\rho\stofgmg\sigma \fin .
\end{equation}
Setting $\nu=\delta$, the family of stochastic products $\{\stofo\colon\fista\in\sta\}$
is right inner equivariant wrt the pair $(\act,\ure)$, where $\act$ is the group action
$\act\colon G\times (\sta\times\prome) \rightarrow \sta\times\prome$, $g[\fista] \defi \fistag$,
namely,
\begin{equation}
\rho\stof\siggmo=\rho\stofg\sigma \fin .
\end{equation}
In the case where the l.c.s.c.\ group $G$ is abelian, the stochastic product~{(\ref{defsto})}
is commutative.
\end{theorem}

\begin{remark}
In the case where the l.c.s.c.\ group $G$ is abelian, due to the commutativity of the stochastic
product~{(\ref{defsto})} and to the fact that $\promg=\promgrmo$, by combining the previous covariance,
equivariance and invariance properties, we find various further symmetry relations:
\begin{equation}
\rhog\stofm\sigma=\Bile\rho\stofm\sigma\Birigr=\rho\stofm\sigg
=\rho\stofmgrmo\sigma=\rho\stofmg\sigma=\rho\stofgmom\sigma \fin ,
\end{equation}
\begin{equation}
\rho\stofm\sigma=\rho\stofgmg\sigma=\rho\stofgmgrmo\sigma=\rho\stofgm\sigg
=\Bile\rho\stofgm\sigma\Birigr=\rhog\stofgm\sigma \fin ,
\end{equation}
\begin{equation}
\rhogmo\motto\stof\sigma=\Bile\rho\stof\sigma\Birigrmo\motto=\rho\stof\siggmo
\mtre=\rho\stofg\sigma  \fin .
\end{equation}
\end{remark}

We call the (associative) stochastic product~{(\ref{defsto})} the \emph{twirled stochastic product}
associated with the triple $(\ure,\fista,\prom)$. The algebra on $\trc$ that is obtained as the
canonical extension of this stochastic product will be called the \emph{twirled stochastic algebra}
associated with the triple $(\ure,\fista,\prom)$.


\section{Twirled stochastic products for every Hilbert space dimension}
\label{examples}


We now show by means of examples that there exist twirled stochastic products for every Hilbert space dimension
$2\le\dim(\hh)\le\infty$ (as usual, we neglect the trivial case where $\dim(\hh)=1$).


\subsection{Finite-dimensional twirled stochastic products}
\label{compactgrps}

Suppose that the group $G$ is compact (hence, unimodular) and $\ure\colon G\rightarrow\unih$ is an
irreducible unitary representation; thus, in this case: $\mul\equiv 1$ (the multiplier is trivial)
and $\dim(\hh)=n<\infty$. Then, $\ure$ is square integrable, because the Haar measure on $G$ is finite;
see~\cite{Aniello-new,Aniello-SP}. Moreover, by the classical Peter-Weyl theorem~\cite{Folland-AA},
assuming that the Haar measure is normalized as a probability measure --- $\hame(G)=1$ --- we have that
$\cu=\dim(\mathcal{H})^{-1}=n^{-1}$; hence:
\begin{equation} \label{compsto}
\rho\stofm\sigma = n
\intG \de\hame(g)\intG \de \prom(h) \nove \tr\big(\rho\sei (\urep(g)\tre\fista)\big)\tre (\urep(gh)\tre \sigma) \fin ,
\ \ \ (\hame(G)=1)
\end{equation}
for all $\rho,\fista,\sigma\in\sta$ and $\prom\in\prome$. In particular, for the maximally mixed state
$\mms\defi n^{-1}I$, the following noteworthy relations hold:
\begin{equation} \label{mmsrels}
\rho\stofm\mms=\mms \fin , \ \ \
\mms\stofm\sigma=\mms \fin , \ \ \
\rho\stomms\sigma= \mms \fin ,
\ \ \ \forall\cinque\rho,\fista,\sigma\in\sta \fin , \ \forall\cinque\prom\in\prome .
\end{equation}
The first of these relations is clear considering the expression~{(\ref{compsto})} of the twirled product
and recalling Lemma~{\ref{tralem}} ($\mu\in\come$, with $\de\mu(g)=n\sei\tr\big(\rho\sei (\urep(g)\tre\fista)\big)\nove\de\hame(g)$,
is a probability measure). The other two are a consequence of the fact that $\tr\big(\rho\sei (\urep(g)\tre\fista)\big)=n^{-1}$,
for $\rho=\mms$ or $\fista=\mms$, and of the relation
\begin{equation} \label{blurra}
\intG \de\hame(g)\nove (\urep(g) A)= n^{-1}\tre \tr(A)\sei I \fin , \ \ \ \forall\cinque A\in\trc \fin ;
\end{equation}
see Proposition~{6} of~\cite{Aniello-FT} (the `first trace formula for square integrable representations';
just make the substitution $g\mapsto\gmo$ in this formula, for $G$ unimodular). Indeed, by the previous facts
and by Lemma~{\ref{lempos}}, we have:
\begin{equation}
\mms\stofm\sigma= \rho\stomms\sigma =
\intG \de\hame(g)\nove\bbile\urep(g)\bbile\intG \de \prom(h) \nove (\urep(h)\tre \sigma)\bbiri\bbiri
= n^{-1} I \ifed \mms \fin .
\end{equation}
Notice that the second of relations~{(\ref{mmsrels})} also descends from Proposition~{\ref{propamms}},
by virtue of the left covariance of the twirled stochastic product.
Therefore, the maximally mixed state is a left and right collapsing point for the twirled product
--- i.e., $\bile\lmap(\mms)\biri(A)=\bile\rmap(\mms)\biri(A)=\tr(A)\sei\mms$, for all $A\in\trc$ ---
and, choosing $\mms$ as a fiducial state, we get a collapse product. Similarly, setting $\prom=\hame$,
one finds a collapse product too:
\begin{equation}
\rho\stoham\sigma = \mms \fin ,
\ \ \ \forall\cinque\rho,\fista,\sigma\in\sta \fin .
\end{equation}
This fact is readily verified using the invariance of the measure $\hame$ and, once again,
relation~{(\ref{blurra})}.

Observe that exploiting, e.g., the irreducible unitary representations of the group $\mathrm{SU}(2)$,
one is able to construct twirled products for every finite Hilbert space dimension. Moreover, composing the
twirled stochastic product (associated with an irreducible representation of a compact group) with suitable stochastic
maps --- regarded as convex-linear maps in $\sta$, see Remark~{\ref{samrea}} --- one can provide examples of stochastic
products having collapsing points different from $\mms$ and giving rise to different collapse channels.


\subsection{Infinite-dimensional products and the quantum convolution}
\label{weysys}

We now consider an infinite-dimensional example. Let $G$ be the group of translations on phase space
--- $G=\phasp$ --- $\hh=\ldrn$ and $\ure$ the \emph{Weyl system}~\cite{Aniello-new,Aniello-FT,Aniello-SP}:
\begin{equation} \label{schreps}
\bile \ure\qp\tre f\big)(x) = \eee^{-\ima\tre q\cdot p/2}\sei\eee^{\ima\tre p\cdot x} f(x-q) \fin ,\ \ \
\qp\in\phasp \fin , \ f\in \ldrn \fin ;
\end{equation}
i.e., $\ure\qp=\eee^{-\ima\tre q\cdot p/2}\sei\eip\sei\eiq$, where $\hq$, $\hp$ are the (vector) position and
momentum operators in $\ldrn$. This is a (strongly continuous) irreducible projective representation, with multiplier
\begin{equation} \label{mulwey}
\mul\qpqpt = \exp(\ima(q\cdot\tip - p\cdot\tiq)/2) \fin .
\end{equation}
Moreover, the Weyl system is square integrable and, setting $\ldg=\lrrnco$, we have that
$\cu=1$. Therefore, in this case the twirled product associated with the triple $(\ure,\fista,\prom)$
is of the form
\begin{eqnarray}
\rho\stofm\sigma
\equa
\frac{1}{(2\pi)^{n}} \intrrn \dqdpn \sedici \tr\big(\rho\sei (\eip\sei\eiq\tre\fista\sei\eiqa\sei\eipa)\big)
\equx
\label{defphspst}
\intrrn \de\prom\qpt \nove (\eipt\sei\eiqt\tre\sigma\sei\eiqta\sei\eipta) \fin .
\end{eqnarray}
We will call this product the \emph{phase-space stochastic product}. Although this is not immediately
clear from the previous expression, the phase-space stochastic product is commutative, since it stems from
a representation of an abelian group. It follows, by Proposition~{\ref{propbmms}}, that this stochastic product
does not admit collapsing points, because the Hilbert space $\ldrn$ is infinite-dimensional.

In particular, for $\prom=\delta$, we have:
\begin{equation}
\tau=\rho\stof\sigma =
\frac{1}{(2\pi)^{n}} \intrrn \dqdpn \sedici \tr\big(\rho\sei (\eip\sei\eiq\tre\fista\sei\eiqa\sei\eipa)\big)
\sei(\eip\sei\eiq\tre\sigma\sei\eiqa\sei\eipa) \fin .
\end{equation}
We will call this stochastic product the \emph{quantum convolution}. In order to justify this term,
it is worth expressing this binary operation in terms of the Wigner (quasi-probability) distributions
--- see~\cite{AnielloPF,AnielloFPT,Aniello-new,AnielloDMF,Aniello-FT,Aniello-SP,Schleich,AnielloCS} and
references therein --- $\wrho$, $\wfista$, $\wsigma$ and $\wtau$ associated with the density operators $\rho$, $\fista$,
$\sigma$ and $\tau$, respectively. It can be shown~\cite{Aniello-IP} that, setting $\hwfista(x)\defi\wfista(-x)$, $x\in\erredn$,
\begin{equation} \label{formconvs}
\wtau(z)=\intrdn\dx\sei \bbile\intrdn\dy\dieci\wrho(y)\otto\hwfista(x-y)\bbiri\tre \wsigma(z-x) \fin.
\end{equation}
Therefore, we actually find a \emph{double} convolution of Wigner functions, where the function $\wfista$ plays a sort of
pivotal role and has no analogue in the classical setting, where the convolution of two probability distributions on phase
space is a probability distribution too. We will further comment on this crucial point in the next section.

The function
\begin{equation} \label{convowigs}
\qp\mapsto \intrrn \dqdpt \sedici\wrho\qpt\otto\hwfista\qpmt = \frac{1}{(2\pi)^{n}}\otto
\tr\big(\rho\sei (\eip\sei\eiq\tre\fista\sei\eiqa\sei\eipa)\big)
\end{equation}
is a probability distribution wrt the Lebesgue measure on $\phasp$. E.g., setting $n=1$, and choosing
the pure state $|\psi\rangle\langle\psi|$ --- $\psi(q)=(2\pi)^{-1/4}\sei\eee^{-q^2/4}$ ---
as the fiducial state $\fista$, whose Wigner function is
\begin{equation}
\wfista\qp\equiv\wpsi\qp=\frac{1}{\pi}\dieci\eee^{-(q^2+p^2)},
\end{equation}
we obtain the probability distribution
\begin{equation}
\hkrho\qp=\frac{1}{\pi}\intrd\dqdptu \sedici\wrho\qpt\otto\eee^{-(q-\tiq)^2-(p-\tip)^2} \fin ,
\end{equation}
the so-called \emph{Husimi-Kano function}~\cite{Schleich,AnielloCS} associated with the state $\rho$.
Considering now a generic pure state $|\psi\rangle\langle\psi|$, one obtains a class of probability distributions
associated with a quantum state that admit a remarkable operational interpretation in the context of
quantum optics~\cite{Wodkiewicz,Wunsche,Lee}.

Hence, in this case, the quantum
convolution, expressed in terms of phase-space functions, is of the form
\begin{equation}
\wtau\qp=\intrd\dqdptu\sedici\hkrho\qp\otto\wsigma\qpmt
= \intrd\dqdptu\sedici\hksigma\qp\otto\wrho\qpmt \fin ,
\end{equation}
where we have used the commutativity of the product (or the associativity and commutativity of convolution in
relation~{(\ref{formconvs})}). Note that here we have the convolution of an integrable function, the Husimi-Kano
distribution, with a square integrable function, which gives rise to a square integrable function as it should
(all Wigner functions are square integrable~\cite{Aniello-new,AnielloDMF,Aniello-FT,Aniello-SP}); see, e.g.,
Proposition~{2.40} of~\cite{Folland-AA}.


\section{Final remarks, conclusions and perspectives}
\label{conclusions}


In this paper, we have introduced the notion of \emph{stochastic product} defined as a binary operation on the convex
set $\sta$ of quantum states --- the density operators --- that preserves the convex structure. Such a product is automatically
(jointly) continuous wrt the natural product topology on $\sta\times\sta$ (Remark~{\ref{topologies}}, Corollary~{\ref{autocont}}),
and can be extended to an algebra on the Banach space $\trc$ of trace class operators, the so-called
\emph{canonical extension} of the stochastic product (Proposition~{\ref{estende}}). We have
also defined a \emph{stochastic algebra} as an algebra on $\trc$ arising from an \emph{associative} stochastic product;
by restricting to a binary operation on the real Banach space $\trcsa$ of \emph{selfadjoint} trace class operators,
one obtains a real \emph{Banach algebra} (Remark~{\ref{banalge}}). We have then shown --- see sect.~{\ref{main}} ---
that one can explicitly construct a class of stochastic algebras, the so-called \emph{twirled stochastic algebras},
by means of a square integrable projective representation $\ure$ of a (unimodular) locally compact group $G$,
and by choosing a \emph{fiducial state} $\fista$ and a Borel \emph{probability measure} $\prom$ on $G$; in the simplest case,
we set $\prom=\delta$ (the Dirac measure at the identity). This is actually a special case of a more general class
of algebras --- the \emph{twirled algebras} (tout court) --- associated with a pair $(\ure,\vre)$ of projective
group representations (where $\ure$, as above, is supposed to be square integrable), with a fiducial trace class operator
$\fiv$ and with a complex Borel measure $\nu$. It turns out that the twirled stochastic algebras are, as well, complex Banach
algebras (i.e., restricting to an algebra on $\trcsa$ is not necessary, in this case). Moreover, these stochastic algebras
enjoy nice covariance, equivariance and invariance properties, and, in the case where the relevant group $G$ is abelian,
our construction yields a \emph{commutative} algebra (Theorem~{\ref{theosto}}). One then easily shows, by means of examples,
the existence of twirled stochastic algebras for every Hilbert space dimension (sect.~{\ref{examples}}).

Not surprisingly, our group-theoretical construction of a stochastic product involves some important tools typical of various
applications of abstract harmonic analysis to the foundations of quantum theory (quantization, phase-space quantum mechanics,
quantum measurement theory).

In this regard, let us first recall that, using our previous notations, the mapping
\begin{equation} \label{copovm}
G\supset\bor\mapsto\int_{\bor}\de\hame(g)\nove (\urep(g)\tre\fista)\ifed \povm(\bor) \fin , \ \ \  \mbox{($\ure$ square integrable)}
\end{equation}
where $\bor$ is a Borel set and $\fista$ a density operator, is a \emph{covariant quantum observable}; i.e., a POVM covariant
wrt the irreducible representation $\ure$~\cite{Cassinelli-POVM}: $\povm(g\tre\bor)=\urep(g)\sei\povm(\bor)$, for every Borel
subset $\bor$ of $G$ and every $g\in G$. Actually, formula~{(\ref{copovm})} provides the general expression of such a POVM;
namely, if a POVM is covariant wrt to the irreducible projective representation $\ure\colon G \rightarrow\unih$, then this
representation must be square integrable and the POVM must be of the form~{(\ref{copovm})}, for exactly one state $\fista\in\sta$.
Hence, the function $g\mapsto\tr\big(\rho\sei(\urep(g)\tre\fista)\big)$ used in our construction of a stochastic product can be
regarded as the probability density on $G$ --- regarded as a sample space --- of the covariant observable $\povm$, relative to
the observed state $\rho$; namely, the Radon-Nikodym derivative of the probability measure
$\mesru\colon\bor\mapsto\int_{\bor}\de\hame(g)\nove\tr\big(\rho\sei (\urep(g)\tre\fista)\big)$ wrt the (suitably normalized)
Haar measure $\hame$. This function can also be regarded as the restriction, to the diagonal subgroup of the direct product
group $G\times G$, of the \emph{frame transform}~\cite{Aniello-FT} of $\rho$, generated by the square integrable representation
$\ure$, with \emph{analyzing operator} $\fista$.

Moreover, for every Borel probability measure $\mu$ on $G$, the mapping
\begin{equation}
\trc\ni A\mapsto\intG\de\mu(g)\nove(\urep(g)\tre A)\ifed \twirl\sei A \in \trc \fin ,
\ \ \ \twirl\tre (\sta) \subset\sta \fin ,
\end{equation}
where this time $\ure\colon G \rightarrow\unih$ is a generic projective representation, defines a quantum dynamical map
(or quantum channel), a so-called \emph{twirling operator}~\cite{AnielloBM}. Therefore, if $\ure$ is square integrable,
then the stochastic twirled product can be thought of as (recall~{(\ref{expcon})})
\begin{equation}
\rho\stofm\sigma = \bile\lmap(\rho)\biri(\sigma) = \big(\twirlrum\big)\sei\sigma  \fin ,
\end{equation}
where $\lmap\colon\trc\rightarrow\ltrc$ is the left partial map associated with this product (Corollary~{\ref{lrmaps}}).

As a further observation within the same circle of ideas, it is worth mentioning that, for every square integrable
representation $\ure\colon G \rightarrow\unih$ ($G$ unimodular) and every pair of states $\fista,\sigma\in\sta$,
the mapping
\begin{equation} \label{covinstr}
G\supset\bor\mapsto\bbile
\trc\ni A\mapsto\int_{\bor}\de\hame(g)\nove \tr\big(A\sei(\urep(g)\tre\fista)\big)
\sei(\urep(g)\tre\sigma)\ifed\instru A\in\trc\bbiri
\end{equation}
--- where $\bor$ is a Borel set and $\instru$ a quantum operation~\cite{Heino} (in particular, a quantum channel for $\bor=G$)
--- is a \emph{quantum instrument}~\cite{Heino,Davies-Lewis,Ozawa-QMP,Carmeli}; more specifically, a \emph{$\ure$-covariant}
quantum instrument based on $G$~\cite{Carmeli}, because
\begin{equation} \label{covaria}
\ginstru\tre (\urep(g)A)=\urep(g)\tre\big(\instru A\big) .
\end{equation}
(The reader may notice, however, that the covariant quantum instrument $\bor\mapsto\instru$, as defined in~{(\ref{covinstr})},
is not expressed in the form displayed in Corollary~{12} of~\cite{Carmeli}; in order to properly compare the two
expressions of the quantum instrument, one has to suitably re-elaborate formula~{(\ref{covinstr})}.)
Clearly, a connection with the twirled stochastic product associated with the triple $(\ure,\fista,\delta)$, and with the
related twirling operator, is obtained by setting $\bor=G$; i.e.,
\begin{equation}
\rho\stof\sigma = \bile\lmap(\rho)\biri(\sigma) = \big(\twirlru\big)\sei\sigma = \bile\rmap(\sigma)\biri(\rho) =
\instrug\rho \fin .
\end{equation}
Therefore, the linear map from $\trc$ into $\ltrc$ that extends the convex-linear application $\sta\ni\sigma\mapsto\instrug\in\pot$
coincides with the right partial map $\rmap$ associated with the given stochastic product. Besides, since $g\tre G=G$, from
relation~{(\ref{covaria})} we recover the left-covariance of the twirled stochastic product. Also note that the associativity of
this product translates into the following relation for the quantum channels $\instrur,\instrug\in\pot$:
\begin{equation}
\instrug\circ\instrur = \instrut \fin , \ \ \ \mbox{where $\tau=\rho\stof\sigma$} \fin .
\end{equation}

Observe, moreover, that the condition
\begin{equation}
\tr(\rho\sei\pov(\bor))=\tr(\instru\rho) \fin ,
\end{equation}
for every $\rho\in\sta$ and every Borel subset $\bor$ of $G$, uniquely determines a ($\ure$-covariant) POVM; see subsect.~{5.1.3} of~\cite{Heino}.
Precisely, recalling the expression~{(\ref{copovm})}, one finds out that $\pov=\povm$. We stress that the class of covariant quantum
instruments of the form~{(\ref{covinstr})} --- with $\ure$ and $\sigma$ (an arbitrary state) kept fixed, and $\fista$ ranging over $\sta$
--- gives rise to the \emph{whole} class of POVMs that are covariant wrt the representation $\ure$. Precisely, for every $\sigma\in\sta$,
the mapping $\povm\mapsto\instrar$, from the set of all $\ure$-covariant POVMs to the set of all $\ure$-covariant quantum instruments, can be regarded
as a \emph{cross section} wrt to the partition of the latter set into equivalence classes of instruments \emph{compatible} with the same observable.

The twirled stochastic product admits a remarkable expression in terms of the \emph{covariant symbols}
associated with the density operators~\cite{Aniello-IP}. Given a square integrable projective representation
$\ure\colon G \rightarrow\unih$ ($G$ unimodular) and a trace class operator $A\in\trc$, the (covariant) symbol
$\chara$ of $A$ is the complex function on $G$ defined by
\begin{equation} \label{defiwt}
\chara(g)\defi\ddumo\,\tr(\ure(g)^\ast A) \fin ,
\end{equation}
where $\ddu\equiv\sqrcu$. Moreover, $\chara=\dequm\tre A$, where $\dequm\colon\hs\rightarrow\ldg\equiv\elleg$
denotes a linear isometry mapping the Hilbert space $\hs$ of Hilbert-Schmidt operators into the Hilbert
space of square integrable functions $\ldg$. This isometry can be regarded as a \emph{dequantization map}, which is directly related to
the Wigner transform in the case where $G$ is the group of translations on phase space~\cite{Aniello-FT,Aniello-SP,Aniello-new, AnielloDMF}.
The operator $A$ an be explicitly re-constructed from its symbol via the quantization map $\qum=\dequm^\ast$~\cite{Aniello-SP}.

In the case where $G$ is the group of phase-space translations and $\ure$ is the Weyl system (see subsect.~{\ref{weysys}}),
for $A\equiv\rho\in\sta$ the symbol $\chara\equiv\charr$ is also called the \emph{quantum characteristic function} of the state $\rho$
--- essentially, the Fourier transform of the Wigner function $\wrho$ --- in analogy with the classical characteristic function
of a probability measure on a locally compact abelian group~\cite{AnielloPF,AnielloFPT,Aniello-new,AnielloDMF}. Recall, indeed, that,
for every $\prom\in\prome$, $G\equiv\phasp$, we can identify the characteristic function of $\prom$ --- i.e., its Fourier-Stieltjes transform
$\hprom\colon\dug\rightarrow\ccc$, where $\dug$ is the Pontryagin dual~\cite{Folland-AA} of $G$ --- with the function $\charpro\colon\phasp\rightarrow\ccc$
defined by
\begin{equation}
\charpro\qp\defi\intrrn\de\prom\qpt\nove\exp(\ima(q\cdot\tip - p\cdot\tiq)) \fin .
\end{equation}
It turns out~\cite{Aniello-IP} that, setting $\ldg=\lrrnco$ (hence, $\ddu=1$), the phase-space stochastic product
associated with the triple $(\ure,\fista,\prom)$ --- expressed in terms of the characteristic function $\charpro$
and of the symbols $\charr\defi\tr(\ure\qp^\ast\rho)$, $\charu$, $\chars$ of the states $\rho$, $\fista$, $\sigma$
--- assumes the simple form
\begin{equation} \label{phspstpr}
\Bile \rho\stofm\sigma \Birich \qp = \charpro\qp\sei\charr\qp\sei\charu\mqp\sei\chars\qp
= \charpro\qp\sei\charr\qp\sei\overline{\charu\qp}\sei\chars\qp \fin ,
\end{equation}
from which it is evident that this product is commutative.

Considering formula~{(\ref{phspstpr})} for the phase-space stochastic product, one may regard this nice
expression as a straightforward way to achieve a commutative stochastic product. Indeed, whereas
the pointwise product of two quantum characteristic functions is not, in general, a function of the same kind,
the product of a `classical' characteristic function on phase space --- the Fourier-Stieltjes transform of a
probability measure on $\phasp$ --- by a quantum characteristic function, is a function of the latter type~\cite{AnielloPF,AnielloFPT}.
Moreover, for every $\rho,\fista\in\sta$, $\charr\sei\overline{\charu}$ is a classical characteristic function, the (symplectic)
Fourier-Stieltjes transform of the probability measure $\mesru$,
$\de\mesru\qp=(2\pi)^{-n}\tre\tr\big(\rho\sei(\urep\qp\tre\fista)\big)\sei\dqdpn\fin$; namely, recalling~{(\ref{convowigs})}
and denoting by $\wrho\conv\hwfista$ the convolution of the functions $\wrho$ and $\hwfista$ ($\hwfista\qp\defi\wfista\mqp$),
\begin{eqnarray}
\bile\charr\sei\overline{\charu}\quattro\biri\qp
\equa
\frac{1}{(2\pi)^{n}}\intrrn\dqdpt\sedici\tr\big(\rho\sei(\teip\sei\teiq\tre\fista\sei\teiqa\sei\teipa)\big)
\sei\exp(\ima(q\cdot\tip - p\cdot\tiq))
\equb
\intrrn \dqdpt \sedici\bile\wrho\conv\hwfista\biri\qpt\sei\exp(\ima(q\cdot\tip - p\cdot\tiq)) \fin .
\end{eqnarray}

Therefore, the function $\charpro\sei\charr\sei\overline{\charu}\sei\chars$ can be regarded as the pointwise product
of two \emph{classical} characteristic functions  --- $\charpro$ and $\charr\sei\overline{\charu}$ --- which is again
a function of the same kind (the Fourier-Stieltjes transform of the convolution of two probability measures),
multiplied by the \emph{quantum} characteristic function $\chars$, so achieving a function of the latter type.
Interestingly, the fact that $\charr\sei\overline{\charu}$ is a classical characteristic function can also
be proved `intrinsically' using the properties of classical and quantum \emph{positive definite} functions~\cite{AnielloPF}.

Analyzing the previous example, various intriguing links connecting the notion of quantum convolution
with Werner's seminal work on quantum harmonic analysis on phase space~\cite{Werner-QHA} come to light.
In Werner's remarkable approach, beside the ordinary convolution $f_1\conv f_2$ of two integrable
functions $f_1$ and $f_2$ on phase space, one can also construct the convolution $f\conv A$ of a function
$f\in\lurrn$ with an operator $A\in\trc$ (and \emph{vice versa}) --- where $\hh=\ldrn$ --- and the convolution
$A\conv B$ of two operators $A,B\in\trc$. Precisely,
\begin{equation} \label{convfuop}
f\conv A = A\conv f\defi (2\pi)^{-n}\intrrn\dqdpn\sedici f\qp\sei (\urep\qp A)\in\trc \fin ,
\end{equation}
where $\ure$ is the Weyl system --- i.e., $\ure\qp=\eee^{-\ima\tre q\cdot p/2}\sei\eip\sei\eiq$ --- and,
denoting by $\parop$ the parity operator in $\ldrn$ ($\parop=\paropa=\paropmo$, $\ure\qp\tre\parop=\parop\ure\mqp=\ure\qp^\ast$),
\begin{equation}
A\conv B \defi \tr\big(A\sei(\urep\qp\tre(\parop B\parop))\big)= B\conv A \in\lurrn \fin .
\end{equation}
It turns out that associativity is satisfied; in particular, $(A\conv B)\conv C = A\conv (B\conv C)$. Therefore,
in this language, the quantum convolution (the phase-space stochastic product with $\prom=\delta$) can be written,
unambiguously, in the following form:
\begin{equation}
\rho\stof\sigma = \rho\conv(\parop\fista\parop)\conv\sigma = \sigma\stof\rho \fin .
\end{equation}
We stress that, here, the commutativity of the quantum convolution may seem, at first sight, an immediate consequence
of the fact that $f\conv A = A\conv f$ in definition~{(\ref{convfuop})} and of the easily verified relation
$A\conv B=B\conv A$; but it actually also involves the non-trivial associativity of Werner's convolution of operators.

Beside the aforementioned formulation of twirled products in terms of covariant symbols, there are several aspects
of stochastic products that are not treated in this paper and that we plan to study; in particular:
\begin{itemize}

\item For the sake of simplicity, we have considered the twirled product associated with a square integrable
representation of a \emph{unimodular} locally compact group. Taking into account some mathematical intricacy
related to an unbounded Duflo-Moore operator (Remark~{\ref{duflom}}), one can work out the non-unimodular case,
as well.

\item There is a natural notion of \emph{complete} positivity for a stochastic product.

\item The classification of covariant stochastic products is an interesting open problem.

\item The behaviour of a quantum entropy~\cite{Ohya,Aniello-QE,Aniello-SchCo} --- not necessarily
the von~Neumann entropy --- wrt to a stochastic product is another interesting issue. In particular,
it is natural to wonder whether the twirled stochastic product is, say, entropy-nondecreasing; i.e.,
whether the entropy of the product of two states is not smaller than the entropy of each of these states.
Observe that such a property would provide a nice physical interpretation of the fact that the maximally mixed
state is both a left and a right collapsing point for the twirled stochastic product associated with an
irreducible unitary representation of a compact group; see subsect.~{\ref{compactgrps}}.

\item An abstract setting for a stochastic algebra may give rise to useful tools and new insights.

\end{itemize}

In connection with the last point, we suggest that the following algebraic structure might be
appropriate for our purposes. Let us consider a mathematical object, consisting of two sets and
four maps, of the form
\begin{equation}
\big(\alge \fin , \cinque \oprod\colon\alge\times\alge\rightarrow\alge \fin , \cinque \involu\colon\alge\rightarrow\alge \fin ;
\cinque \stalg = \alge\opro\alge\subset\alge \fin , \cinque \bifo\colon\stalg\times\stalg\rightarrow\stalg, \cinque
\trace\colon\stalg\rightarrow\ccc\big) ,
\end{equation}
where:
\begin{enumerate}
\item $\alge$ is a proper \emph{$\Hstar$-algebra}~\cite{Aniello-SP,Ambrose,Rickart}, equipped with the product $\oprod$ and
the \emph{involution} $\alge\ni A\mapsto \involu(A)\equiv A^\ast\in\alge$ (let us denote by $\scap$ the scalar product,
assumed to be linear in its second argument, in the complex Hilbert space $\alge$).

\item $\stalg=\alge\opro\alge$ is the \emph{trace class}~\cite{Saworotnow} of the proper $\Hstar$-algebra $\alge$, and an
associative algebra wrt the binary operation $\bifo$.

\item the bilinear map $\bifo$ preserves the convex cone
\begin{equation}
\stalgp\defi\{T\in\stalg\colon \langle A, TA\rangle\ge 0, \ \forall\cinque A\in\alge\}=\{A^\ast\mtre\opro A \colon \ A\in\alge\}
\end{equation}
of all positive elements of $\stalg$; i.e., $\stalgp\pro\stalgp\subset\stalgp$.

\item $\trace\colon\stalg\rightarrow\ccc$ is the the canonical (positive) \emph{trace} functional of the trace class $\stalg$;
therefore, $\trace(A\opro B)=\langle A^\ast,B\rangle$, for all $A,B\in\alge$.

\item  The trace is multiplicative wrt binary operation $\bifo$: $\trace(A\pro B) = \trace(A)\sei\trace(B) = \trace(B\pro A)$.
\end{enumerate}

We propose to call such an object a (abstract) \emph{stochastic $\Hstar$-algebra}. A stochastic algebra on $\trc$, as defined
in subsect.~{\ref{stochal}}, fits in this abstract scheme with the following identifications: the $\Hstar$-algebra $\alge$
is the Hilbert-Schmidt space $\hs$, endowed with the standard product of operators (composition) and with the natural involution
$A\mapsto A^\ast$ (adjoining map); the associated trace class $\stalg$ is, of course, the space $\trc$, endowed with the trace
functional $\tr\argo$ and with an associative, state-preserving bilinear map $\bifo\colon\trc\times\trc\rightarrow\trc$.


\end{document}